\documentclass[a4paper, 11pt]{article}

\usepackage{amsmath, amsfonts, fullpage, hyperref}
\usepackage{hyperref}
\usepackage{amsthm}
\usepackage{braket}
\usepackage[inline]{enumitem}
\usepackage{color}
\usepackage{graphicx}

\newcommand{\realR}{\mathbb{R}}
\newcommand{\compC}{\mathbb{C}}
\newcommand{\intZ}{\mathbb{Z}}
\newcommand{\Prob}{\mathbb{P}}
\newcommand{\Proj}{\mathbf{P}}
\newcommand{\Andreief}{Andr\'{e}if}
\newcommand{\K}{\mathbf{K}}
\newcommand{\M}{\mathbf{M}}
\newcommand{\R}{\mathbf{R}}
\newcommand{\x}{\mathbf{x}}
\newcommand{\Z}{\mathbb{Z}}
\newcommand{\A}{\mathbb{A}}

\newcommand{\E}{\mathbb{E}}
\newcommand{\MM}{\mathbf{MM}}
\newcommand{\bigO}{\mathcal{O}}
\newcommand{\dettwo}{\det\nolimits_{2}}
\newcommand{\qbinom}[3][q]{{\genfrac{[}{]}{0pt}{}{#2}{#3}}_{#1}}

\newcommand{\ep}{\epsilon}

\DeclareMathOperator{\Ai}{Ai}
\DeclareMathOperator{\Li}{Li}
\DeclareMathOperator{\Airy}{Airy}
\DeclareMathOperator{\sub}{sub}
\DeclareMathOperator{\inter}{mid}
\DeclareMathOperator{\res}{res}
\DeclareMathOperator{\arccosh}{arccosh}
\DeclareMathOperator{\crossover}{cross}
\DeclareMathOperator{\GUE}{GUE}
\DeclareMathOperator{\sgn}{sgn}
\DeclareMathOperator{\Tr}{Tr}
\DeclareMathOperator{\U}{U}
\DeclareMathOperator{\interpolating}{inter}

\newtheorem{prop}{Proposition}
\newtheorem{thm}{Theorem}
\newtheorem{lem}{Lemma}

\theoremstyle{remark}
\newtheorem{rmk}{Remark}

\title{Asymptotics of free fermions in a quadratic well at finite temperature and the Moshe--Neuberger--Shapiro random matrix model}

\author{Karl Liechty \thanks{Department of Mathematical Sciences, DePaul University, Chicago, IL, 60614 USA \href{mailto:kliechty@depaul.edu}{\nolinkurl{kliechty@depaul.edu}}. Supported by a Simons Foundation Collaboration Grant \#357872.} \and Dong Wang \thanks{Department of Mathematics, National University of Singapore, Singapore, 119076, \href{mailto:matwd@nus.edu.sg}{\nolinkurl{matwd@nus.edu.sg}}. Supported by the Singapore AcRF Tier 1 grant R-146-000-217-112.}}

\begin{document}

\maketitle

\begin{abstract}
  We derive the local statistics of the canonical ensemble of free fermions in a quadratic potential well at finite temperature, as the particle number approaches infinity. This free fermion model is equivalent to a random matrix model proposed by Moshe, Neuberger and Shapiro. Limiting behaviors obtained before for the grand canonical ensemble are observed in the canonical ensemble: We have at the edge the phase transition from the Tracy--Widom distribution to the Gumbel distribution via the Kardar--Parisi--Zhang (KPZ) crossover distribution, and in the bulk the phase transition from the sine point process to the Poisson point process. A similarity between this model and a class of models in the KPZ universality class is explained. We also derive the multi-time correlation functions and the multi-time gap probability formulas for the free fermions along the imaginary time.
\end{abstract}

\section{Introduction}

In this paper we consider the spinless free fermions on $\realR^1$ in quadratic potential well (aka harmonic oscillators) at finite temperature. This model was defined by Moshe, Neuberger and Shapiro \cite{Moshe-Neuberger-Shapiro94} in the 1990's, further studied by Johansson \cite{Johansson07} in the 2000's, and very recently considered in the physics literature by Dean, Le Doussal, Majumdar, Schehr et al \cite{Dean-Le_Doussal-Majumdar-Schehr15}, \cite{Dean-Le_Doussal-Majumdar-Schehr16}, \cite{Le_Doussal-Majumdar-Rosso-Schehr16}. See also \cite{Le_Doussal-Majumdar-Schehr17} for a dynamical version of the model, and \cite{Cunden-Mezzadri-OConnell17} for a generalization to other symmetry types.

The most interesting question on this model (later called the MNS model) is the limiting behavior of the fermions at the edge or in the bulk as the number of particles $n \to \infty$. From the physical point of view, the existing result is already rather complete. When the temperature is low enough, the limiting distribution of the rightmost particle is given by the celebrated Tracy--Widom distribution, and when the temperature is high enough, the limiting distribution is given by the Gumbel distribution. At the critical temperature, the limiting distribution is found to be the crossover distribution in the $1$-dimensional Kardar--Parisi--Zhang (KPZ) universality class. For particles in the bulk, analogous results are obtained which interpolate between the sine point process and the Poisson point process.

The original version proposed by Moshe, Neuberger and Shapiro is the \emph{canonical ensemble} of the model, but all the asymptotic results available currently in the mathematical literature are for the \emph{grand canonical ensemble} of the model.  It is a universally accepted wisdom in statistical physics that the physical properties of the grand canonical ensemble are the same as those of the canonical ensemble as the particle number approaches infinity. In the case of the MNS model, the grand canonical ensemble has a special mathematical feature that it is a \emph{determinantal point process}, which makes it easier to analyze mathematically than the canonical ensemble. Currently all results on the MNS model in the mathematics literature deal with the grand canonical ensemble, although several recent works in the physical literature \cite{Dean-Le_Doussal-Majumdar-Schehr15}, \cite{Dean-Le_Doussal-Majumdar-Schehr16}, \cite{Le_Doussal-Majumdar-Rosso-Schehr16} have considered the canonical ensemble.  The goal of this paper is to analyze the canonical ensemble of the MNS model directly, and rigorously prove that the limiting results obtained for the grand canonical ensemble hold for the canonical ensemble as well.

Our purpose is not rigor for rigor's sake. As suggested by the title, the canonical ensemble of the MNS model is associated to a random matrix model (later referred as the MNS random matrix model) whose dimension is equal to the number of particles in the MNS model. Such a relation is not preserved when we move to the grand canonical ensemble. Also in the course of our derivation, we find that the algebraic as well as the analytic properties of the canonical ensemble of the MNS model are analogous to those of the Asymmetric Simple Exclusion Process (ASEP) and the $q$-Whittaker processes, which are a subclass of the extensively studied Macdonald processes \cite{Borodin-Corwin13}. The $q$-Whittaker processes contain many interacting particle models in the KPZ universality class as specializations. Although the ASEP and the $q$-Whittaker processes are in some sense integrable, they are considerably more difficult than determinantal processes. The similarity between probability models in KPZ universality class and free fermions at positive temperature has been noticed in \cite{Imamura-Sasamoto15}, but the relation is via determinantal process. We hope that our analysis of the canonical ensemble of the MNS model sheds light on the study of the integrable particle models in the KPZ universality class.

\subsection{$q$-analogue Notation}

Throughout this paper, we use the following $q$-analogue notations, which converge to their common counterparts as $q \to 1_-$. 

The $q$-Pochhammer symbol is
\begin{equation}
  (a; q)_n = \prod^{n - 1}_{k = 0} (1 - aq^k), \quad n = 0, 1, 2, \dotsc, \infty.
\end{equation}
The $q$-binomial is
\begin{equation}
  \qbinom{n}{m} = \frac{(1 - q^n)(1 - q^{n - 1}) \dotsb (1 - q^{n - m + 1})}{(1 - q^m)(1 - q^{m - 1}) \dotsb (1 - q)}, \quad 0 \leq m \leq n.
\end{equation}

\subsection{Definition of the MNS model}

First recall the one-dimension harmonic oscillator in quantum mechanics. The time-independent Hamiltonian of the free particle in a quadratic potential well is, on the position space,
\begin{equation}
  H = -\frac{\hbar^2}{2m} \frac{\partial^2}{\partial x^2} + \frac{m \omega^2}{2} x^2.
\end{equation}
In this paper, we assume $\hbar = 1$, $m = 1/2$, and $\omega = 1$, and then
\begin{equation} \label{eq:Hamiltonian_normalized}
  H = -\frac{\partial^2}{\partial x^2} + \frac{x^2}{4}.
\end{equation}
The eigenfunctions of the Hamiltonian $H$ defined in \eqref{eq:Hamiltonian_normalized} are
\begin{equation}
  \varphi_k(x) = \left( \frac{1}{\sqrt{2\pi} k!} \right)^{1/2} H_k(x) e^{-x^2/4}, \quad k = 0, 1, 2, \dotsc,
\end{equation}
where $H_k(x)$ is the Hermite polynomial, defined to be the monic polynomial of degree $k$ satisfying the orthogonality
\begin{equation}
\int^{\infty}_{-\infty} H_k(x) H_j(x) e^{-x^2/2} dx = \sqrt{2\pi} k! \delta_{kj}.
\end{equation}
The functions $\{\varphi_k(x)\}_{k=0}^\infty$ form an orthonormal basis for $L^2(\realR)$.
See \cite[Chapter 22]{Abramowitz-Stegun64} for basic properties of Hermite polynomials. Note that in \cite{Abramowitz-Stegun64}, polynomial $H_n(x)$ is denoted as $He_n(x)$, while the notation $H_n(x)$ is reserved for a slightly different polynomial, see \cite[22.5.18]{Abramowitz-Stegun64}. The eigenvalue/energy level for eigenstate $\varphi_k(x)$ is $k + 1/2$, ($k = 0, 1, 2, \dotsc$,) since
\begin{equation}
  H \varphi_k(x) = \left( -\frac{d^2}{dx^2} + \frac{x^2}{4} \right) \varphi_k(x) = \left( k + \frac{1}{2} \right) \varphi_k(x).
\end{equation}

Suppose $n$ identical fermions are independent harmonic oscillators, or in other words they are free fermions in the quadratic potential well. The fermionic system has eigenstates indexed by $(k_1, k_2, \dotsc, k_n)$ where $0 \leq k_1 < k_2 < \dotsb < k_n$ are integers, and the energy level of the eigenstate is $k_1 + k_2 + \dotsb + k_n + n/2$, The corresponding eigenfunction is given by the Slater determinant
\begin{equation} \label{eq:defn_Phi_k}
  \Phi_{k_1, \dotsc, k_n}(x_1, \dotsc, x_n) = \frac{1}{\sqrt{n!}}
  \begin{vmatrix}
    \varphi_{k_1}(x_1) & \dots & \varphi_{k_1}(x_n) \\
    \vdots & & \vdots \\
    \varphi_{k_n}(x_1) & \dots & \varphi_{k_n}(x_n) \\
  \end{vmatrix}.
\end{equation}
In this eigenstate, the density function for the $n$ particles is $\lvert \Phi_{k_1, \dotsc, k_n}(x_1, \dotsc, x_n) \rvert^2$. %(This density function is normalized, see \eqref{eq:normalization_check}.)

For a quantum system at temperature $T$, all eigenstates occur at a certain chance according to the \emph{Boltzmann distribution}, so that the probability for an eigenstate with energy level $E$ to occur is $Z^{-1} e^{-E/(\kappa T)}$ where $Z$ is the normalization constant and $\kappa$ is the Boltzmann constant \cite[Section 6.2]{Beale-Pathria11}, which we assume to be $1$ later. Hence for the $n$-particle canonical ensemble of the MNS model, that is, $n$ free fermions in the quadratic potential well, if the temperature is $T > 0$, and if we denote
\begin{equation}
  q = e^{-1/(\kappa T)} = e^{-1/T},
\end{equation}
the probability for eigenstate $(k_1, k_2, \dotsc, k_n)$ to occur is $Z_n(q)^{-1} q^{k_1 + \dotsb + k_n + n/2}$, where
\begin{equation} \label{eq:explicit_Z_n(q)}
  Z_n(q) = \sum_{0 \leq k_1 < k_2 < \dotsb < k_n} q^{k_1 + \dotsb + k_n + n/2} = \frac{q^{n^2/2}}{(q; q)_n}.
\end{equation}
We then have that the density function for the $n$ particles is
\begin{equation} \label{eq:density_n_particle}
  \begin{split}
    P_n(x_1, \dotsc, x_n) = {}& \frac{1}{Z_n(q)} \sum_{0 \leq k_1 < k_2 < \dotsb < k_n} \lvert \Phi_{k_1, \dotsc, k_n}(x_1, \dotsc, x_n) \rvert^2 q^{k_1 + \dotsb + k_n + n/2} \\
    = {}& \frac{q^{n/2}}{Z_n(q)} \sum_{0 \leq k_1 < k_2 < \dotsb < k_n} \lvert \Phi_{k_1, \dotsc, k_n}(x_1, \dotsc, x_n) \rvert^2 q^{k_1 + \dotsb + k_n}.
  \end{split}
\end{equation}
The equivalence of the two expressions in \eqref{eq:explicit_Z_n(q)} may not be obvious, but it is easily proven by induction on $n$.
%By the inductive formula
%\begin{equation} \label{eq:inductive_Z_n}
%  \begin{split}
%    \frac{Z_n(q)}{q^{n/2}} = {}& \sum_{0 \leq k_1 < k_2 < \dotsb < k_n} q^{k_1 + \dotsb + k_n} \\
%    = {}& \sum^{\infty}_{k_1 = 0} q^{k_1} q^{(k_1 + 1)(n - 1)} \sum_{0 \leq l_1 < l_2 < \dotsb < l_{n - 1}} q^{l_1 + \dotsb + l_{n - 1}} \\
%    = {}& \sum^{\infty}_{k_1 = 0} q^{(k_1 + 1)n - 1} \frac{Z_{n - 1}(q)}{q^{(n - 1)/2}},
%  \end{split}
%\end{equation}
%where $l_j = k_{j + 1} - k_1 - 1$ for $j = 1, 2, \dotsc, n - 1$, and the result
%\begin{equation}\label{eq:inductive_basecase_Z_n}
%  Z_1(q) = q^{1/2} \sum^{\infty}_{k_1 = 0} q^{k_1} = \frac{q^{1/2}}{1 - q},
%\end{equation}
%we derive the second identity in \eqref{eq:explicit_Z_n(q)}.
% \begin{equation}\label{eq:part_funct_formula}
%   Z_n(q) = \frac{q^{n^2/2}}{\prod^n_{k = 1} (1 - q^k)} = \frac{q^{n^2/2}}{(q; q)_n}.
% \end{equation}

The $n$-particle canonical ensemble of the MNS model at temperature $T = -(\log q)^{-1} > 0$, which is called simply the MNS model if there is no possibility of confusion, is the main topic of this paper. Although it is defined in the language of quantum mechanics, all our analysis is based on the density function \eqref{eq:density_n_particle}, so it is harmless to understand the MNS model as a particle model with density \eqref{eq:density_n_particle}. We note that in the limit $T \to 0$, the density function $P_n(x_1, \dotsc, x_n)$ degenerates into $\lvert \Phi_{0, 1, \dotsc, n - 1}(x_1, \dotsc, x_n) \rvert^2$, the density function for the ground state of the quantum system. One readily  recognizes that this $T \to 0$ limiting density is the density of eigenvalues of a random matrix in the Gaussian Unitary Ensemble (GUE) \cite[Section 2.5]{Anderson-Guionnet-Zeitouni10}, that is, the random Hermitian matrix model defined below in \eqref{eq:GUE_density}. It is then not a surprise that for general $T > 0$, density \eqref{eq:density_n_particle} is also the eigenvalue density function of a random matrix ensemble.

\subsection{MNS random matrix model}

The random matrix model defined by Moshe, Neuberger and Shapiro \cite{Moshe-Neuberger-Shapiro94} is an unitarily invariant generalization of the GUE with a continuous parameter. As the parameter varies, the limiting local statistics of the MNS random matrix model interpolate between the sine point process, which is the hallmark of random Hermitian matrices including the GUE, and the Poisson point process.

The space of $n$-dimensional Hermitian matrices has a natural measure
\begin{equation}
  dX = \prod^n_{i = 1} dx_{ii} \prod_{1 \leq j < k \leq n} d\Re x_{jk} d\Im x_{jk},
\end{equation}
where $X = (x_{jk})^n_{j, k = n}$. Let $U$ be a random unitary matrix in $\U(n)$ with respect to the Haar measure. We say that a random Hermitian matrix $H$ is an MNS random matrix if \cite[Formulas (1) and (2)]{Moshe-Neuberger-Shapiro94}
\begin{equation} \label{eq:MNS_random_matrix}
  \begin{split}
    P(H) dH = {}& \frac{1}{C(n, b)} e^{\Tr H^2} e^{-b \Tr([U, H][U, H]^{\dagger})} dH \\
    = {}& \frac{1}{C(n, b)} e^{-(2b + 1) \Tr H^2} \left[ \int_{\U(n)} dU e^{2b \Tr(UHU^{\dagger}H)} \right] dH.
  \end{split}
\end{equation}
By comparing the eigenvalue distribution of $H$ and the known density function of free fermions in a quadratic potential well at finite temperature, Moshe, Neuberger and Shapiro observe the following relation.
\begin{prop} \cite[Formula (4)]{Moshe-Neuberger-Shapiro94} \label{prop:Moshe-Neuberger-Shapiro94}
  Suppose the $n$-dimensional Hermitian random matrix is defined by \eqref{eq:MNS_random_matrix}, and suppose the parameter $b = q/(1 - q)^2$ with $q \in (0, 1)$. Then the joint probability density function of the eigenvalues of $\sqrt{\frac{1}{2}(1 - q)/(1 + q)} H$ is the same as the density function $P_n(x_1, \dotsc, x_n)$ defined in \eqref{eq:density_n_particle}.
\end{prop}
If we denote the $q \to 0$ limit of $2^{-1/2} H$ by $X$, then $X$ has the density function
\begin{equation} \label{eq:GUE_density}
  P(X) dX = \frac{1}{2^{n/2} \pi^{n^2/2}} \exp \left( -\frac{1}{2} \Tr(X^2) \right) dX,
\end{equation}
or equivalently, $X_{ii} = N(0, 1/2)$, $\Re X_{jk} = N(0, 1)$, $\Im X_{jk} = N(0, 1)$ for $1 \leq i \leq n$ and $1 \leq j < k \leq n$, and they are independent. This is the celebrated GUE ensemble in dimension $n$ \cite[Section 2.5]{Anderson-Guionnet-Zeitouni10}.

The authors of \cite{Moshe-Neuberger-Shapiro94} give half of the proof to Proposition \ref{prop:Moshe-Neuberger-Shapiro94}, and point out that the other half is available in physics literature, see \cite{Boulatov-Kazakov92}. For the sake of our readers, we provide a brief proof of the part omitted in \cite{Moshe-Neuberger-Shapiro94} in Appendix \ref{sec:relation_to_MNS_RM}.

\subsection{Statement of results}

As the particle number $n \to \infty$, we are interested in the limiting distribution of the rightmost particle in the MNS model. The distribution of the position of the rightmost particle,
\begin{equation} \label{eq:rightmost_pt_as_gap_prob}
  \Prob_n(\max(x_1, \dotsc, x_n) \leq s) = \Prob_n(x_1, \dotsc, x_n \in (-\infty, s]),
\end{equation}
is a special case of the \emph{gap probability}, which is the probability $\Prob_n(x_1 , \dotsc, x_n \in A)$ for a measurable set $A \subseteq \realR$.

We are also interested in the limiting local statistics of particles in the bulk. The gap probability is not an efficient way to describe the local statistics in the bulk, and we compute the limiting \emph{$m$-correlation functions}, which are defined as
\begin{multline} \label{eq:defn_corr_func}
  R^{(m)}_n(x_1, \dotsc, x_m) = \\
  \lim_{\Delta \to 0} \frac{1}{(\Delta)^m} \Prob_n(\text{there is at least one particle in each $[x_i, x_i + \Delta)$, $i = 1, 2, \dotsc, m$}),
\end{multline}
or equivalently as 
\begin{equation} \label{eq:defn_corr_func_alt}
  R^{(m)}_n(x_1, \dotsc, x_m) = 
\frac{n!}{(n-m)!}\int_\realR dx_n\int_\realR  dx_{n-1} \dots \int_\realR dx_{n-m+1}\, P_n(x_1,\dots,x_n),
\end{equation}
where $P_n(x_1, \dots, x_n)$ is the joint density of particles.
Since the eigenvalue distribution of the MNS random matrix model is also given in \eqref{eq:density_n_particle}, the gap probability \eqref{eq:rightmost_pt_as_gap_prob} and the $m$-correlation functions \eqref{eq:defn_corr_func} are the same for the eigenvalues of the MNS random matrix model.

For the MNS (random matrix) model, the gap probability and $m$-correlation functions can be explicitly computed by a contour integral.
\begin{thm} \label{thm:algebraic}
  Given the joint distribution $P_n(x_1, \dotsc, x_n)$ in \eqref{eq:density_n_particle} for $n$ particles, we have the following:
  \begin{enumerate}[label=(\alph*)]
  \item \label{enu:thm:algebraic_1}
    The gap probability is
    \begin{equation} \label{eq:general_gap_prob}
      \Prob_n(x_1, \dotsc, x_n \in A) = \frac{1}{2\pi i} \oint_0 F(z) \det(I - \K(z; q)\chi_{A^c})\frac{dz}{z},
    \end{equation}
    where
    \begin{equation} \label{de:defn_F(z)}
      F(z) = q^{-n(n - 1)/2} (q; q)_n \frac{(-z; q)_{\infty}}{z^n},
    \end{equation}
    and $\K(z; q)$ is the integral operator on $L^2(\realR)$, defined by 
    \begin{equation} \label{eq:kernel_for_each_z}
      \K(z; q)(f)(x) = \int_{\realR} K(x, y; z; q) f(y) dy, \quad K(x, y; z; q) = \sum^{\infty}_{k = 0} \frac{q^k z}{1 + q^k z} \varphi_k(x) \varphi_k(y).
    \end{equation}
  \item \label{enu:thm:algebraic_2}
    The $m$-correlation function is
    \begin{equation} \label{eq:formula_for_R^m}
      R^{(m)}_n(x_1, \dotsc, x_m) = \frac{1}{2\pi i} \oint_0 F(z) \det(K(x_i, x_j; z; q))^m_{i, j = 1} \frac{dz}{z},
    \end{equation}
    and $K(x_i, x_j; z; q)$ is defined in \eqref{eq:kernel_for_each_z}.
  \end{enumerate}
\end{thm}
%\begin{rmk} \label{rmk:double_int_MNS}
%  The kernel function $K(x, y; z; q)$ has the double contour integral representation
%  \begin{equation} \label{eq:double_int_MNS}
%    K(x, y; z; q) = \frac{z}{1 + z} \frac{e^{\frac{y^2 - x^2}{4}}}{(2\pi i)^2} \int^{i\infty}_{-i\infty} ds \oint_{\Gamma_s} dt \frac{e^{(s - x)^2/2}}{e^{(t - y)^2/2}} \qhypergeo{2}{1}{-z, q}{-zq}{\frac{qs}{t}},
%  \end{equation}
%  where $\Gamma_s$ is a positive oriented contour around $0$ such that $0$ and all the poles $q^k s$ ($k = 1, 2, \dotsc$) are enclosed in $\Gamma_s$. The equivalence of \eqref{eq:kernel_for_each_z} and \eqref{eq:double_int_MNS} is proved in the end of Section \ref{subsec:gap_prob}. % This kernel formula is comparable to the kernel formulas \eqref{eq:K_kernel_contour} occurring in Section \ref{sec:relations} for $q$-Whittaker processes and related particle models. 
%\end{rmk}
We note here that a formula equivalent to \eqref{eq:formula_for_R^m} has appeared recently in the physical literature \cite[equation (86)]{Dean-Le_Doussal-Majumdar-Schehr16}. We also remark that the kernel \eqref{eq:kernel_for_each_z} with $z=\lambda>0$ is exactly the one which appears in the grand canonical version of the MNS model \cite{Johansson07}. This is not at all surprising, since the grand canonical ensemble is the superposition of canonical ensembles. Indeed, using the concept of superposition, it is straightforward to prove Theorem \ref{thm:algebraic} using the known determinantal formulas in the grand canonical ensemble. In Section \ref{sec:algebraic} below, we present a different proof of Theorem \ref{thm:algebraic}\ref{enu:thm:algebraic_1} which does not rely on known results for the grand canonical ensemble. Our reason for presenting this longer proof is two-fold. Firstly, it makes the results of the current paper self-contained (independent of the grand canonical ensemble); and secondly, in the process we prove an identity of operators which may have applications in other models in integrable probability, see Section \ref{sec:relations}.

In the theory of point processes, gap probabilities and correlation functions are intimately connected, and it is a standard result that knowledge of one implies knowledge of the other. Thus Theorem \ref{thm:algebraic}\ref{enu:thm:algebraic_1} implies \ref{thm:algebraic}\ref{enu:thm:algebraic_2} (and vice-versa). We prove Theorem \ref{thm:algebraic}\ref{enu:thm:algebraic_1} in detail in Section \ref{subsec:gap_prob}.  The general argument to derive the correlation functions from the gap probabilities is a rather straightforward application of \eqref{eq:defn_corr_func} together with the inclusion/exclusion principle, and we present a short proof or Theorem \ref{thm:algebraic}\ref{enu:thm:algebraic_2} in Section \ref{subsec:corr_func} in the case $m=2$.

%We note that in a determinantal process, the $m$-correlation function is determined by a correlation kernel, and the gap probability is a Fredholm determinant involving the correlation kernel. Moreover, the gap probability has a series expansion by the integral of $m$-correlation functions, see \cite{Tracy-Widom98}. Although the MNS model is not a determinantal process, the gap probability and $m$-correlation functions are also related. See Section \ref{subsec:alt_gap_prob} for detail.

\medskip

For the rightmost particle in the MNS model, or equivalently, the largest eigenvalue in the MNS random matrix model, we state the limiting distribution in two regimes. If the parameter $q$ is in a compact subset of $(0, 1)$, the limiting distribution is the celebrated Tracy--Widom distribution, whose probability distribution function is defined by the Fredholm determinant of $\K_{\Airy}$, an operator on $L^2(\realR)$ with kernel $K_{\Airy}(x, y)$: 
\begin{equation} \label{eq:TW_limit_result}
  F_{\GUE}(t) = \det(I - \Proj_t \K_{\Airy} \Proj_t), \quad \text{and} \quad K_{\Airy}(x, y) = \int^{\infty}_0 \Ai(x + r) \Ai(y + r) dr,
\end{equation}
where $\Proj_t$ is the projection operator defined such that $\Proj_t f(x) = f(x) \chi_{(t, \infty)}(x)$.

If the parameter $q$ is scaled to be close to $1$, such that $1 - q = \bigO(n^{-1/3})$ as $n \to \infty$, the limiting distribution is the so-called crossover distribution that occurs in the weak asymmetric limit of models in the Kardar--Parisi--Zhang (KPZ) universality class \cite{Amir-Corwin-Quastel11}, \cite{Corwin11}, \cite{Quastel12}, and interpolates the Tracy--Widom distribution and the Gumbel distribution \cite{Johansson07}. Its probability distribution function is defined by the Fredholm determinant of $\K_{\crossover}(c)$, an integral operator on $L^2(\realR)$ depending on a continuous parameter $c \in \realR$, whose kernel is $K_{\crossover}(x, y; c)$ given below:
\begin{equation} \label{eq:crossover_limit_result}
  F_{\crossover}(t; c) = \det(I - \Proj_t \K_{\crossover}(c) \Proj_t), \ \ \text{and} \ \ K_{\crossover}(x, y; c) = \int^{\infty}_{-\infty} \frac{e^{-cr}}{1 + e^{-cr}} \Ai(x - r) \Ai(y - r) dr.
\end{equation}
It is clear that as the parameter $c \to -\infty$, $F_{\crossover}(t; c) \to F_{\GUE}(t)$. Our $K_{\crossover}(x, y; c)$ is the correlation kernel of the ``interpolating process'' in \cite{Johansson07}.
\begin{thm} \label{thm:edge}
  Suppose as $n \to \infty$, $s$ depends on $n$ as 
  \begin{equation}
    s \equiv s_n = 2\sqrt{n} + t n^{-1/6}.
  \end{equation}
  Then we have the following.
  \begin{enumerate}[label=(\alph*)]
  \item\label{thm:edge_fixedq}
    Suppose $q \in (0, 1)$ is independent of $n$,
    \begin{equation} \label{eq:TW_limit_fixed_q}
      \lim_{n \to \infty} \Prob_n(\max(x_1, \dotsc, x_n) \leq s_n) = F_{\GUE}(t).
    \end{equation}
  \item\label{thm:edge_varyingq}
    Suppose $q = \exp(-cn^{-1/3})$ depending on $n$, where $c >0$ is a constant,
    \begin{equation} \label{eq:crossover_limit}
      \lim_{n \to \infty} \Prob_n(\max(x_1, \dotsc, x_n) \leq s_n) = F_{\crossover}(t; c).
    \end{equation}
  \end{enumerate}
\end{thm}

For the particles/eigenvalues in the bulk, we also consider their limiting behavior in two regimes. If the parameter $q$ is in a compact subset of $(0, 1)$, the positions of particles in an $\bigO(n^{-1/2})$ window converge to the sine point process \cite[Sections 3.5 and 4.2]{Anderson-Guionnet-Zeitouni10}, with the $m$-correlation functions defined by the correlation kernel
\begin{equation}
  R^{(m)}_{\sin}(x_1, \dotsc, x_m) = \det(K_{\sin}(x_i, x_j))^m_{i, j = 1}, \quad \text{where} \quad K_{\sin}(x, y) = \frac{\sin(\pi(x - y))}{\pi(x - y)}.
\end{equation}
If the parameter is scaled to be close to $1$, such that $1 - q = \bigO(n^{-1})$, the positions of particles in an $\bigO(n^{-1})$ window converge to a determinantal point process that interpolates the sine process and the Poisson process. The $m$-correlation functions of this process are defined by the correlation kernel 
\begin{equation}
  R^{(m)}_{\interpolating}(x_1, \dotsc, x_m; a) = \det(K_{\interpolating}(x_i, x_j; a))^m_{i, j = 1}, \ \ \text{where} \ \ K_{\interpolating}(x, y; a) = \int^{\infty}_0 \frac{\cos(\pi(x - y)t)}{a e^{t^2} + 1} dt,
\end{equation}
which depends on a continuous parameter $a>0$.
We note that as $a \to 0_+$, if $x = \xi/\sqrt{-\log a}$ and $y = \eta/\sqrt{-\log a}$, then
\begin{equation}
  \lim_{a \to 0_+} K_{\interpolating} \left( \frac{\xi}{\sqrt{-\log a}}, \frac{\eta}{\sqrt{-\log a}}; a \right) dy = K_{\sin}(\xi, \eta) d\eta, \quad \text{for $\xi, \eta$ in a compact subset of $\realR$}.
\end{equation}
Our correlation kernel $K_{\interpolating}$ is the same as the kernel $L_c$ in \cite[Theorem 1.9]{Johansson07} up to a change of scaling.

\begin{thm} \label{thm:bulk}
  \begin{enumerate}[label=(\alph*)]
  \item \label{enu:thm:bulk_a}
    Suppose $n \to \infty$, $q \in (0, 1)$ is independent of $n$, and $x_1, \dotsc, x_m$ depend on $n$ as
    \begin{equation}
      x_i = 2x\sqrt{n} + \frac{\pi \xi_i}{(1 - x^2)^{1/2} \sqrt{n}}, \quad i = 1, \dotsc, m,
    \end{equation}
    where $\xi_i$ are constants and $x \in (-1, 1)$. Then
    \begin{equation} \label{eq:sine_limit}
      \lim_{n \to \infty} \left( \frac{\pi}{(1 - x^2)^{1/2} \sqrt{n}} \right)^m R^{(m)}_n(x_1, \dotsc, x_m) = R^{(m)}_{\sin}(\xi_1, \dotsc, \xi_m).
    \end{equation}
  \item \label{enu:thm:bulk_b}
    Suppose $n \to \infty$, $q = e^{-c/n}$, and $x_1, \dotsc, x_m$ depend on $n$ as
    \begin{equation}
      x_i = 2x\sqrt{n} + \frac{\pi \xi_i}{\sqrt{n/c}}, \quad i = 1, \dotsc, m,
    \end{equation}
    where $\xi_i$ are constants and $x \in \realR$. Then
    \begin{equation}
      \lim_{n \to \infty} \left( \frac{\pi}{\sqrt{n/c}} \right)^m R^{(m)}_n(x_1, \dotsc, x_m) = R^{(m)}_{\interpolating} \left( \xi_1, \dotsc, \xi_m; \frac{e^{cx^2}}{e^c - 1} \right).
    \end{equation}
  \end{enumerate}
\end{thm}
\begin{rmk}
  \begin{enumerate}[label=(\roman*)]
  \item 
    As $q \to 0$, the MNS random matrix model \eqref{eq:MNS_random_matrix} converges to the GUE \eqref{eq:GUE_density}. The Tracy--Widom limit at the edge and the sine limit in the bulk for GUE is a well known result in random matrix theory \cite[Chapter 3]{Anderson-Guionnet-Zeitouni10}.
  \item
    Our limiting results for the canonical ensemble of the MNS model agree with those obtained in recent physical works \cite{Dean-Le_Doussal-Majumdar-Schehr15}, \cite{Dean-Le_Doussal-Majumdar-Schehr16}, as well as results for the grand canonical ensemble  \cite{Johansson07}. Although the canonical ensemble is not a determinantal point process, as $n \to \infty$ its scaling limits at the edge and in the bulk are both determinantal point processes.
  \item
    Since the MNS model can be interpreted as a random matrix model, we would like to expect some universality result in the local statistics. However, in the regime $1 - q = \bigO(n^{-1})$, Theorem \ref{thm:bulk}\ref{enu:thm:bulk_b} shows that the limiting local correlation functions depend on $x$, the limiting position of the particles. This is different from most other random matrix models, and is a feature which was not observed in earlier studies of the grand canonical ensemble  \cite{Johansson07}, although the kernel $K_{\interpolating}(x_i, x_j; \frac{e^{cx^2}}{e^c - 1})$ is a specialization of the one obtained recently in \cite[equation (274)]{Dean-Le_Doussal-Majumdar-Schehr16} for free fermions in $d$ dimensions with general potentials.
  \end{enumerate}
\end{rmk}

We note that the $1$-correlation function yields the empirical probability density function $\rho_n(x)$, since
\begin{equation}
  \rho_n(x) = \frac{1}{n} R^{(1)}_n(x).
\end{equation}
From \eqref{eq:sine_limit} we obtain that if $q$ is fixed, then the limiting empirical probability density function is
\begin{equation} \label{eq:semi_circle_law}
  \lim_{n \to \infty} 2\sqrt{n} \rho_n(2\sqrt{n} x) = \frac{2}{\pi} \sqrt{1 - x^2}, \quad x \in (-1, 1).
\end{equation}
Here we use the simple property that $K_{\sin}(x, x) = 1$. This shows that the limiting empirical probability density of the eigenvalues is the semicircle law, the same as that of the GUE random matrix. On the other hand,
\begin{equation}
  K_{\interpolating}(x, x; a) = \frac{-\sqrt{\pi}}{2} \Li_{1/2}(-a^{-1}),
\end{equation}
where $\Li_{1/2}$ is the polylogarithm \cite[25.12.11]{Boisvert-Clark-Lozier-Olver10}. Hence if $q = e^{-c/n}$,
\begin{equation} \label{eq:limiting_global_q_to_1}
  \lim_{n \to \infty} 2\sqrt{n} \rho_n(2\sqrt{n} x) = \frac{-1}{\sqrt{\pi c}} \Li_{1/2}(e^{-cx^2} - e^{c(1 - x^2)}).
\end{equation}
This limiting distribution on the right-hand side of \eqref{eq:limiting_global_q_to_1} is supported on $\realR$, but as $c \to +\infty$, it converges to the semicircle law on the right-hand side of \eqref{eq:semi_circle_law} which is supported on $[-1, 1]$. The limiting empirical probability density function \eqref{eq:limiting_global_q_to_1} agrees with \cite[Formula (8)]{Dean-Le_Doussal-Majumdar-Schehr15}. The asymptotics of $\Li_{1/2}$ can be found in \cite{Wood92}.

\subsection{Generalizations and related models}

The most interesting feature of the MNS model is that its rightmost particle has a similar distribution to the edge particle of several interacting particle models related to Kardar--Parisi--Zhang (KPZ) universality class. In fact, the similarity is not only at the level of the limiting distribution, but also at the level of the algebraic structure for the finite systems. However, this similarity will be clear only after some technical results are established, so we refer the reader to Section \ref{sec:relations} for detail. It is also worth noticing that the very recent preprint \cite{Cunden-Mezzadri-OConnell17} suggests other random matrix models analogous to the MNS random matrix model. Below we discuss the dynamical generalization of the MNS model and compare it with the nonintersecting Brownian motions on a circle.

\subsubsection{Relation to time-periodic Ornstein--Uhlenbeck (OU) processes and the multi-time correlations} \label{subsubsec:multi-time_OU}

It was first noticed by Johansson \cite[Section 1.2]{Johansson07} that the MNS model has an interpretation in terms of time-periodic nonintersecting paths. Our presentation of the relation to Ornstein--Uhlenbeck (OU) process and the multi-time correlations is based on the recent preprint \cite{Le_Doussal-Majumdar-Schehr17} and the physical concepts are explained therein.

%It is not a surprise that the quantum harmonic oscillator is related to the OU process. 
The imaginary time propagator of the harmonic oscillator, or more precisely, the particle with Hamiltonian \eqref{eq:Hamiltonian_normalized}, is, by letting with $\hbar = 1$, $m = 1/2$ and $\omega = 1$ in \cite[Formula (17)]{Le_Doussal-Majumdar-Schehr17},
\begin{equation} \label{eq:imaginary_time_prop}
  G(y, \tau \mid x, 0) = \sum^{\infty}_{k = 0} \varphi_k(x) \varphi_k(y) e^{-k\tau},
\end{equation}
where $\tau$ is the imaginary time. Consider the OU process defined by the stochastic differential equation
\begin{equation} \label{eq:OU_proc}
  dx(t) = -x(t) dt + \sqrt{2} dW(t),
\end{equation}
where $W(t)$ is the Wiener process. (Note that our OU process differs from that defined by \cite[Formula (1)]{Le_Doussal-Majumdar-Schehr17} by the choice of constants in the stochastic differential equation.) The imaginary time propagator in \eqref{eq:imaginary_time_prop} is equal to the OU propagator up to a conjugation, see \cite[Formula (19)]{Le_Doussal-Majumdar-Schehr17}. Hence free fermions in a quadratic potential well are related to the nonintersecting OU processes, due to the analogy between the Slater determinant for the former and the Karlin--McGregor formula for the latter. In particular, the ground state of free fermions in a quadratic potential has the same probability density function as that of the one-time distribution of the stationary nonintersecting OU processes. That is, the limit of nonintersecting OU processes starting from time $-M$ and ending at time $M$ as $M \to \infty$ and both the starting and ending positions are close to the origin, since their probability density functions are both time invariant and identical to that of the eigenvalues of a GUE random matrix, see \cite[Formulas (7) and (62)]{Le_Doussal-Majumdar-Schehr17}. 

A key observation in \cite{Le_Doussal-Majumdar-Schehr17} is that the probability density function of the MNS model, or more precisely, the density $P_n(x_1, \dotsc, x_n)$ defined by \eqref{eq:Hamiltonian_normalized}--\eqref{eq:density_n_particle}, is the same as the stationary distribution of the $n$ particles in nonintersecting OU processes defined in \eqref{eq:OU_proc} with time-periodic boundary condition and the period $\beta = 1/(\kappa T) = 1/T$, see Figure \ref{fig:OU_paths}. To explain the stationary distribution, we consider the OU processes $x_1(t), \dotsc, x_n(t)$, such that they are conditioned not to intersect during time $[0, \beta]$, and they satisfy $x_k(0) = x_k(\beta) = y_k$. Suppose $y_1, \dotsc, y_n$ has a joint probability distribution $F$. Then for any $\tau \in (0, \beta)$, $x_1(\tau), \dotsc, x_n(\tau)$ has a joint probability distribution $F_{\tau}$ that depends on $\tau$ and $F$. By explicit computation we verify that $F_{\tau} = F$ for all $\tau \in (0, \beta)$ if and only if $F$ has the probability density function $P_n(y_1, \dotsc, y_n)$ given in \eqref{eq:density_n_particle}. Hence we claim that the distribution of the free fermions at temperature $T$ given by \eqref{eq:density_n_particle} is the stationary distribution of nonintersecting OU processes with time period $1/T$. Also we call the nonintersecting OU processes with time period $1/T$ stationary if its marginal distribution at time $0$ is given by $P_n$ in \eqref{eq:density_n_particle}.

\begin{figure}[htb]
  \centering
  \includegraphics{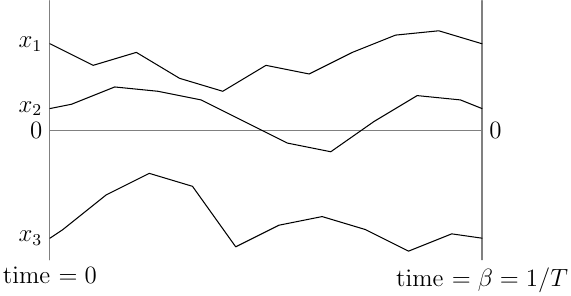}
  \caption{Schematic paths for three particles in nonintersecting OU processes with time period $\beta = 1/T$.}
  \label{fig:OU_paths}
\end{figure}

As a quantum mechanical ensemble, we can consider the dynamics of the MNS model. As often happens, the dynamics of the MNS model along imaginary time is mathematically easier. In \cite{Le_Doussal-Majumdar-Schehr17}, the multi-time joint probability density function of the MNS model along imaginary time is derived, and also the multi-time correlation functions along imaginary time. To be precise, suppose that $\tau_1 < \tau_2 < \dotsb < \tau_m$ are in the interval $[0, \beta)$ and they denote the imaginary times, the multi-time joint probability density function is obtained in \cite[Formula (79)]{Le_Doussal-Majumdar-Schehr17} and the multi-time correlation functions are obtained in \cite[Formula (83)]{Le_Doussal-Majumdar-Schehr17}. Moreover, the multi-time distribution of the stationary nonintersecting OU processes with time period $1/T$ is the same as that of the MNS model along imaginary time, see \cite[Section VII.A]{Le_Doussal-Majumdar-Schehr17}.
\begin{prop} \cite[Formulas (79) and (83), and Section VII.A]{Le_Doussal-Majumdar-Schehr17} \label{prop:multi}
  \begin{itemize}
  \item 
    Let $n$ free fermions be in the quadratic potential well, defined by the Hamiltonian \eqref{eq:Hamiltonian_normalized}, at temperature $T > 0$. Suppose $0 \leq \tau_1, \tau_2, \dotsc, \tau_m < \beta = 1/T$. Then the joint probability density of the fermions at imaginary times $\tau_1, \dotsc, \tau_m$ is, if $\tau_1 < \tau_2 < \dotsb < \tau_m$,
    \begin{multline} \label{eq:multi-time_jpdf}
      P_n(\x^{(1)}, \dotsc, \x^{(m)}) = \frac{q^{n/2}}{Z_n(q)} \left[ \prod^{m - 1}_{l = 1} \det \left( G(x^{(l + 1)}_j, \tau_{l + 1} - \tau_l \mid x^{(l)}_k, 0) \right)^n_{j, k = 1} \right] \\
      \times \det \left( G(x^{(1)}_j, \beta - (\tau_m - \tau_1) \mid x^{(m)}_k, 0) \right)^n_{j, k = 1},
    \end{multline}
    where $\x^{(l)} = (x^{(l)}_1, \dotsc, x^{(l)}_n)$ are the positions of the fermions at time $\tau_l$; the multi-time correlation function of the fermions at imaginary times $\tau_1, \dotsc, \tau_m$ is
    \begin{equation} \label{eq:multi-time_corr_func}
      R^{(m)}_n(x_1, \dotsc, x_m; \tau_1, \dotsc, \tau_m) = \frac{q^{n/2}}{Z_n(q)} \sum_{0 \leq k_1 < k_2 < \dotsb < k_n} \det \left( K_{k_1, \dotsc, k_n}(x_i, x_j; \tau_i, \tau_j) \right)^m_{i, j = 1},
    \end{equation}
    where the kernel function $K_{k_1, \dotsc, k_n}(x_i, x_j; \tau_i, \tau_j)$ will be defined in \eqref{eq:multi_time_corr_kernel} in Section \ref{sec:multi_time}.
  \item
    Let $x_1(t), \dotsc, x_n(t)$ be $n$ independent OU processes defined in \eqref{eq:OU_proc}. Condition them to be nonintersecting over time $[0, \beta = 1/T]$, and $x_k(0) = x_k(\beta) = y_k$, with the positions $y_1, \dotsc, y_n$ be random variables with joint probability density function $P_n(y_1, \dotsc, y_n)$ defined in \eqref{eq:density_n_particle}. Then the joint probability density function of the particles at times $\tau_1, \tau_2, \dotsc, \tau_m \in [0, \beta)$ is given by \eqref{eq:multi-time_jpdf} if $\tau_1 < \tau_2 < \dotsb < \tau_m$, and the multi-time correlation function is given by \eqref{eq:multi-time_corr_func}.
  \end{itemize}
\end{prop}
Here we remark that since the finite-temperature Green's function for a quantum system at temperature $T > 0$ is (anti)periodic in imaginary time with period $\beta = 1/(\kappa T)$ (see \cite[Chapter 7]{Fetter-Walecka12} for an explanation), it suffices to consider multi-time joint probability density function and correlation functions at imaginary times in $[0, \beta)$. 

We can simplify the multi-time correlation function \eqref{eq:multi-time_corr_func} in to a form analogous to \eqref{eq:formula_for_R^m}, and derive a formula for the multi-time gap probability that is analogous to \eqref{eq:general_gap_prob}. Before giving our results, we introduce some notations. Define
\begin{equation} \label{eq:defn_E}
  E(x, y; \tau, \sigma) =
  \begin{cases} \displaystyle
    0 & \text{if $\tau \geq \sigma$}, \\
    \begin{aligned}[b]
      & \sum^{\infty}_{k = 0} \varphi_k(x) \varphi_k(y) e^{k(\tau - \sigma)} = \frac{1}{\sqrt{2\pi(1 - e^{2(\tau - \sigma)})}}  \\ & \times \exp \left( \frac{-(1 + e^{2(\tau - \sigma)})(x^2 + y^2) + 4e^{\tau - \sigma} xy}{4(1 - e^{2(\tau - \sigma)}} \right) 
    \end{aligned}
     & \text{if $\tau < \sigma$}.
  \end{cases}
\end{equation}
Note that for $\tau < \sigma$, $E(x, y; \tau, \sigma) = G(y, \tau - \sigma \mid x, 0)$, the imaginary time propagator defined in \eqref{eq:imaginary_time_prop}. Then define
\begin{equation} \label{eq:kernel_K(xytausigmazq)}
  K(x, y; \tau, \sigma; z; q) = \sum^{\infty}_{k = 0} \frac{q^k z}{1 + q^k z} \varphi_k(x) \varphi_k(y) e^{k(\tau - \sigma)} - E(x, y; \tau, \sigma),
\end{equation}
of which the function $K(x, y; z; q)$ in \eqref{eq:kernel_for_each_z} is the $\tau = \sigma$ specialization. Furthermore, we define the integral operator $\K(\tau_1, \dotsc, \tau_m; z; q)$ on $L^2(\realR \times \{ 1, \dotsc, m \})$ whose kernel is represented by an $m \times m$ matrix $(K(x_i, x_j; \tau_i, \tau_j; z; q))^m_{i, j = 1}$, where $K(x, y; \tau, \sigma; z; q)$ is defined in \eqref{eq:kernel_K(xytausigmazq)}. To be concrete, for a function $f$ on $\realR \times \{ 1, \dotsc, m \}$, we denote it by $(f(x; 1), \dotsc, f(x; m))$ where $f(x; k)$ is a function on $\realR \times \{ k \}$, and have
\begin{equation} \label{eq:defn_final_K_op}
  (\K(\tau_1, \dotsc, \tau_m; z; q)f)(x; k) = \sum^m_{j = 1} \int_{\realR} K(x, y; \tau_k, \tau_j; z; q) f(y; j) dy.
\end{equation}
At last if $A_1, \dotsc, A_m \subseteq \realR$ are measurable sets, we denote $\chi_{A_1, \dotsc, A_m}$, the projection operator on $L^2(\realR \times \{ 1, 2, \dotsc, m \})$, such that
\begin{equation} \label{eq:proj_op_A}
  (\chi_{A_1, \dotsc, A_m}f)(x; k) =
  \begin{cases}
    f(x; k) & \text{if $x \in A_k$ for $k = 1, \dotsc, m$}, \\
    0 & \text{otherwise}.
  \end{cases}
\end{equation}
Our result is as follows:
\begin{thm} \label{thm:multi_formulas}
  Consider either the $n$ free fermions at temperature $1/T$ or the $n$-particle time-periodic nonintersecting OU processes with period $1/T$ that is defined in Proposition \ref{prop:multi}. Let $\tau_1, \dotsc, \tau_m \in [0, 1/T)$ be either the imaginary times for free fermions or the times for OU processes.
  \begin{enumerate}[label=(\alph*)]
  \item \label{enu:thm:multi_formulas_a}
    The multi-time $m$-correlation function \eqref{eq:multi-time_corr_func} at $\tau_1, \dotsc, \tau_m \in (0, 1/T)$, as stated in Proposition \ref{prop:multi}, can be written as
    \begin{equation} \label{eq:contour_integral_repr_m-corr}
      R^{(m)}_n(x_1, \dotsc, x_m; \tau_1, \dotsc, \tau_m) = \frac{1}{2\pi i} \oint_0 F(z) \det \left( K(x_i, x_j; \tau_i, \tau_j; z; q) \right)^m_{i, j = 1} dz,
    \end{equation}
    where $F(z)$ is defined in \eqref{de:defn_F(z)} and $K(x_i, x_j; \tau_i, \tau_j; z; q)$ is defined in \eqref{eq:kernel_K(xytausigmazq)}.
  \item \label{enu:thm:multi_formulas_b}
    Suppose $\tau_1, \dotsc, \tau_m$ are distinct. Let $A_1, \dotsc, A_m \subseteq \realR$ be measurable sets. The gap probability that at imaginary time $\tau_k$, all fermions are in $A_k$, or that at time $\tau_k$, all particles in OU processes are in $A_k$, is given by
    \begin{equation}
      \Prob_n(A_1, \dotsc, A_m; \tau_1, \dotsc, \tau_m) = \frac{1}{2\pi i} \oint_0 F(z) \det(I - \K(\tau_1, \dotsc, \tau_m; z; q) \chi_{A^c_1, \dotsc, A^c_m}),
    \end{equation}
    where the operators $\K(\tau_1, \dotsc, \tau_m; z; q)$ and $\chi_{A^c_1, \dotsc, A^c_m}$ are defined in \eqref{eq:defn_final_K_op} and \eqref{eq:proj_op_A}.
  \end{enumerate}
\end{thm}
We note that if $\tau_1 = \dotsb = \tau_m$, Theorem \ref{thm:multi_formulas}\ref{enu:thm:multi_formulas_a} degenerates to Theorem \ref{thm:algebraic}\ref{enu:thm:algebraic_2}. Hence we can compute the limits of the multi-time correlation functions and gap probability, which will be done in a subsequent publication.

\subsubsection{Nonintersecting Brownian motions on a circle}

The $m$-correlation function formula \eqref{eq:formula_for_R^m} for the MNS model is  given by a contour integral whose integrand is a $m \times m$ determinant depending on a formal correlation kernel. This feature occurs in another model, nonintersecting Brownian motions on a circle with a fixed winding number, studied by the authors in \cite{Liechty-Wang14-2}. To see the analogy, we recall that the counterpart $m$-correlation function in \cite{Liechty-Wang14-2} is $(R^{(n)}_{0 \to T})_{\omega}(a^{(1)}_1, \dotsc, a^{(1)}_{k_1}; \dotsc; a^{(m)}_1, \dotsc, a^{(m)}_{k_m}; t_1, \dotsc, t_m)$ defined in \cite[Formulas (36) and (136)]{Liechty-Wang14-2}, where $\omega$ is the winding number. (The correlation function $(R^{(n)}_{0 \to T})_{\omega}$ is a multi-time correlation function, more analogous to $R^{(m)}(x_1, \dotsc, x_m; \tau_1, \dotsc, \tau_m)$ defined in \eqref{eq:multi-time_corr_func} and re-expressed in \eqref{eq:contour_integral_repr_m-corr}, to which $R^{(m)}_n(x_1, \dotsc, x_m)$ defined in \eqref{eq:defn_corr_func} is a special case.) By \cite[Formula (135)]{Liechty-Wang14-2}, we have that\footnote{In \cite[Formula (135)]{Liechty-Wang14-2}, the symbol $o$ in the third line should be $\omega$.}
\begin{multline} \label{eq:NIBM_winding_k_corr}
  (R^{(n)}_{0 \to T})_{\omega}(a^{(1)}_1, \dotsc, a^{(1)}_{k_1}; \dotsc; a^{(m)}_1, \dotsc, a^{(m)}_{k_m}; t_1, \dotsc, t_m) = \\
  \frac{(-1)^{n + 1}}{2\pi i} \oint_0 R^{(n)}_{0 \to T} \left( a^{(1)}_1, \dotsc, a^{(1)}_{k_1}; \dotsc; a^{(m)}_1, \dotsc, a^{(m)}_{k_m}; t_1, \dotsc, t_m; \frac{\log z}{2\pi i} \right) \frac{dz}{z^{\omega + 1}}.
\end{multline}
Then by \cite[Formula (116)]{Liechty-Wang14-2}
\begin{multline} \label{eq:NIBM_k_corr}
  R^{(n)}_{0 \to T}(a^{(1)}_1, \dotsc, a^{(1)}_{k_1}; \dotsc; a^{(m)}_1, \dotsc, a^{(m)}_{k_m}; t_1, \dotsc, t_m; \tau) = \\
  \det \left( K_{t_i, t_j}(a^{(i)}_{l_i}, a^{(j)}_{l'_j}) \right)_{i, j = 1, \dotsc, m, l_i = 1, \dotsc, k_i, l'_j = 1, \dotsc, k_j},
\end{multline}
where $K_{t_i, t_j}(a^{(i)}_{l_i}, a^{(j)}_{l'_j})$ are given in \cite[Formula (117)]{Liechty-Wang14-2} and depend on $\tau$, see \cite[Remark 2.2]{Liechty-Wang14-2}. We note that the right-hand side of \eqref{eq:NIBM_winding_k_corr} is a holomorphic function since $K_{t_i, t_j}(a^{(i)}_{l_i}, a^{(j)}_{l'_j})$ in \eqref{eq:NIBM_k_corr} is analytic in $e^{2\pi i \tau}$.

The formal similarity of correlation functions between the MNS model and the nonintersecting Brownian motions on a circle is intuitively explained by their periodicities. The MNS model is related to the nonintersecting OU processes with time periodicity, see Section \ref{subsubsec:multi-time_OU}, so it is comparable to the nonintersecting Brownian motions with space periodicity.

Another similarity of the two models is as follows. The grand canonical ensemble of the MNS model, which is the superposition of (the canonical ensemble of) the MNS model according to the Boltzmann distribution, is a determinantal process \cite{Dean-Le_Doussal-Majumdar-Schehr15}. Its counterpart, the nonintersecting Brownian motions on a circle with free winding number, also forms a determinantal process \cite[Section 2.3]{Liechty-Wang14-2}.

\subsection*{Outline}

In Section \ref{sec:algebraic} we prove Theorem \ref{thm:algebraic}. 
%First we prove Theorem \ref{thm:algebraic}\ref{enu:thm:algebraic_1} in Section \ref{subsec:gap_prob}, and then prove Theorem \ref{thm:algebraic}\ref{enu:thm:algebraic_2} in Section \ref{subsec:corr_func}. In Section \ref{subsec:alt_gap_prob} we give an alternative proof of Theorem \ref{thm:algebraic}\ref{enu:thm:algebraic_1}  based on Theorem \ref{thm:algebraic}\ref{enu:thm:algebraic_2}. Although the proof in Section \ref{subsec:alt_gap_prob} is shorter, since that in Section \ref{subsec:gap_prob} has analogy in the study of particle models related to KPZ universality class (see Section \ref{sec:relations}), we also present it. 
In Sections \ref{sec:proof_edge} and \ref{sec:proof_bulk} we prove Theorems \ref{thm:edge} and \ref{thm:bulk} respectively. In Section \ref{sec:relations} we discuss some particle models related to KPZ universality class. In Section \ref{sec:multi_time} we prove Theorem \ref{thm:multi_formulas} for the dynamic generalization of the MNS model. In Appendix \ref{sec:relation_to_MNS_RM} we present a proof of Proposition \ref{prop:Moshe-Neuberger-Shapiro94}.

\subsection*{Acknowledgments}

We thank Gr\'{e}gory Schehr for helpful discussion on the dynamics of the MNS model and Jacek Grela for comments on literature. We also thank Ivan Corwin for calling our attention to mistakes in Section \ref{sec:relations} in an earlier version.

\section{Proof of Theorem \ref{thm:algebraic}} \label{sec:algebraic}
Here we present a proof of Theorem \ref{thm:algebraic}\ref{enu:thm:algebraic_1} which is independent of known results for the grand canonical ensemble. Then Theorem \ref{thm:algebraic}\ref{enu:thm:algebraic_2} follows from the general theory of point processes, and we present a short proof in the case $m=2$. The extension to general $m$ is straightforward. 

\subsection{Gap probability} \label{subsec:gap_prob}

Let $A \subseteq \realR$ be a measurable set. We consider the probability that all the $n$ particles are in $A$, which we denote by $\Prob_n(x_1, \dotsc, x_n \in A)$. We have
\begin{equation} \label{eq:gap_prob_in_each_state}
  \begin{split}
     \Prob_n(x_1, \dotsc, x_n \in A) = {}& \int_A \dotsi \int_A P_n(x_1, \dotsc, x_n) dx_1 \dotsm dx_n \\
    = {}& \frac{q^{n/2}}{Z_n(q)} \int_A \dotsi \int_A \sum_{0 \leq k_1 < k_2 < \dotsb < k_n} \lvert \Phi_{k_1, \dotsc, k_n}(x_1, \dotsc, x_n) \rvert^2 q^{k_1 + \dotsb + k_n} dx_1 \dotsm dx_n \\
    = {}& \frac{q^{n/2}}{Z_n(q)} \sum_{0 \leq k_1 < k_2 < \dotsb < k_n} \int_A \dotsi \int_A \lvert \Phi_{k_1, \dotsc, k_n}(x_1, \dotsc, x_n) \rvert^2 q^{k_1 + \dotsb + k_n} dx_1 \dotsm dx_n.
  \end{split}
\end{equation}
Note that by the \Andreief\ formula,
\begin{equation}
  \begin{split}
    & \int_A \dotsi \int_A
    \begin{vmatrix}
      \varphi_{k_1}(x_1) & \dots & \varphi_{k_1}(x_n) \\
      \vdots & & \vdots \\
      \varphi_{k_1}(x_1) & \dots & \varphi_{k_1}(x_n) \\
    \end{vmatrix}^2
    q^{k_1 + \dotsb + k_n} dx_1 \dotsm dx_n \\
    = {}& n! q^{k_1 + \dotsb + k_n} \det \left( \langle \varphi_{k_i}(x), \varphi_{k_j}(x) \rangle_A \right)^n_{i, j = 1} \\
    = {}& n! \det \left( \langle q^{k_i} \varphi_{k_i}(x), \varphi_{k_j}(x) \rangle_A \right)^n_{i, j = 1}, \\
    % = {}& n! \det \left( q^{k_i} \delta_{k_i, k_j} - \langle q^{k_i} \varphi_{k_i}(x), \varphi_{k_j}(x) \rangle_A \right)^n_{i, j = 1},
  \end{split}
\end{equation}
where
\begin{equation}
  \langle f(x), g(x) \rangle_A = \int_A f(x)g(x) dx.
\end{equation}
Hence 
\begin{equation} \label{eq:gap_prob_in_product_series}
  \Prob_n(x_1, \dotsc, x_n \in A) = \frac{q^{n/2}}{Z_n(q)} \sum_{0 \leq k_1 < k_2 < \dotsb < k_n} \det \left( \langle q^{k_i} \varphi_{k_i}(x), \varphi_{k_j}(x) \rangle_A \right)^n_{i, j = 1}.
\end{equation}

Recall the integral operator $\K(z; q)$ defined in \eqref{eq:formula_for_R^m}. We now introduce another integral operator $\M(q)$ acting on $L^2(\realR)$, depending on the parameter $q\in (0,1)$. It is defined by (here $M(x, y; q)$ is identical to $G(y, \tau \mid x, 0)$ in \eqref{eq:imaginary_time_prop})
\begin{equation} \label{eq:kernel_M(q)}
  \M(q)(f)(x) = \int_{\realR} M(x, y; q) f(y) dy, \quad M(x, y; q) = \sum_{k = 0}^\infty q^k \varphi_k(x) \varphi_k(y).
\end{equation}
Let $A\subseteq \realR$ be a measurable set, and let $\chi_A$ be the projection onto $L^2(A)$. It is straightforward to check by definition that $\M(q)$ and $\K(z; q)$ are trace class operators for $0<q<1$, and then $\M(q) \chi_A$ and $\K(z; q) \chi_{A^c}$ are also trace class operators \cite{Simon05}. Hence the Fredholm determinants $\det(I + z \M(q) \chi_A)$ and $\det(I - \K(z; q) \chi_{A_c})$ are well defined. We have the following relation between $\M (q)$ and $\K (z; q)$.

\begin{lem}\label{prop:MK_identity}
  Let $q \in (0, 1)$. For any $z \in \compC$, and for any measurable $A\subseteq \realR$, the following identity holds:
  \begin{equation} \label{eq:MK_identity}
    (I+z \M(q) \chi_A) = (I + z\M(q)) (I - \K(z; q) \chi_{A^c}).
  \end{equation}
  Hence
  \begin{equation} \label{eq:MK_identity_det}
    \det(I + z \M(q) \chi_A) = \det(I + z\M(q)) \det(I - \K(z; q) \chi_{A^c}).
  \end{equation}
\end{lem}

\begin{proof}
  Since the Hermite functions $\{\varphi_k(x)\}_{k=0}^\infty$ form an orthonormal basis for $L^2(\realR)$, it is easy to see that
  \begin{equation}\label{eq:powers_of_M}
    \M(q^k) = \M(q)^k.
  \end{equation}
We define the resolvent operator $\R(z; q)$ by
  \begin{equation} \label{eq:defn_R_inverse}
    I - \R(z; q) = (I + z\M(q))^{-1}.
  \end{equation}
  If $\lvert z \rvert < 1$, we have that $\R(z; q)$ is a well-defined integral operator and 
  \begin{equation}
    \R(z; q) = -\sum_{l=1}^{\infty} (-z \M(q))^l.
  \end{equation}
  Assuming for now that $\lvert z \lvert < 1$, and using the fact that the functions $\varphi_k(x)$ are uniformly bounded in $k$ and $x$ (see, e.g. \cite[22.14.17]{Abramowitz-Stegun64}), we have that uniformly for all $x, y \in \realR$
  \begin{equation} \label{eq:kernel_expanded}
    \begin{aligned}
      K(x, y; z; q) &= \sum^{\infty}_{k = 0} q^k z\varphi_k(x) \varphi_k(y) \sum_{l =0}^\infty (-1)^l z^l q^{l k} \\
      &= \sum_{l=1}^\infty (-1)^{l+1} z^l \sum_{k=0}^\infty q^{kl} \varphi_k(x) \varphi_k(y) \\
      &= \sum_{l=1}^\infty (-1)^{l+1} z^l M(x,y; q^l). \\
    \end{aligned}
  \end{equation}
  This implies the identity that
    \begin{equation} \label{eq:K_expressed_in_M}
    \K(z; q) = \R(z; q),
  \end{equation}
  for all $\lvert z \rvert < 1$.
Using the identity $\K(z; q) \chi_{A^c} = \R(z; q) \chi_{A^c}$ we find
  \begin{equation}\label{eq:inv_sub}
    \begin{aligned}
      (I + z \M(q))(I - \K(z; q) \chi_{A^c}) = {}& I + z \M(q) - \R(z; q) \chi_{A^c} -\M(q)\R(z; q) \chi_{A^c} \\
      = {}& I+z \M(q) \chi_A + (z\M(q) - \R(z; q) - \M(q) \R(z; q)) \chi_{A^c} \\
      = {}& I+z \M(q) \chi_A,
    \end{aligned}
  \end{equation}
  where in the last step we use $z\M(q) - \R(z; q) - \M(q) \R(z; q) = 0$, which is a consequence of \eqref{eq:defn_R_inverse}. Hence we prove \eqref{eq:MK_identity} in the case $\lvert z \rvert < 1$. Since the integral operator $\K(z; q)$ is well defined for all $z \in \compC$, by analytic continuation \eqref{eq:MK_identity} holds for all $z \in \compC$. 
\end{proof}

We expand the Fredholm determinant $\det(I + z \M(q)\chi_A)$ into a series of multiple integrals by \cite[Theorem 3.10]{Simon05}, and then simplify it by the Cauchy--Binet identity as follows.
\begin{equation}
  \begin{split}
    \det(I + z \M(q)\chi_A) = {}& 1 + \frac{z}{1!} \int_A M(x, x; q) dx + \frac{z^2}{2!} \int_A dx_1 \int_A dx_2 \det(M(x_i, x_j; q))^2_{i, j = 1} + \dotsb \\
    = {}& 1 + z \left( \sum_{0 \leq k_1} \langle q^{k_1} \varphi_{k_1}, \varphi_{k_1} \rangle_A \right) \\
    & + z^2 \left( \sum_{0 \leq k_1 < k_2} 
      \begin{vmatrix}
        \langle q^{k_1} \varphi_{k_1}, \varphi_{k_1} \rangle_A & \langle q^{k_1} \varphi_{k_1}, \varphi_{k_2} \rangle_A \\
        \langle q^{k_2} \varphi_{k_2}, \varphi_{k_1} \rangle_A & \langle q^{k_2} \varphi_{k_2}, \varphi_{k_2} \rangle_A
      \end{vmatrix} \right) \\
    & + z^3 \left( \sum_{0 \leq k_1 < k_2 < k_3} 
      \begin{vmatrix}
        \langle q^{k_1} \varphi_{k_1}, \varphi_{k_1} \rangle_A & \langle q^{k_1} \varphi_{k_1}, \varphi_{k_2} \rangle_A & \langle q^{k_1} \varphi_{k_1}, \varphi_{k_3} \rangle_A \\
        \langle q^{k_2} \varphi_{k_2}, \varphi_{k_1} \rangle_A & \langle q^{k_2} \varphi_{k_2}, \varphi_{k_2} \rangle_A & \langle q^{k_2} \varphi_{k_2}, \varphi_{k_3} \rangle_A \\
        \langle q^{k_3} \varphi_{k_3}, \varphi_{k_1} \rangle_A & \langle q^{k_3} \varphi_{k_3}, \varphi_{k_2} \rangle_A & \langle q^{k_3} \varphi_{k_3}, \varphi_{k_3} \rangle_A
      \end{vmatrix} \right) \\
    & + \dotsb.
  \end{split}
\end{equation}
With the help of \eqref{eq:explicit_Z_n(q)}, \eqref{eq:density_n_particle} and \eqref{eq:gap_prob_in_product_series} thus find
\begin{equation}
  \det(I + z \M(q)\chi_A)  =  1 + \sum^{\infty}_{n = 1} \frac{z^n}{q^{n(n - 1)/2} (q; q)_n} \Prob_n(x_1, \dotsc, x_n \in A),
\end{equation}
and arrive at the formula for any dimension $n$,
\begin{equation} \label{eq:contour_int_formula_right_most_particle1}
  \Prob_n(x_1, \dotsc, x_n \in A) = q^{n(n - 1)/2} (q; q)_n \frac{1}{2\pi i} \oint_0 \det(I + z \M(q)\chi_A) \frac{dz}{z^{n + 1}}.
\end{equation}
In order to do asymptotic analysis it is convenient to work with the operator $\K(z; q)$ rather than $\M(q)$. Since the operator $\M(q)$ is diagonalized by $\{\varphi_k\}$, the determinant is simple to compute:
\begin{equation} \label{eq:Fred_det_id_M(q)}
  \det(I + z \M(q)) = \prod^{\infty}_{k = 0} (1 + q^k z) = (-z; q)_{\infty}.
\end{equation}
Thus substituting \eqref{eq:MK_identity_det} and \eqref{eq:Fred_det_id_M(q)} into \eqref{eq:contour_int_formula_right_most_particle1}, we obtain the formula \eqref{eq:general_gap_prob} and prove Theorem \ref{thm:algebraic}\ref{enu:thm:algebraic_1}. In particular, when $A = (-\infty, s]$, \eqref{eq:rightmost_pt_as_gap_prob} implies
\begin{equation} \label{eq:contour_int_formula_right_most_particle2}
  \Prob_n(\max(x_1, \dotsc, x_n) \leq s) = q^{-n(n - 1)/2} (q; q)_n \frac{1}{2\pi i} \oint_0 \frac{(-z; q)_{\infty}}{z^n} \det(I - \K(z; q)\chi_{(s, \infty)})\frac{dz}{z}.
\end{equation}

\subsection{Correlation functions} \label{subsec:corr_func}
We now prove Theorem \ref{thm:algebraic}\ref{enu:thm:algebraic_2} assuming the result \ref{thm:algebraic}\ref{enu:thm:algebraic_1}. We present the proof of \eqref{eq:formula_for_R^m} for $m=2$, but the proof is nearly identical for any positive integer $m$. Fix $x_1, x_2\in \realR$ and $\Delta>0$, and introduce the notations
\begin{equation}
A_i^\Delta = [x_i, x_i+\Delta), \quad G_i^\Delta = \{ \textrm{there are no particles in} \ A_i^\Delta\}.
\end{equation} 
We will use the definition \eqref{eq:defn_corr_func} for the $m$-correlation function, and note that 
\begin{multline}
\Prob_n(\text{there is at least one particle in each $[x_i, x_i + \Delta )$, $i = 1, 2$}) = 1-\Prob_n(G_1^\Delta\cup G_2^\Delta) \\
= 1-\left[\Prob_n(G_1^\Delta)+\Prob_n(G_2^\Delta)-\Prob_n(G_1^\Delta\cap G_2^\Delta)\right].
\end{multline}
Using the formula \eqref{eq:general_gap_prob} for the gap probabilities and expanding the Fredholm determinants as series, this is
\begin{multline}
1-\frac{1}{2\pi i } \oint_0 \frac{F(z)}{z}\left[1+\int_{A_1^\Delta} K(y,y)\,dy + \frac{1}{2!}\int_{A_1^\Delta} \int_{A_1^\Delta} \det[K(y_i,y_j)]_{i,j=1}^2\,dy_1\,dy_2 + \bigO(\Delta^3) \right. \\
+1+\int_{A_2^\Delta} K(y,y)\,dy + \frac{1}{2!}\int_{A_2^\Delta} \int_{A_2^\Delta} \det[K(y_i,y_j)]_{i,j=1}^2\,dy_1\,dy_2  + \bigO(\Delta^3) \\
 -1-\int_{A_1^\Delta\cup A_2^\Delta} K(y,y)\,dy - \frac{1}{2!}\int_{A_1^\Delta} \int_{A_1^\Delta} \det[K(y_i,y_j)]_{i,j=1}^2\,dy_1\,dy_2 - \frac{1}{2!}\int_{A_2^\Delta} \int_{A_2^\Delta} \det[K(y_i,y_j)]_{i,j=1}^2\,dy_1\,dy_2\\
\left. - \frac{1}{2!}\int_{A_1^\Delta} \int_{A_2^\Delta} \det[K(y_i,y_j)]_{i,j=1}^2\,dy_1\,dy_2- \frac{1}{2!}\int_{A_2^\Delta} \int_{A_1^\Delta} \det[K(y_i,y_j)]_{i,j=1}^2\,dy_1\,dy_2 + \bigO(\Delta^3)\right]\,dz,
\end{multline} 
where for brevity we have used $K(y_1, y_2) \equiv K(y_1, y_2; z; q)$.
Noting all of the cancellations and the fact that $\frac{1}{2\pi i}\oint \frac{F(z)}{z}\,dz = \Prob_n(\textrm{all particles are in} \ \realR)=1$ we find
\begin{multline}
\Prob_n(\text{there is at least one particle in each $[x_i, x_i + \Delta x)$, $i = 1, 2$}) =\\
 \frac{1}{2\pi i } \oint_0 \frac{F(z)}{z}\left[ \int_{A_1^\Delta} \int_{A_2^\Delta} \det[K(y_1,y_2)]_{i,j=1}^2\,dy_1\,dy_2 + \bigO(\Delta^3)\right]\,dz,
\end{multline} 
from which it immediately follows
\begin{multline}
\lim_{\Delta \to 0} \frac{\Prob_n(\text{there is at least one particle in each $[x_i, x_i + \Delta x)$, $i = 1, 2$})}{\Delta^2} = \\
\frac{1}{2\pi i }\oint_0 \frac{F(z)}{z}  \det[K(x_1,x_2)]_{i,j=1}^2\,dz.
\end{multline} 
This proves \eqref{eq:formula_for_R^m} in the case $m=2$ and $x_1 \ne x_2$. The extension to the general case is straightforward.

\section{Proof of Theorem \ref{thm:edge}} \label{sec:proof_edge}

Our starting point is formula \eqref{eq:contour_int_formula_right_most_particle2}, the special case of \eqref{eq:general_gap_prob} with $A=(-\infty, s]$. After the change of variable
\begin{equation} \label{eq:w=q^nz}
  w = q^n z,
\end{equation}
formula \eqref{eq:contour_int_formula_right_most_particle2} becomes
\begin{equation} \label{eq:integral_rightmost_in_w}
\Prob_n(\max(x_1, \dotsc, x_n) \leq s)= q^{n(n + 1)/2} (q; q)_n \frac{1}{2\pi i} \oint_0 (-q^{-n}w; q)_{\infty} \det(I -  \Proj_s\mathbf{K}(q^{-n} w; q) \Proj_s) \frac{dw}{w^{n + 1}},
\end{equation}
where $\Proj_s$ is the projection onto $L^2(s,\infty)$. It is straightforward to see that
\begin{equation}
  (-q^{-n}w; q)_{\infty} = (-w; q)_{\infty} w^n q^{-n(n + 1)/2} (-q/w; q)_n.
\end{equation}
Thus we have that the integral in \eqref{eq:integral_rightmost_in_w} can be written as
\begin{equation}
  \frac{1}{2\pi i} \oint_0(q; q)_n  (-w; q)_{\infty} (-q/w; q)_n \det(I - {\bf P}_s \mathbf{K}(q^{-n} w; q) {\bf P}_s) \frac{dw}{w}.
\end{equation}
By the triple product identity \cite[Theorem 10.4.1]{Andrews-Askey-Roy99}
\begin{equation}\label{eq:triple_product}
  (-w; q)_{\infty} (-q/w; q)_{\infty} (q; q)_{\infty} = \sum^{\infty}_{k = -\infty}  q^{\frac{k(k - 1)}{2}} w^k, %= \sum^{\infty}_{k=-\infty} (\sqrt{q})^{k^2} \left( \frac{w}{\sqrt{q}} \right)^k 
  %= \vartheta \left( \frac{\log(w/\sqrt{q})}{2\pi i}; \sqrt{q} \right),
\end{equation}
the integral in \eqref{eq:integral_rightmost_in_w} is written as
\begin{equation} \label{eq:final_integral_rightmost_in_w}
  \frac{(q; q)_n}{(q; q)_{\infty}} \frac{1}{2\pi i} \oint_0 \left( \sum^{\infty}_{k = -\infty}  q^{\frac{k(k - 1)}{2}} w^k \right) \frac{(-q/w; q)_n}{(-q/w; q)_{\infty}} \det(I -  {\bf P}_s\mathbf{K}(q^{-n} w; q){\bf P}_s) \frac{dw}{w}.
\end{equation}
We take the contour in \eqref{eq:final_integral_rightmost_in_w} as $\lvert w \rvert = \sqrt{q}$ and make the change of variable $w=\sqrt{q}e^{i \pi \theta}$. Then \eqref{eq:final_integral_rightmost_in_w} becomes
\begin{equation}\label{eq:asy_formula}
  \frac{1}{2}\int_{-1}^{1} \left( \sum^{\infty}_{k = -\infty} q^{k^2/2} e^{ik\pi\theta} \right) \det(I -   \Proj_s\mathbf{K}(q^{-n+1/2} e^{i\pi \theta}; q) \Proj_s) F_n(\theta; q) d\theta,
\end{equation}
where
\begin{equation} \label{eq:defn_F_n(theta;q)}
  F_n(\theta; q) = \frac{(q; q)_n}{(q; q)_{\infty}} \frac{(-\sqrt{q}e^{-i\pi \theta}; q)_n}{(-\sqrt{q}e^{-i\pi \theta}; q)_{\infty}}.
\end{equation}

\subsection{Preliminary estimates of $\tilde{K}_n(x, y)$}

In what follows we will need to compute the limit of the Fredholm determinant in the integrand of \eqref{eq:asy_formula} as $n\to\infty$ in the scaling limit $s = s_n =  2\sqrt{n}+tn^{-1/6}$ for $t\in \realR$. In this scaling
\begin{equation} \label{eq:tilde_K_kernel}
  \det(I - \Proj_s\K(q^{-n+1/2} e^{i\pi \theta}; q) \Proj_s) = \det(I - \Proj_t \tilde{\K} \Proj_t),
\end{equation}
where $\tilde{\K} = \tilde{\K}(\theta)$ has the kernel
\begin{equation}\label{eq:scaled_kernel1}
  \begin{split}
    \tilde{K}_n(x, y) = \tilde{K}_n(x,y; \theta) := {}& n^{-1/6} K(2\sqrt{n}+xn^{-1/6}, 2\sqrt{n}+yn^{-1/6}) \\
    = {}& n^{-1/6}\sum_{k=0}^\infty c_k \varphi_k(2\sqrt{n}+xn^{-1/6})\varphi_k(2\sqrt{n}+yn^{-1/6}),
  \end{split}
\end{equation}
where 
\begin{equation}\label{eq:cks}
  c_k = c_k(\theta) :=  \frac{e^{\pi i \theta}q^{k-n+1/2}}{1+e^{\pi i \theta}q^{k-n+1/2}}=\frac{e^{\pi i \theta/2}\sqrt{q}^{k-n+1/2}}{2\cosh\left(\frac{k-n+1/2}{2}\log q+\frac{ i\pi \theta}{2}\right)},
\end{equation}
with the dependence on $\theta$ suppressed if there is no chance of confusion. 

We need to compute the $n \to \infty$ limit of $\tilde{K}_n(x, y)$ for $x, y$ in a compact subset of $\realR$, and show that $\tilde{K}_n(x, y)$ vanishes exponentially fast as $\max(x, y) \to +\infty$ and $\min(x, y)$ is bounded below. We will use the following global approximation formula for $\varphi_k$, which is from \cite[Section 11.4, Exercises 4.2 and 4.3]{Olver97}. For $x$ in a compact subset of $(-1, +\infty)$ and $k = 0, 1, 2, \dotsc$ uniformly,
\begin{multline} \label{eq:refined_global_est}
  (k + 1/2)^{1/12} \varphi_k(2\sqrt{k + 1/2} x) = 2^{1/6} \left( \frac{\zeta(x)}{x^2 - 1} \right)^{1/4} \left( \Ai \left( (2k + 1)^{2/3} \zeta(x) \right) + \varepsilon_k(x) \right) \\
  \times (1 + \bigO((k + 1/2)^{-1})),
\end{multline}
such that
\begin{enumerate}[label=(\roman*)]
\item 
  the $1 + \bigO((k + 1/2)^{-1})$ factor depends on $k$ only;
\item 
  $\zeta(x)$ is a continuous, differentiable and monotonically increasing function on $(-1, +\infty)$. Moreover, it is bounded below as $x \to -1_+$ and has $x^2$ growth as $x \to +\infty$. The explicit formula of $\zeta(x)$ is given in \cite[Section 11.4, Exercise 4.2]{Olver97}. Around $1$, it satisfies
  \begin{equation} \label{eq:local_prop_zeta}
    \zeta(1) = 0 \quad \text{and} \quad \zeta'(1) = 2^{1/3};
  \end{equation}
\item
  $\varepsilon_n(x)$ is defined in \cite[Section 11.4, Exercise 4.2]{Olver97}, where it is denoted as $\varepsilon(x)$. From \cite[Section 11.2]{Olver97}, we have the estimate uniform in $k$ and $x$,
  \begin{equation}
    \lvert \varepsilon(x) \rvert =
    \begin{cases}
      \bigO \left( (k + 1/2)^{-7/6} (-\zeta(x))^{-1/4} \right) & \text{if $x \in (-1, 1]$}, \\
      \bigO \left( (k + 1/2)^{-7/6} \zeta(x)^{-1/4} \exp \left( -\frac{2}{3}(2k + 1) \zeta(x)^{3/2} \right) \right) & \text{if $x \in [1, +\infty)$}.
    \end{cases}
  \end{equation}
\end{enumerate}
To use estimate \eqref{eq:refined_global_est}, we also need that by \cite[Sections 11.1-2]{Olver97}, especially \cite[Formulas (2.05), (2.13) and (2.15) in Chapter 11]{Olver97},
\begin{equation} \label{eq:asy_Airy_func}
  \lvert \Ai(x) \rvert \leq f(x), \quad \text{where} \quad f(x) =
  \begin{cases}
    \frac{1}{2} \pi^{-1/2} x^{-1/4} e^{-\frac{2}{3} x^{3/2}} & x > 1, \\
    1 & -1 \leq x \leq 1, \\
    \lambda^{1/2} \pi^{-1/2} (-x)^{-1/4} & x < -1,
  \end{cases}
\end{equation}
with the constant $\lambda = 1.04\dots$. 

Below we provide computational results that are used in the proof of both part \ref{thm:edge_fixedq} and part \ref{thm:edge_varyingq} of Theorem \ref{thm:edge}. Note that we use $C$ to denote a large enough positive constant and $\gamma$ a small enough positive constant. It is harmless to assume $C = 1000$ and $\gamma = 1/10$.

\paragraph{$x, y$ in a compact subset.}

First we consider the case that $x, y \in [-M/2, M/2]$ where $M$ is a positive constant. Without loss of generality, we assume that $Mn^{1/3}$ is an integer, and then write
\begin{equation}\label{eq:three_kernels}
  \tilde{K}_n(x,y)=  K^{(1, M)}_n(x,y)+K^{(2, M)}_n(x,y)+K^{(3, M)}_n(x,y),
\end{equation}
where
\begin{align}
  K^{(1, M)}_n(x,y) &= n^{-1/6}\sum_{k=0}^{n-Mn^{1/3}-1} c_k \varphi_k(2\sqrt{n}+xn^{-1/6})\varphi_k(2\sqrt{n}+yn^{-1/6}), \label{eq:three_kernels2:1} \\
  K^{(2, M)}_n(x,y) &= n^{-1/6}\sum_{k=n-Mn^{1/3}}^{n+Mn^{1/3}} c_k \varphi_k(2\sqrt{n}+xn^{-1/6})\varphi_k(2\sqrt{n}+yn^{-1/6}), \label{eq:three_kernels2:2} \\
  K^{(3, M)}_n(x,y) &= n^{-1/6}\sum_{k=n+Mn^{1/3}+1}^{\infty} c_k \varphi_k(2\sqrt{n}+xn^{-1/6})\varphi_k(2\sqrt{n}+yn^{-1/6}). \label{eq:three_kernels2:3}
\end{align}
The following estimates on the coefficients $c_k$ are uniform in $k$ and $\theta$:
\begin{equation} \label{eq:est_c_k}
  c_k = 
  \begin{cases}
    1+\bigO(q^{Mn^{1/3}}),  & k=0, \dots, n-Mn^{1/3} -1, \\
    \bigO(q^{l + Mn^{1/3}}),  & k=n+Mn^{1/3} +l, \quad l=1, 2, \dots.
  \end{cases}
\end{equation}
With the estimates \eqref{eq:est_c_k} for $c_k$ and \eqref{eq:refined_global_est} for $\varphi_k$, it follows that
\begin{equation} \label{eq:K3_est}
  \lvert K^{(3, M)}_n(x,y) \rvert \leq n^{-1/6} C \sum^{\infty}_{l = 1} q^{l + Mn^{1/3}} (n + Mn^{1/3} + l + 1/2)^{-1/6} \leq C \frac{n^{-1/3}}{1 - q} q^{Mn^{1/3}},
\end{equation}
where $C$ is a constant independent of $n, x, y, M, \theta$ and $q$. Similarly,
\begin{multline} \label{eq:inter_est_K^1}
  \lvert K^{(1, M)}_n(x,y) \rvert \leq n^{-1/6} C \sum^{n-Mn^{1/3}-1}_{k=0} (k + 1/2)^{-1/6} \exp \left( -\frac{2}{3}(2k + 1) \zeta \left( \frac{2\sqrt{n} + x n^{-1/3}}{2\sqrt{k + 1/2}} \right)^{3/2} \right) \\
  \times \exp \left( -\frac{2}{3}(2k + 1) \zeta \left( \frac{2\sqrt{n} + y n^{-1/3}}{2\sqrt{k + 1/2}} \right)^{3/2} \right),
\end{multline}
where $C$ is independent of $n, x, y, M, \theta$ and $q$. After some calculation, the sum $K^{(1, M)}_n(x,y)$ is estimated as
\begin{equation}\label{eq:K1_est}
  \lvert K^{(1, M)}_n(x,y) \rvert \leq C \exp(-\gamma M^{3/2}),
\end{equation}
where $C$ and $\gamma$ are independent of $n, x, y, M, \theta$ and $q$.

The approximation of $K^{(2, M)}_n(x, y)$ depends on $\theta$ and will be given later.

\paragraph{$x \to +\infty$ and $y$ is bounded below.}

Let $M$ be the same as above and $N > M$ be a large positive constant, and without loss of generality assume that $N n^{1/3}$ is an integer. Suppose $x \geq 2N$ and $y \geq -M/2$. We write
\begin{equation}
  \tilde{K}_n(x,y) = K^{(4, M, N)}_n(x,y) + K^{(2, M)}_n(x,y) + K^{(5, N)}_n(x,y),
\end{equation}
where $K^{(2, M)}_n(x,y)$ is defined in \eqref{eq:three_kernels2:2}, and
\begin{align}
  K^{(4, M, N)}_n(x,y) = {}& n^{-1/6} \sum_{\substack{0 \leq k \leq n - Mn^{1/3} - 1 \\ \text{or } n + Mn^{1/3} + 1 \leq k \leq n + N n^{1/3}}} c_k \varphi_k(2\sqrt{n} + xn^{-1/6}) \varphi_k(2\sqrt{n} + yn^{-1/6}), \\
  K^{(5, N)}_n(x,y) = {}& n^{-1/6} \sum^{\infty}_{k = n + Nn^{1/3} + 1} c_k \varphi_k(2\sqrt{n} + xn^{-1/6}) \varphi_k(2\sqrt{n} + yn^{-1/6}).
\end{align}
Similar to \eqref{eq:K3_est}, we have the estimate 
\begin{equation}\label{eq:est_K^(5,N)}
  \lvert K^{(5, N)}_n(x,y) \rvert \leq C \frac{n^{-1/3}}{1 - q} q^{N n^{1/3}},
\end{equation}
where $C$ is independent of $n, x, y, N, \theta$ and $q$. Similarly, like \eqref{eq:K3_est} and \eqref{eq:inter_est_K^1}, by the estimate \eqref{eq:est_c_k} for $c_k$ and \eqref{eq:refined_global_est} for $\varphi_k$,
\begin{equation} \label{eq:est_K^(4,M,N)}
  \begin{split}
    & \lvert K^{(4, M, N)}_n(x,y) \rvert \\
    \leq {}& n^{-1/6} C \sum_{\substack{0 \leq k \leq n - Mn^{1/3} - 1 \\ \text{or } n + Mn^{1/3} + 1 \leq k \leq n + N n^{1/3}}} (k + 1/2)^{-1/6} \exp \left( -\frac{2}{3}(2k + 1) \zeta \left( \frac{2\sqrt{n} + x n^{-1/3}}{2\sqrt{k + 1/2}} \right)^{3/2} \right) \\
    \leq {}& C \exp(-\gamma N^{3/2}),
  \end{split}
\end{equation}
where $C$ and $\gamma$ are independent of $n, x, y, M, N, \theta$ and $q$. Note that in \eqref{eq:est_K^(4,M,N)} the estimate of $\varphi_k(2\sqrt{n} + xn^{-1/6})$ is the same as in \eqref{eq:inter_est_K^1}, while the estimate of $\varphi_k(2\sqrt{n} + yn^{-1/6})$ is roughly $\bigO((k + 1/2)^{-1/12})$, as in \eqref{eq:K3_est}.

Below we prove Theorem \ref{thm:edge}. We first give full detail for part \ref{thm:edge_varyingq}, and then show that a simplified argument works for part \ref{thm:edge_fixedq}. The technical core is the estimate of $K^{(2, M)}_n(x, y)$.

% We will also use the following estimate, which is obtained after some calculation from \eqref{eq:refined_global_est}:
% \begin{multline}
% \varphi_{n+k}(2\sqrt{n}+xn^{-1/6})\varphi_{n+k}(2\sqrt{n}+yn^{-1/6})-\varphi_{n-k-1}(2\sqrt{n}+xn^{-1/6})\varphi_{n-k-1}(2\sqrt{n}+yn^{-1/6}) = \\
% n^{-1/6}\left[(\Ai(x)\Ai'(y)+\Ai'(x)\Ai(y))\frac{2k+1}{n^{1/3}}+\bigO(n^{-2/3})+\bigO\left(\frac{k}{n}\right)\right]
% \end{multline}

\subsection{Gap probability for the rightmost particle: $q = e^{-c n^{-1/3}}$}

Now consider the scaling $q= e^{-c n^{-1/3}}$ for some $c>0$. We begin with the following lemma on the asymptotics of the $q$-Pochhammer symbols appearing in \eqref{eq:asy_formula}.
\begin{lem} \label{lem:Pochhammer}
  For $q=e^{-c n^{-1/3}}$, we have the estimate uniformly for $\theta\in [-1,1]$:
  \begin{equation}\label{eq:qPoch_asymptotics}
    \frac{(q; q)_n}{(q; q)_{\infty}} = 1+\bigO(n^{1/3} e^{-cn^{2/3}}), \qquad 
    \frac{(-\sqrt{q} e^{-\pi i \theta}; q)_n}{(-\sqrt{q} e^{-\pi i \theta}; q)_{\infty}}= 1+\bigO(n^{1/3} e^{-cn^{2/3}}).
  \end{equation}
Thus uniformly for $\theta\in [-1,1]$, the function $F_n(\theta; q)$ defined in \eqref{eq:defn_F_n(theta;q)} satisfies
  \begin{equation}
    F_n(\theta; q) = 1 + \bigO(n^{1/3} e^{-cn^{2/3}}).
  \end{equation}
\end{lem}
\begin{proof}
  We only prove the second equation of \eqref{eq:qPoch_asymptotics}. We have
  \begin{equation}
    \frac{(\sqrt{q} e^{-\pi i \theta}; q)_n}{(\sqrt{q} e^{-\pi i \theta}; q)_{\infty}} = \frac{1}{\prod_{k=0}^\infty (1 -  e^{-\pi i \theta} q^{k+n + 1/2})},
  \end{equation}
  thus
  \begin{equation}
    \begin{split}
      \left\lvert \log\frac{(\sqrt{q} e^{-\pi i \theta}; q)_n}{(\sqrt{q} e^{-\pi i \theta}; q)_{\infty}} \right\rvert \leq {}& \sum_{k=0}^\infty \lvert \log(1 -  e^{-\pi i \theta} q^{k+n + 1/2}) \rvert \\
      < {}& \frac{2q^{n}}{1-q} = \frac{2e^{-cn^{2/3}}}{1-e^{-cn^{-1/3}}}=\frac{2e^{-cn^{2/3}}}{cn^{-1/3}}(1+\bigO(n^{-1/3})).
    \end{split}
  \end{equation}
  The result is obtained by exponentiating.
\end{proof}

Also note that the Poisson summation formula gives 
\begin{equation}\label{eq:Pos_sum}
\begin{aligned}
\sum^{\infty}_{k = -\infty} e^{\frac{-cn^{-1/3}k^2}{2}} e^{k\pi i\theta}  &= n^{1/6} \sqrt{2\pi c^{-1}} \sum_{k=-\infty}^\infty e^{-\frac{n^{1/3}\pi^2(2k- \theta)^2}{2c}}. \\
%&= n^{1/6} \sqrt{2\pi c^{-1}} e^{-\frac{n^{1/3}\pi^2 \theta^2}{2c}}\left(1+\bigO(e^{-\frac{n^{1/3}\pi^2}{2c}})\right),
\end{aligned}
\end{equation}
Applying formulas \eqref{eq:tilde_K_kernel}, \eqref{eq:qPoch_asymptotics} and \eqref{eq:Pos_sum} to the integral \eqref{eq:asy_formula}, we find that \eqref{eq:asy_formula} becomes
\begin{equation}\label{eq:cross1}
  n^{1/6} \int_{-1}^{1} \sqrt{\frac{\pi}{2c}}\left( \sum_{k=-\infty}^\infty e^{-\frac{n^{1/3}\pi^2(2k- \theta)^2}{2c}}\right)\det(I - \Proj_t \tilde{\K}(\theta) \Proj_t) F_n(\theta; q) d\theta.
\end{equation}

%uniformly for $\theta$ in compact subsets of $(-1/2, 1/2)$. 
Fix a small $\ep>0$. We plug the formula \eqref{eq:Pos_sum} into \eqref{eq:cross1}, use the estimates in Lemma \ref{lem:Pochhammer}, and split the integral \eqref{eq:cross1} into two parts, $I_1$ and $I_2$, where
\begin{align} % \label{eq:I1andI2}
  I_1 = n^{1/6} \int_{-1+\ep}^{1-\ep} \sqrt{\frac{\pi}{2c}}\left( \sum_{k=-\infty}^\infty e^{-\frac{n^{1/3}\pi^2(2k- \theta)^2}{2c}}\right) {}& \det(I - \Proj_t \tilde{\K}(\theta) \Proj_t) F_n(\theta; q) d\theta, \\
  I_2 = n^{1/6} \int_{1-\ep}^{1+\ep} \sqrt{\frac{\pi}{2c}}\left( \sum_{k=-\infty}^\infty e^{-\frac{n^{1/3}\pi^2(2k- \theta)^2}{2c}}\right) {}& \det(I - \Proj_t \tilde{\K}(\theta) \Proj_t) F_n(\theta; q) d\theta.\label{I2_def}
\end{align}

In order to evaluate these integrals as $n\to\infty$, we need some estimates on the determinant $\det(I - \Proj_t \tilde{\K}(\theta) \Proj_t)$ which are uniform in $\theta$. These are given in the following lemma.
\begin{lem}\label{det_estimates}

\begin{enumerate}[label=(\alph*)]
  \item \label{det_estimates_I1}
For $\theta\in (-1+\epsilon, 1-\epsilon)$, the determinant $\det(I - \Proj_t \tilde{\K}(\theta) \Proj_t)$ is bounded uniformly in $\theta$ as $n\to\infty$. Furthermore it has the limit
 \begin{equation} \label{eq:det_conv_edge}
\lim_{n\to\infty}\det(I - \Proj_t \tilde{\K}(\theta) \Proj_t)= \det(I - \mathbf{P}_t \mathbf{K}_{\crossover}(c; \theta) \mathbf{P}_t), 
\end{equation}
where $\mathbf{K}_{\crossover}(c; \theta)$ is the integral operator on $L^2(\realR)$ with kernel
\begin{equation}
K_{\crossover}(x,y;c; \theta) = \int_{-\infty}^\infty \frac{e^{i\pi \theta} e^{-cr}}{1+e^{i \pi \theta} e^{-cr}} \Ai(x - r)\Ai(y-r)\,dr.
\end{equation}
  \item\label{det_estimates_I2}
There exist positive constants $\tilde{C}$ such that for all $n\in \mathbb{N}$ and all $\theta\in (1-\epsilon, 1+\epsilon)$,
    \begin{equation} \label{eq:det_estimatesI2}
      \lvert\det(I - \Proj_t \tilde{\K}(\theta) \Proj_t)\rvert \le \exp\left( \left( \tilde{C} e^{-ct} \log n \right)^2 + \tilde{C} e^{-ct} \log n \right).
    \end{equation}
  \end{enumerate}
\end{lem}
Given the results of this lemma, it is fairly straightforward to prove Theorem \ref{thm:edge}\ref{thm:edge_varyingq}.
Consider $I_1$ first. Clearly as $n\to\infty$ the dominant term in the infinite sum is $k=0$, and we have
\begin{equation}\label{eq:I1}
  I_1 = n^{1/6} \sqrt{\frac{\pi}{2c}}\int_{-1+\ep}^{1-\ep}e^{-\frac{n^{1/3}\pi^2 \theta^2}{2c}}\det(I - \Proj_t \tilde{\K}(\theta) \Proj_t)  \left(1+\bigO(e^{-\frac{n^{1/3}\pi^2}{2c}})\right)d\theta\, .
\end{equation}
Since the Fredholm determinant in the integrand has a limit as $n\to \infty$, we can use Laplace's method to evaluate the integral as $n\to\infty$. The integral $I_1$ is localized close to $\theta=0$, and Laplace's method immediately gives
\begin{equation}
\begin{aligned}
\lim_{n\to\infty} I_1 &= \lim_{n\to\infty}  \det(I -\Proj_t \tilde{\K}(\theta=0) \Proj_t)+o(1) \\
& = \det(I - \mathbf{P}_t \mathbf{K}_{\crossover}(c; \theta=0) \mathbf{P}_t).
\end{aligned}
\end{equation}
Noting that  $\mathbf{K}_{\crossover}(c; \theta=0) \equiv  \mathbf{K}_{\crossover}(c)$ defined in \eqref{eq:crossover_limit_result}, we find
\begin{equation}
\lim_{n\to\infty} I_1 =  \det(I - \mathbf{P}_t \mathbf{K}_{\crossover}(c) \mathbf{P}_t).
\end{equation}
It remains only to show that $\lim_{n\to\infty} I_2 =0$. This follows immediately from \eqref{I2_def} and \eqref{eq:det_estimatesI2}, since the infinite sum in \eqref{I2_def} is vanishing like the exponent of a power of $n$ whereas the determinant is growing at most as the exponent of a power of $\log n$. This completes the proof of Theorem \ref{thm:edge}\ref{thm:edge_varyingq}, provided that Lemma \ref{det_estimates} is true. The remainder of this subsection is dedicated to the proof of this lemma.

\paragraph{Proof of Lemma \ref{det_estimates}\ref{det_estimates_I1}}

We use the expression for the kernel $\tilde{K}_n(x,y)$ in \eqref{eq:three_kernels} and \eqref{eq:three_kernels2:1}--\eqref{eq:three_kernels2:3} for the pointwise approximation as $x, y$ in a compact subset of $\realR$. In the scaling $q=e^{-cn^{-1/3}}$, the estimate \eqref{eq:K3_est}  becomes
\begin{equation}\label{eq:K3_est2}
  \lvert K^{(3, M)}_n(x,y) \rvert \leq Cc^{-1} e^{-cM}.
\end{equation}
Combined with \eqref{eq:K1_est} we see that as $M$ becomes large, $K^{(1, M)}_n(x, y)$ and $K^{(3, M)}_n(x, y)$ vanish, and the dominant contribution should come from $K^{(2, M)}_n(x,y)$.
In the sum $K^{(2, M)}_n(x,y)$, we denote $r=n^{-1/3}(k-n)$ and write the sum as
\begin{multline}\label{eq:K2_int}
  K^{(2, M)}_n(x,y) = \\
  n^{-1/6}\sum_{r\in \{n^{-1/3} \Z \cap [-M,M]\}} \frac{e^{\pi i \theta}q^{1/2+rn^{1/3}}}{1+e^{\pi i \theta}q^{1/2+rn^{1/3}}} \varphi_{n + r n^{1/3}}(2\sqrt{n}+xn^{-1/6})\varphi_{n + r n^{1/3}}(2\sqrt{n}+yn^{-1/6}) \\
  = \int_{-M}^M \frac{e^{\pi i \theta}q^{1/2+\lfloor rn^{1/3}\rfloor}}{1+e^{\pi i \theta}q^{1/2+\lfloor rn^{1/3}\rfloor}}n^{1/12} \varphi_{n + \lfloor rn^{1/3}\rfloor}(2\sqrt{n}+xn^{-1/6})n^{1/12}\varphi_{n + \lfloor rn^{1/3}\rfloor}(2\sqrt{n}+yn^{-1/6})\,dr.
\end{multline}
From \eqref{eq:refined_global_est}, we find that
\begin{equation}\label{eq:Airy_asymptotics}
  \lim_{n\to\infty}n^{1/12} \varphi_{n + \lfloor rn^{1/3} \rfloor} (2\sqrt{n} + n^{-1/6} x) =  \Ai(x - r),
\end{equation}
thus the integrand in \eqref{eq:K2_int} has the pointwise limit 
\begin{equation}
\frac{e^{i\pi \theta} e^{-cr}}{1+e^{i \pi \theta} e^{-cr}} \Ai(x - r)\Ai(y-r),
\end{equation}
and the bounded convergence theorem gives
\begin{equation}
\lim_{n\to\infty} K^{(2, M)}_n(x,y) = \int_{-M}^M \frac{e^{i\pi \theta} e^{-cr}}{1+e^{i \pi \theta} e^{-cr}} \Ai(x - r)\Ai(y-r)\,dr.
\end{equation}
Since both $K_n^{(1,M)}(x,y)$ and $K_n^{(3,M)}(x,y)$ are bounded in $n$ and vanish as $M\to\infty$, we now take $M\to\infty$ and obtain
\begin{equation}\label{eq:kernel_convergence}
  \lim_{n\to\infty} \tilde{K}_n(x,y) =\lim_{M\to\infty} \lim_{n\to\infty} K^{(2, M)}_n(x,y) = \int_{-\infty}^\infty \frac{e^{i\pi \theta} e^{-cr}}{1+e^{i \pi \theta} e^{-cr}} \Ai(x - r)\Ai(y-r)\,dr,
\end{equation}
which is the kernel of a trace-class operator for all $\theta \in (-1+\ep,1-\ep)$. 

We have proved the pointwise convergence of the kernels in the determinant, and actually the convergence in \eqref{eq:kernel_convergence} is uniform if $x, y$ are in a compact subset of $\realR$. To prove the determinant convergence \eqref{eq:det_conv_edge}, we will need estimates on the kernel $\tilde{K}_n(x,y)$ as $\max(x,y)\to\infty$. The estimates \eqref{eq:est_K^(5,N)} and \eqref{eq:est_K^(4,M,N)} imply that, if $q=e^{-cn^{-1/3}}$ and $y \geq t$, then for all $x > 4 \max(-t, 1)$, we take $M = 2 \max(-t, 1)$ and $N = x$, and have
\begin{equation}\label{eq:K^4K^5largex}
  \lvert K^{(4, 2 \max(-t, 1), x)}_n(x,y) \rvert  \leq C e^{-\gamma x^{3/2}} \quad \text{and} \quad \lvert K^{(5, x)}_n(x,y) \rvert \leq Cc^{-1} e^{-c x},
\end{equation}
with constants $C$ and $\gamma$ independent of $n$. Using the method of estimating $K^{(4, M, N)}_n(x, y)$ in \eqref{eq:est_K^(4,M,N)}, we have a similar estimate for $K^{(2, 2 \max(-t, 1))}_n(x,y)$, provided that $\theta\in (-1+\ep, 1-\ep)$:
\begin{equation}\label{K^2largex}
  \lvert K^{(2, 2 \max(-t, 1))}_n(x,y) \rvert  \leq C e^{-\gamma x^{3/2}}.
\end{equation}
Combining \eqref{eq:K^4K^5largex} and \eqref{K^2largex} we obtain the uniform estimate for $x, y \geq t$
\begin{equation}\label{Ktilde_largex}
  \lvert  \tilde{K}_n(x,y) \rvert  \leq \tilde{C} e^{-c x},
\end{equation}
where the constant $\tilde{C}$ depends on $t$ and $c$, but independent of $n$.

The Fredholm determinant $\det(I -   {\bf P}_s\mathbf{K}(q^{-n+1/2} e^{i\pi \theta}; q) {\bf P}_s)=\det(I-\Proj_t \tilde{\K} \Proj_t)$ is given by the series
\begin{equation}\label{eq:det_series}
  \det(I-\Proj_t \tilde{\K} \Proj_t) = 1+\sum_{k=1}^\infty \frac{(-1)^k}{k!} \int_t^\infty dx_1 \cdots \int_t^\infty dx_k \det(\tilde{K}_n(x_i, x_j))_{i,j=1}^k.
\end{equation} 
Each of the determinants in this series can be estimated using \eqref{Ktilde_largex} along with Hadamard's inequality, giving
\begin{equation}\label{Hadamard_est}
  \lvert \det(\tilde{K}(x_i, x_j))_{i,j=1}^k \rvert \le k^{k/2} \tilde{C}^k \prod_{i=1}^k e^{-cx_i},
\end{equation}
so each term in \eqref{eq:det_series} is bounded by
\begin{equation}
  \begin{aligned}
    \left\lvert \frac{(-1)^k}{k!} \int_t^\infty dx_1 \cdots \int_t^\infty dx_k \det(\tilde{K}_n(x_i, x_j))_{i,j=1}^k\right\rvert &\le \frac{k^{k/2}}{k!} \tilde{C}^k \int_t^\infty dx_1 e^{-cx_1} \cdots \int_t^\infty dx_k e^{-cx_k} \\
    &\le \frac{k^{-k/2}}{k!} \left( \tilde{C} c^{-1} e^{-ct} \right)^k.
  \end{aligned}
\end{equation}
Thus the series \eqref{eq:det_series} is dominated by an absolutely convergent series, and the dominated convergence theorem gives that the sum converges to the term-by-term limit. This is exactly $\det(I - \mathbf{P}_t \mathbf{K}_{\crossover}(c; \theta) \mathbf{P}_t)$, since the integrands are dominated by an absolutely integrable function according to \eqref{Hadamard_est}, so the dominated convergence theorem implies that each term converges to the corresponding term in the series for  $\det(I - \mathbf{P}_t \mathbf{K}_{\crossover}(c; \theta) \mathbf{P}_t)$.  This completes the proof of Lemma \ref{det_estimates}\ref{det_estimates_I1}. 

\paragraph{Proof of Lemma \ref{det_estimates}\ref{det_estimates_I2}}

Our estimate of $\det(I-\Proj_t \tilde{\K}(\theta) \Proj_t)$ for $\theta$ close to $1$ is based on the identity (see \cite[Theorem 9.2(d)]{Simon05})
\begin{equation}\label{det2_identity}
  \det(I-\Proj_t \tilde{\K} \Proj_t) = \dettwo (I-\Proj_t \tilde{\K} \Proj_t)e^{\Tr \Proj_t \tilde{\K} \Proj_t},
\end{equation}
where $\dettwo$ is defined in \cite[Chapter 9]{Simon05}. The $\dettwo$ functional can be estimated using the Hilbert-Schmidt norm, see \cite[Theorem 9.2(b)]{Simon05}. In particular we have
\begin{equation}
  \lvert \dettwo (I-\Proj_t \tilde{\K} \Proj_t) \rvert \le \exp(\lVert \Proj_t \tilde{\K} \Proj_t \rVert_2^2),
\end{equation}
where $\lVert \cdot \rVert_2$ represents the Hilbert-Schmidt norm. Combining this inequality with \eqref{det2_identity}, we have
\begin{equation}\label{trace_HS_estimate}
  \lvert \det(I-\Proj_t \tilde{\K} \Proj_t)\rvert \le \exp(\lVert \Proj_t \tilde{\K} \Proj_t \rVert_2^2) e^{\lvert \Tr \Proj_t \tilde{\K} \Proj_t \rvert},
\end{equation}
and we are left to estimate the trace and the Hilbert-Schmidt norms of $\Proj_t \tilde{\K} \Proj_t$.

We begin by  estimating the kernel $\tilde{K}_n(x,y)$ for $\theta \in (1-\ep,1+\ep)$. 
Since \eqref{eq:K3_est2}, \eqref{eq:K1_est} and \eqref{eq:K^4K^5largex} still hold for $\theta \in (1 - \epsilon, 1 + \epsilon)$, we concentrate on $K^{(2, M)}_n(x,y)$. Let us estimate this sum.  Using \eqref{eq:refined_global_est} we obtain the following estimate, which is uniform for $x,y$ in compact sets and $n-Mn^{1/3} < k <n+Mn^{1/3}$:
\begin{multline}
\varphi_{k}(2\sqrt{n}+xn^{-1/6})\varphi_{k}(2\sqrt{n}+yn^{-1/6}) = \\
 n^{-1/6}\Ai(x-(k-n)/(2n^{1/3}))\Ai(y-(k-n)/(2n^{1/3}))(1+\bigO(n^{-2/3})).
\end{multline}
The kernel $K^{(2, M)}_n(x,y)$ is thus estimated as 
\begin{equation}\label{eq:K2_I2}
  \lvert K^{(2, M)}_n(x,y) \rvert \le Cn^{-1/3} \sum_{k=n-Mn^{1/3}}^{n+Mn^{1/3}} \left\lvert c_k(\theta) \Ai(x-(k-n)/(2n^{1/3}))\Ai(y-(k-n)/(2n^{1/3}))\right\rvert
\end{equation}
for some constant $C$ which is independent of $n, M$ and $\theta$. We therefore need to estimate the coefficients $c_k$, and it is convenient to estimate the real and imaginary parts separately. They are
  \begin{equation}
  \Re c_{n+j}(\theta) =  \frac{\cos(\pi\theta)+q^{j+1/2}}{2\cos(\pi\theta)+q^{j+1/2}+q^{-j-1/2}}, \qquad   \Im c_{n+j}(\theta) =  \frac{\sin(\pi\theta)}{2\cos(\pi\theta)+q^{j+1/2}+q^{-j-1/2}}.
\end{equation}
To estimate the imaginary part, note that $\Im c_{n+j}(1) = 0$, but $c_{n+j}$ becomes large in a neighborhood of $\theta=1$ when $q$ is close to $1$. In this neighborhood the critical points of $\Im c_{n+j}(\theta) $ are found to be at 
  \begin{equation}
    \theta = 1 \pm \arcsin \left( \frac{q^{-j - 1/2} - q^{j + 1/2}}{q^{-j - 1/2} + q^{j + 1/2}} \right),
  \end{equation}
where $\lvert \Im c_{n + j}(\theta) \rvert$ attains the maximum. Plugging these critical points into $\Im c_{n+j}(\theta)$ we find the maximum value of $|\Im c_{n+j}(\theta)|$, obtaining
  \begin{equation}\label{eq:im_ck_est}
    \lvert \Im c_{n+j}(\theta) \rvert = \frac{1}{\lvert 1 - q^{-2j - 1} \rvert}.
  \end{equation}
  Now consider the real part of $c_k$. The maximum value of $\lvert \Re c_k\rvert$ is attained at $\theta=1$. At this point we have
  \begin{equation}\label{eq:re_ck_est}
    \lvert \Re c_{n+j} \rvert = \frac{1}{\lvert 1 - q^{-j - 1/2} \rvert}.
  \end{equation}
  Combining \eqref{eq:im_ck_est} and \eqref{eq:re_ck_est} with \eqref{eq:K2_I2} we obtain the estimate for $x, y$ in a compact set and $n$ large enough
  \begin{equation}\label{eq:log_est_K2}
    \left\lvert K^{(2, M)}_n(x,y)\right\rvert \le
    C\sum_{k=n-Mn^{1/3}}^{n+Mn^{1/3}} \frac{ \lvert \Ai(x-(k-n)/(2n^{1/3}))\Ai(y-(k-n)/(2n^{1/3}))) \rvert}{2c \lvert k-n \rvert +1} = \tilde{C} \log n,
  \end{equation}
where $\tilde{C}$ is a positive constant depending on $M, c$ but not $n, \theta$. Now consider the behavior of $\tilde{K}_n(x,y)$ as $x \to \infty$ when $\theta \in (1-\ep,1+\ep)$. The estimates \eqref{eq:K^4K^5largex} still hold here. The estimate  \eqref{K^2largex} needs to be modified slightly for $\theta \in (1-\ep,1+\ep)$. Since the dependence of $K^{(2, M)}_n(x,y)$ on $\theta$ comes entirely from the coefficients $c_k$, and the dependence on $x$ and $y$ comes entirely from the Hermite functions, we can combine the analysis leading to \eqref{eq:log_est_K2} with \eqref{K^2largex} to obtain the estimate
\begin{equation}\label{K^2largex_I2}
  \lvert K^{(2, 2 \max(-t, 1))}_n(x,y) \rvert  \leq C  e^{-\gamma x^{3/2}}\log n,
\end{equation}
for  $\theta \in (1-\ep,1+\ep)$, where once again $C$ and $\gamma$ are constants independent of $n$. Analogous to \eqref{Ktilde_largex}, we therefore have the uniform estimate for all $x, y \geq t$
\begin{equation}\label{Ktilde_largex_I2}
  \lvert \tilde{K}_n(x,y) \rvert  \le \tilde{C} e^{-cx}\log n,
\end{equation}
where $\tilde{C}$ depends on $t$ and $c$, but not $n, \theta$. The trace of $\Proj_t \tilde{\K} \Proj_t$ can thus be estimated as
\begin{equation}\label{Ktilde_trace}
  \lvert \Tr \Proj_t \tilde{\K} \Proj_t \rvert \le \int_t^\infty \lvert \tilde{K}_n(x,x)\rvert \,dx \le \tilde{C} \log n \int_t^\infty e^{-cx}\, dx =  \tilde{C}c^{-1} e^{-ct} \log n,
\end{equation}
and the Hilbert-Schmidt norm is estimated as 
\begin{equation}\label{Ktilde_HS}
  \begin{split}
    \lVert \Proj_t \tilde{\K} \Proj_t \rVert_2^2 = {}& \int_t^\infty\int_t^\infty \lvert \tilde{K}_n(x,y)\rvert^2 \,dy\,dx  \\
    \le {}& \tilde{C}^2(\log n)^2 \int_t^\infty \int_t^\infty e^{-cx}e^{-cy}\,dy\,dx = (\tilde{C}c^{-1} e^{-ct} \log n)^2.
  \end{split}
\end{equation}

Combining this with \eqref{trace_HS_estimate}, \eqref{Ktilde_trace}, and \eqref{Ktilde_HS} we obtain \eqref{eq:det_estimatesI2}. This completes the proof of Lemma \ref{det_estimates}\ref{det_estimates_I2} (with $\tilde{C}c^{-1}$ replaced by $\tilde{C}$).

\subsection{Gap probability for the rightmost particle: fixed $q \in (0, 1)$}

Let $q \in (0, 1)$ be fixed. Then we have the following lemma.
\begin{lem}\label{lem:Airy_convergence}
Let $s\equiv s_n = 2\sqrt{n}+tn^{-1/6}$. The following holds uniformly for all $\theta\in [-1,1]$.
\begin{equation}
 \det(I -  {\bf P}_{s} \mathbf{K}(q^{-n+1/2} e^{i\pi \theta}; q) {\bf P}_s) = \det(I-{\bf P}_t {\bf K}_{\rm Airy}{\bf P}_t) + o(1).
\end{equation}
\end{lem}
\begin{proof}[Sketch of proof]
  In the sum $K^{(2, M)}_n(x,y)$ given by \eqref{eq:K2_int}, formula \eqref{eq:Airy_asymptotics}
  implies that the piecewise constant function in the integrand of \eqref{eq:K2_int} has the pointwise limit
  \begin{equation}
    \Ai(x - r)\Ai(y-r) \chi_{[-M,0]}(r).
  \end{equation}
  The bounded convergence theorem then implies that 
  \begin{equation}
    \begin{aligned}
      \lim_{n\to\infty} K^{(2, M)}_n(x,y) &= \int_{-M}^0  \Ai(x - r)\Ai(y-r)\,dr \\
      &=\int_0^M  \Ai(x + r)\Ai(y+r)\,dr.
    \end{aligned}
  \end{equation}
  Since $M$ was arbitrary we can take it to infinity, in which case $K^{(1, M)}_n(x,y)$ and $K^{(3, M)}_n(x,y)$ vanish by \eqref{eq:K3_est} and \eqref{eq:K1_est}, leaving
  \begin{equation}
    \lim_{n\to\infty}  K_n(x,y) = \lim_{M\to\infty} \lim_{n\to\infty} K^{(2, M)}_n(x,y) 
    =\int_0^\infty  \Ai(x + r)\Ai(y+r)\,dr,
  \end{equation}
  which is the kernel of ${\bf K}_{\rm Airy}$.

  To prove the convergence of the Fredholm determinant, we need to control the vanishing of $\tilde{K}_n(x, y)$ as $\max(x, y) \to \infty$. Since the procedure is the same as the proof of Lemma \ref{det_estimates}\ref{det_estimates_I1}, we omit the detailed verification. We only note that the proof works for all $\theta \in [-1, 1]$, since the coefficients $c_k(\theta)$ are uniformly bounded even if $\theta$ is around $\pm 1$.
\end{proof}

As $n \to \infty$, we have the very fast convergence analogous to Lemma \ref{lem:Pochhammer}
\begin{equation} \label{eq:convergence_inessential}
  \frac{(-q/w; q)_n}{(-q/w; q)_{\infty}} = 1 + \bigO(q^n), \quad  \frac{(q; q)_n}{(q; q)_{\infty}} = 1 + \bigO(q^n).
\end{equation}
Combining this fact with Lemma \ref{lem:Airy_convergence}, we see that the integral \eqref{eq:asy_formula} is
\begin{equation}
  \Prob_n(\max(x_1, \dotsc, x_n) \leq s) = \frac{1}{2}\int_{-1}^{1} \left( \sum^{\infty}_{k = -\infty} q^{k^2/2} e^{ik\pi\theta} \right)  \det(I - \Proj_t \K_{\Airy} \Proj_t)(1+o(1)) d\theta.
\end{equation}
After integrating, the only nonvanishing term in the infinite sum is $k=0$, thus we find
\begin{equation}
  \begin{split}
    \Prob_n(\max(x_1, \dotsc, x_n) \leq s) = {}& \frac{1}{2}\int_{-1}^{1}  \det(I -   \Proj_s \K_{\Airy} \Proj_s)(1+o(1)) d\theta \\
    = {}& \det(I - \Proj_t \K_{\Airy} \Proj_t)(1+o(1)).
  \end{split}
\end{equation}
This proves part \ref{thm:edge_fixedq} of Theorem \ref{thm:edge}.

\section{Proof of Theorem \ref{thm:bulk}} \label{sec:proof_bulk}

As in the proof of Theorem \ref{thm:edge}, we give the detail in part \ref{enu:thm:bulk_b}, and then show that a simplified argument works for part \ref{enu:thm:bulk_a}. Also for notational simplicity we only consider the $2$-correlation function. The generalization to $m$-correlation function is straightforward.

% Like in Section \ref{sec:proof_edge}, we take the change of variable $w = q^n z$ as in \eqref{eq:w=q^nz}, and then the $2$-correlation function
% \begin{equation} \label{eq:bulk_2corr_in_w}
%   R^{(2)}(x, y) = q^{n(n + 1)/2} (q; q)_n \frac{1}{2\pi i} \oint_0 (-q^{-n}w; q)_{\infty}
%   \begin{vmatrix}
%     K(x, x; q^{-n}w; q) & K(x, y; q^{-n}w; q) \\
%     K(y, x; q^{-n}w; q) & K(y, y; q^{-n}w; q)
%   \end{vmatrix}
%   \frac{dw}{w^{n + 1}}.
% \end{equation}
% Taking the contour around $0$ as $\lvert w \rvert = \sqrt{q}$ in \eqref{eq:bulk_2corr_in_w}, we have that the $2$-correlation function becomes, like \eqref{eq:asy_formula},
% \begin{equation}
%   \frac{1}{2} \int^1_0 \left( \sum^{\infty}_{-\infty} q^{k^2/2} e^{ik\pi \theta} \right) 
%   \begin{vmatrix}
%     K(x, x; q^{-n + 1/2} e^{i\pi \theta}; q) & K(x, y; q^{-n + 1/2} e^{i\pi \theta}; q) \\
%     K(y, x; q^{-n + 1/2} e^{i\pi \theta}; q) & K(y, y; q^{-n + 1/2} e^{i\pi \theta}; q)
%   \end{vmatrix}
%   F_n(\theta; q) d\theta,
% \end{equation}
% where $F_n(\theta; q)$ is defined in \eqref{eq:defn_F_n(theta;q)}.

\subsection{Correlation functions for the bulk particles:  $q = e^{-c/n}$}

We assume the contour in \eqref{eq:formula_for_R^m} is
\begin{equation}
  \Gamma = \left\{ \lvert z \rvert = q^{-n} - 1 + \frac{\delta_n}{n} = e^{c} - 1 + \frac{\delta_n}{n} \right\}, 
\end{equation}
such that $\lvert \delta_n \rvert < 1$ and there exists $\epsilon(q) > 0$ independent of $n$ and
\begin{equation} \label{eq:z_abs_value}
  \lvert 1 - q^k(q^{-n} - 1 + \delta_n/n) \rvert > \epsilon(q)/n
\end{equation}
for all $k \geq 0$. The term $\delta_n$ in the definition of $\Gamma$ makes $\Gamma$ away from poles at $-q^{-k}$. For notational simplicity, we assume $\delta_n = 0$ later in this section. 

We compute the asymptotics of $F(z)$ and $K(x, y; z; q)$ separately, and then prove Theorem \ref{thm:edge}\ref{enu:thm:bulk_b}.

For the asymptotics of $F(z)$, we have the following estimate:
\begin{lem} \label{lem:est_F_bulk}
  Let $\epsilon > 0$ be a small constant independent of $n$.
  \begin{enumerate}[label=(\alph*)]
  \item \label{enu:lem:est_F_bulk_a}
    If $z \in \Gamma$ and $\lvert z - (e^{c} - 1) \rvert < \epsilon$, then there exist $\delta > 0$ and $C > 0$ such that
    \begin{equation} \label{eq:est_F_bulk_meso}
      \lvert F(z) \rvert < C \frac{\sqrt{2\pi n}}{\sqrt{c e^{c} (e^{c} - 1)}} \exp(-n \delta \lvert z - (e^{c} - 1) \rvert^2),
    \end{equation}
    and if $\lvert z - (e^{c} - 1) \rvert < n^{-2/5}$, then
    \begin{equation} \label{eq:est_F_bulk_micro}
      \frac{F(z)}{e^c -1} = \frac{\sqrt{2\pi n}}{\sqrt{c e^{c} (e^{c} - 1)}} \exp \left( \frac{n(z - (e^{c} - 1))^2}{2c e^{c}(e^{c} - 1)} \right) (1 + \bigO(n^{-1/5})).
    \end{equation}
  \item \label{enu:lem:est_F_bulk_b}
    If $z \in \Gamma$ and $\lvert z - (e^{c} - 1) \rvert \geq \epsilon$, then there exists $\delta > 0$ such that for large enough $n$,
    \begin{equation} \label{eq:est_F_bulk_macro}
      \lvert F(z) \rvert < e^{-\delta n}.
    \end{equation}
  \end{enumerate}
\end{lem}
\begin{proof}
  We write
  \begin{equation} \label{eq:log_F}
    \frac{1}{n} \log F(z) = \frac{1}{n} \log \left( q^{-n(n - 1)/2} (q; q)_n \right) - \log z + \int^{\infty}_0 \log (1 + e^{-c \lfloor nx \rfloor/n} z) dx.
  \end{equation}
  Unless $z$ is very close to the negative real line, $n^{-1} \log F(z)$ is approximated by
  \begin{equation}
    \frac{1}{n} \log F(z) = G_n(z) + \bigO(n^{-1}) \quad \text{if $\arg z \in (-\pi + \epsilon', \pi - \epsilon')$},
  \end{equation}
  where $\epsilon'$ is any positive constant and
  \begin{equation}
    G_n(z) = \frac{1}{n} \log \left( q^{-n(n - 1)/2} (q; q)_n \right) - \log z + \int^{\infty}_0 \log(1 + e^{-c x}z) dx.
  \end{equation}
  
  Hence by differentiation, we have that for $\lvert z \rvert = e^{c} - 1$ and $\arg z \in (-\pi + \epsilon', \pi - \epsilon')$, 
  \begin{align}
    \frac{1}{n} \frac{d}{dz} \log F(z) = {}& G'_n(z) + \bigO(n^{-1}) = -\frac{1}{z} + \int^{\infty}_0 \frac{e^{-c x}}{1 + z e^{-c x}} dx + \bigO(n^{-1}), \\
    \frac{1}{n} \frac{d^2}{dz^2} \log F(z) = {}& G''_n(z) + \bigO(n^{-1}) = \frac{1}{z^2} - \int^{\infty}_0 \frac{e^{-2c x}}{(1 + z e^{-c x})^2} dx + \bigO(n^{-1}),
  \end{align}
  and furthermore
  \begin{equation}
    \left. \frac{1}{n} \frac{d}{dz} \log F(z) \right\rvert_{z = e^c - 1} =  \bigO(n^{-1}), \quad \left. \frac{1}{n} \frac{d^2}{dz^2} \log F(z) \right\rvert_{z = e^c - 1} = \frac{1}{c} \frac{1}{e^{c}(e^{c} - 1)} + \bigO(n^{-1}).
  \end{equation}
  Hence $z = e^{c} - 1$ is a saddle point for $\log F(z)$, and as $z$ moves away from the saddle point $e^{c} - 1$, $\lvert F(z) \rvert$ decreases rapidly, provided that $z$ is in the vicinity of the saddle point. Actually, for $z$ on $\Gamma$ but not in the vicinity of $e^{c} - 1$, note that $\lvert z^{-n - 1} \rvert$ is a constant for $z \in \Gamma$ while $\lvert 1 + q^k z \rvert$ decreases as $\arg z$ changes from $0$ to $\pm \pi$, $\lvert F(z) \rvert$ decreases as $\arg z$ changes from $0$ to $\pm \pi$.
   
  The remaining task is to evaluate $F(e^{c} - 1)$ as $n \to \infty$. Although a direct computation is possible, it is difficult due to the evaluation of $(q; q)_n$ with $q$ close to $1$. Instead, we take an indirect approach. 
  
  In the gap probability formula \eqref{eq:general_gap_prob}, if we take $A = \realR$, we have that the probability on the left-hand side is $1$, and Fredholm determinant on the right-hand side is trivially $1$, so we have
  \begin{equation}
    \frac{1}{2\pi i} \oint_{\Gamma} F(z) \frac{dz}{z} = 1.
  \end{equation}
  By the asymptotic properties of $F(z)$ discussed above, we apply the steepest-descent analysis, and have that
  \begin{equation} \label{eq:saddle}
    \frac{1}{2\pi i} \int^{\infty \cdot i}_{-\infty \cdot i} F(e^{c} - 1) e^{\frac{w^2}{2c e^{c}(e^{c} - 1)}} \frac{dw}{\sqrt{n}(e^c-1)} = 1 + \bigO(n^{-1}),
  \end{equation}
  and then
  \begin{equation} \label{eq:est_F_saddle_pt}
  \frac{  F(e^{c} - 1)}{e^c-1} = \frac{\sqrt{2\pi n}}{\sqrt{c e^{c} (e^{c} - 1)}} (1 + \bigO(n^{-1})).
  \end{equation}
  Hence the lemma is proved.
\end{proof}

We compute the asymptotics of $K(x_1, x_2; z; q)$ with the scaling
\begin{equation} \label{eq:specify_x_y_crit_bulk}
  x_1 = 2x \sqrt{n} + \frac{\pi \xi}{\sqrt{n/c}}, \quad x_2 = 2x\sqrt{n} + \frac{\pi\eta}{\sqrt{n/c}},
\end{equation}
where $x\in \realR$ is fixed and $\xi, \eta$ in a compact subset of $\realR$. The result we need is as follows.
\begin{lem} \label{lem:est_K_bulk}
  Let $\epsilon > 0$ be a small constant independent of $n$. In both parts of the lemma we assume $q=e^{-c/n}$ and $x_1$, $x_2$ are as in \eqref{eq:specify_x_y_crit_bulk}.
  \begin{enumerate}[label=(\alph*)]
  \item \label{enu:lem:est_K_bulk:a}
    If $z \in \Gamma$ and $\lvert z - (e^{c} - 1) \rvert < \epsilon$, then
    \begin{equation} \label{eq:est_K_bulk_near}
      \lim_{n \to \infty} \frac{\sqrt{c}}{\sqrt{n}} K(x_1, x_2; z; q) = K_{\interpolating}(\xi, \eta;x;  c; z),
    \end{equation}
    where
    \begin{equation}
       K_{\interpolating}(\xi, \eta;x; c; z) = \frac{1}{\pi} \int^{\infty}_0 \frac{z}{e^{u^2} e^{c x^2} + z} \cos \left( \pi u(\xi - \eta) \right) du.
    \end{equation}
  \item \label{enu:lem:est_K_bulk:b}
    If $z \in \Gamma$ and $\lvert z - (e^{c} - 1) \rvert \geq \epsilon$, then there exists $C > 0$ such that for large enough $n$,
    \begin{equation} \label{eq:est_K_bulk_2}
      \left\lvert K(x_1, x_2; z; q) \right\rvert < C n^2.
    \end{equation}
  \end{enumerate}
\end{lem}
Here we note that 
\begin{equation} \label{eq:relation_K_inters}
  K_{\interpolating}(\xi, \eta; x; c; e^{c} - 1) =  K_{\interpolating} \left(\xi, \eta; \frac{e^{c x^2}}{e^{c} - 1} \right).
\end{equation}

\begin{proof}[Proof of Lemma \ref{lem:est_K_bulk}\ref{enu:lem:est_K_bulk:a}]
  We concentrate on the case $x > 0$. The argument for the $x < 0$ case is the same, since $\varphi_k$ are even or odd functions, depending on the parity of $k$. The case $x = 0$ requires some modification, and we discuss it in Remark \ref{rmk:c=0}.
  
  Recall that $K(x_1, x_2; z; q)$ is an infinite linear combination of $\varphi_k(x_1) \varphi_k(x_2)$ with $k \geq 0$. Let $\epsilon > 0$ be a small constant. Then we divide $K(x_1, x_2; z; q)$ into four parts as follows:
  \begin{align}
    K^{\sup}(x_1; x_2; z; q) = {}& \sum^{\infty}_{k > n(x^2 + \epsilon)} \frac{q^k z}{1 + q^k z} \varphi_k(x_1) \varphi_k(x_2), \\
    K^{\inter}(x_1; x_2; z; q) = {}& \sum_{n(x^2 - \epsilon) < k \leq n(x^2 + \epsilon)} \frac{q^k z}{1 + q^k z}  \varphi_k(x_1) \varphi_k(x_2), \\
    K^{\sub}(x_1; x_2; z; q) = {}& \sum_{n\epsilon < k \leq n(x^2 - \epsilon)} \frac{q^k z}{1 + q^k z} \varphi_k(x_1) \varphi_k(x_2), \\
    K^{\res}(x_1; x_2; z; q) = {}& \sum_{0 \leq k \leq n\epsilon} \frac{q^k z}{1 + q^k z}  \varphi_k(x_1) \varphi_k(x_2).
  \end{align}
  
  Below we show that as $n \to \infty$, for all small enough $\epsilon > 0$, there exists $C > 0$ that is independent of $\epsilon$, such that
  \begin{gather}
    \left\lvert \frac{\sqrt{c}}{\sqrt{n}} K^{\sup}(x_1, x_2; z; q) - \frac{1}{\pi} \int^{\infty}_{\sqrt{\epsilon c}} \frac{ z}{e^{ u^2}e^{c x^2} + z} \cos \left( \pi u(\xi - \eta) \right) du \right\rvert < C \sqrt{\epsilon}, \label{eq:est_ker_crit_sup} \\
    \left\lvert \frac{\sqrt{c}}{\sqrt{n}} K^{\inter}(x_1, x_2; z; q) \right\rvert < C \sqrt{\epsilon}, \quad \left\lvert \frac{\sqrt{c}}{\sqrt{n}} K^{\sub}(x_1, x_2; z; q) \right\rvert = o(1), \quad \left\lvert \frac{\sqrt{c}}{\sqrt{n}} K^{\res}(x_1, x_2; z; q) \right\rvert = o(1). \label{eq:est_ker_crit_remainder}
  \end{gather}
  By taking $\epsilon \to 0$ in the inequalities above, we prove \eqref{eq:est_K_bulk_near}. Below we prove the four results. For notational simplicity, when we prove the three estimates in \eqref{eq:est_ker_crit_remainder}, we only consider the case that $x_1 = x_2 =  2x\sqrt{n}$. 
  
  First we prove \eqref{eq:est_ker_crit_sup}. By \cite[Formula 8.22.12]{Szego75}, for $k > n(x^2 + \epsilon)$, we have
  \begin{equation} \label{eq:limit_Herm_bulk}
    \begin{split}
      \sqrt{\pi} k^{1/4} \varphi_k(2x\sqrt{n}) = {}& \sin(\phi_k)^{-1/2} \sin \left[ \frac{2k + 1}{4} (\sin(2\phi_k) - 2\phi_k) + \frac{3\pi}{4} \right] + \bigO(n^{-1}) \\
      = {}& \left( 1 - \frac{x^2 n}{k + \frac{1}{2}} \right)^{-1/4} \sin \left[ -(2k + 1) \theta_k \right] + \bigO(n^{-1}),
    \end{split}
  \end{equation}
  where
  \begin{equation}
    \phi_k = \arccos \left( x\sqrt{\frac{n}{k + 1/2}} \right), \quad \theta_k = -(2k + 1) \int^1_{x\sqrt{\frac{n}{k + 1/2}}} \sqrt{1 - t^2} dt + \frac{3\pi}{4}.
  \end{equation}
If $x_1$ is as specified in \eqref{eq:specify_x_y_crit_bulk}, then 
  \begin{equation} \label{eq:limit_Herm_bulk_x}
    \sqrt{\pi} k^{1/4} \varphi_k(x_1) = \left( 1 - \frac{x^2 n}{k + \frac{1}{2}} \right)^{-1/4} \sin \left( \theta_k +x\pi \xi\sqrt{c} \sqrt{\frac{k + 1/2}{x^2 n} - 1}  \right) + \bigO(n^{-1}),
  \end{equation}
  and also have an analogous formula for $\varphi_k(x_2)$ with $x_2$ specified in \eqref{eq:specify_x_y_crit_bulk}. Then we have
  \begin{multline}
    \pi k^{1/2} \varphi_k(x_1) \varphi_k(x_2) = \left( 1 - \frac{x^2 n}{k + \frac{1}{2}} \right)^{-1/2} \\
    \times \frac{1}{2} \left[ \cos \left( \pi\sqrt{c} (\xi - \eta)\sqrt{\frac{k + 1/2}{n} - x^2} \right) - \cos \left( 2\theta_k + \pi\sqrt{c}  (\xi + \eta)\sqrt{\frac{k + 1/2}{n} - x^2}\right) \right] + \bigO(n^{-1}).
  \end{multline}
  Now we define
  \begin{equation} \label{eq:terms_K^sup_1}
    K^{\sup, 1}(x_1, x_2; z; q) = \sum^{\infty}_{k > n(x^2 + \epsilon)} \frac{q^k z}{1 + q^k z} \frac{k^{-1/2}}{2\pi} \left( 1 - \frac{x^2 n}{k + \frac{1}{2}} \right)^{-1/2} \cos \left(\pi \sqrt{c} (\xi - \eta)\sqrt{\frac{k + 1/2}{n} - x^2}  \right),
  \end{equation}
  and
  \begin{equation} \label{eq:terms_K^sup_2}
    \begin{split}
      & K^{\sup, 2}(x_1, x_2; z; q) = K^{\sup}(x_1, x_2; z; q) - K^{\sup, 1}(x_1, x_2; z; q) \\
      = {}& \sum^{\infty}_{k > n(x + \epsilon)} \frac{q^k z}{1 + q^k z} \left[ \frac{k^{-1/2}}{2\pi} \left( 1 - \frac{x^2 n}{k + \frac{1}{2}} \right)^{-1/2} \cos \left( 2\theta_k +\pi\sqrt{c}(\xi + \eta)  \sqrt{\frac{k + 1/2}{n} - x^2} \right) + \bigO(1) \right].
    \end{split}
  \end{equation}
  
  It is not hard to see that if $\arg(z) \in (-\pi + \epsilon', \pi - \epsilon')$ for $\epsilon' > 0$, then
  \begin{equation} \label{eq:asy_K^sup_1}
    \begin{split}
      \frac{\sqrt{c}}{\sqrt{n}} K^{\sup, 1}(x_1, x_2; z; q) = {}&\frac{\sqrt{c}}{2\pi} \int^{\infty}_{x^2 + \epsilon} \frac{e^{-c \kappa} z}{1 + e^{-c \kappa} z} \frac{1}{\sqrt{\kappa}} \left( 1 - \frac{x^2}{\kappa} \right)^{-1/2} \cos \left(\pi\sqrt{c} \sqrt{\kappa - x^2} (\xi - \eta) \right) d\kappa + \bigO(n^{-1}) \\
      = {}&\frac{\sqrt{c}}{2\pi} \int^{\infty}_{\epsilon} \frac{ z}{e^{c t}e^{c x^2 }+ z} \cos \left(\pi \sqrt{tc}(\xi - \eta) \right) \frac{dt}{\sqrt{t}} + \bigO(n^{-1}) \\
      = {}&\frac{1}{\pi} \int^{\infty}_{\sqrt{\epsilon c}} \frac{ z}{e^{ u^2}e^{c x^2}  + z} \cos \left(\pi u(\xi - \eta) \right) du + \bigO(n^{-1}).
    \end{split}
  \end{equation}
  On the other hand, we need to show that
  \begin{equation} \label{eq:asy_K^sup_2}
    \left\lvert \frac{c}{\sqrt{n}} K^{\sup, 2}(x_1, x_2; z; q) \right\rvert = o(1).
  \end{equation}
  Since $q^kz/(1 + q^k z) = \bigO(e^{-ck/n})$, although $K^{\sup, 2}$ is defined by an infinite sum in \eqref{eq:terms_K^sup_2}, it suffices to show that for any $\epsilon > 0$ and $N > x$, as $n \to \infty$,
  \begin{equation} \label{eq:auxiliary_sum}
    \sum_{n(x + \epsilon) < k < nN} \frac{q^k z}{1 + q^k z} \frac{k^{-1/2}}{2\pi} \left( 1 - \frac{x^2 n}{k + \frac{1}{2}} \right)^{-1/2} \cos \left( 2\theta_k +\pi\sqrt{c}(\xi + \eta)  \sqrt{\frac{k + 1/2}{n} - x^2} \right) = o(\sqrt{n}).
  \end{equation}
  We note that if we sum up the absolute values of the terms in \eqref{eq:auxiliary_sum}, the result is $\bigO(\sqrt{n})$. For any $k > n(x^2 + \epsilon)$,
  \begin{equation}
    \theta_k - \theta_{k - 1} = \arcsin(x\sqrt{n/k}) - \frac{\pi}{2} + \bigO(n^{-1}),
  \end{equation}
  hence the terms in \eqref{eq:auxiliary_sum} has cancellations. It is not hard to see that the cancellations make the left-hand side of \eqref{eq:auxiliary_sum} to be $o(\sqrt{n})$.
  
  The approximations \eqref{eq:asy_K^sup_1} and \eqref{eq:asy_K^sup_2} imply \eqref{eq:est_ker_crit_sup}.
  
  Next we prove the estimates \eqref{eq:est_ker_crit_remainder} in the special case $x_1 = x_2 =  2x\sqrt{n}$. The analysis is nearly identical for general $\xi$ and $\eta$.

  To prove the first estimate, we use the approximation formula \eqref{eq:refined_global_est}. For $n(x^2 - \epsilon) < k \leq n(x^2 + \epsilon)$, and $x$ in a compact subset of $(-1, +\infty)$,
  \begin{equation}
    k^{1/12} \varphi_k(2\sqrt{k + 1/2} x) = 2^{1/6} \left( \frac{\zeta(x)}{x^2 - 1} \right)^{1/4} \Ai((2k + 1)^{2/3} \zeta(x)) + \bigO(n^{-1}).
  \end{equation}
  Hence we have
  \begin{equation}
    k^{1/12} \varphi_k(2x\sqrt{n}) = 2^{1/6} \left( \frac{\zeta(x_k)}{x^2_k - 1} \right)^{1/4} \Ai((2k + 1)^{2/3} \zeta(x_k)) + \bigO(n^{-1}), \quad \text{where} \quad x_k = \sqrt{\frac{x^2 n}{k + 1/2}}.
  \end{equation}
  
  Hence using the estimate \eqref{eq:asy_Airy_func} of Airy function, we have that if $\arg z \in (-\pi + \delta, \pi - \delta)$ and $n$ is large enough, the first inequality of \eqref{eq:est_ker_crit_remainder} is proved by the estimate
  \begin{equation}
    \begin{split}
      & \frac{1}{\sqrt{n}} \left\lvert K^{\inter}(2x\sqrt{n}, 2x\sqrt{n}; z) \right\rvert \\
      = {}& \left( \frac{2}{n} \right)^{1/3} \left\lvert \int^{x^2 + \epsilon}_{x^2 - \epsilon} \frac{e^{-c \kappa} z}{1 + e^{-c \kappa} z} \left( \frac{\zeta(\frac{x}{\sqrt{\kappa}})}{\frac{x^2}{\kappa} - 1} \right)^{1/2} \Ai \left( (2n)^{2/3} \zeta \left( \frac{x}{\sqrt{\kappa}} \right) \right)^2 d\kappa + \bigO(n^{-1}) \right\rvert \\
      \leq {}& \left( \frac{2}{n} \right)^{1/3} \int^{x^2 + \epsilon}_{x^2 - \epsilon} \left\lvert \frac{e^{-c \kappa} z}{1 + e^{-c \kappa} z} \right\rvert 2^{1/3} \left( \frac{\zeta(\frac{x}{\sqrt{\kappa}})}{\frac{x^2}{\kappa} - 1} \right)^{1/2} f \left( (2n)^{2/3} \zeta \left( \frac{x}{\sqrt{\kappa}} \right) \right)^2 d\kappa \\
      \leq {}& C \sqrt{\epsilon},
    \end{split}
  \end{equation}
  where $f$ is defined in \eqref{eq:asy_Airy_func}, $\zeta(x)$ has the behavior close to $1$ given in \eqref{eq:local_prop_zeta}, and $C > 0$ is independent of $n$ and $\epsilon$. 
  
  To prove the second estimate, By \cite[Formula 8.22.13]{Szego75}, for $n\epsilon < k \leq n(x^2 - \epsilon)$, we have
  \begin{equation} \label{eq:limit_Herm_out}
    \begin{split}
      \sqrt{\pi} k^{1/4} \varphi_k(2x\sqrt{n}) = {}& \frac{1}{2} \sinh(\phi_k)^{-1/2} \exp \left[ \frac{2k + 1}{4} (2\phi_k - \sinh(2\phi_k)) \right] (1 + \bigO(n^{-1})) \\
      = {}& \frac{1}{2} \left( \frac{x^2 n}{k + \frac{1}{2}} - 1 \right)^{-1/4} \exp \left[ -(2k + 1) \int^{x\sqrt{\frac{n}{k + 1/2}}}_1 \sqrt{t^2 - 1} dt \right] (1 + \bigO(n^{-1})),
    \end{split}
  \end{equation}
  where
  \begin{equation}
    \phi_k = \arccosh \left( x\sqrt{\frac{n}{k + 1/2}} \right).
  \end{equation}
  It is clear that
  \begin{equation}
    \lvert \varphi_k(2x\sqrt{n}) \rvert < e^{-\epsilon' n}
  \end{equation}
  for all $n\epsilon < k \leq n(x^2 - \epsilon)$, where $\epsilon' > 0$ is a constant depending on $\epsilon$ and $c$. This estimate implies the second inequality of \eqref{eq:est_ker_crit_remainder} with $x_1 = x_2 =  2x\sqrt{n}$. 
  
Finally, by the estimate of Hermite polynomials provided in \cite[Section 11.4, Exercises 4.2 and 4.3]{Olver97}, we have that
  \begin{equation} \label{eq:est_Herm_res}
    \lvert \varphi_k(2x\sqrt{n}) \rvert < e^{-\epsilon'' n}
  \end{equation}
  for all $k \leq n\epsilon$, where $\epsilon'' > 0$ depends on $c$ only, provided that $\epsilon$ is small enough. This estimate implies the last inequality of \eqref{eq:est_ker_crit_remainder} with $x = y = 2x\sqrt{n}$. Here we note that the result in \cite[Section 11.4, Exercises 4.2 and 4.3]{Olver97} is valid even for very small $k$, like $k = 1, 2, \dotsc$, except for $k = 0$. But it is obvious that when $k = 0$, \eqref{eq:est_Herm_res} holds.
\end{proof}

\begin{rmk} \label{rmk:c=0}
  The case $x = 0$ is different, because $K^{\sub}$ is not longer meaningful, and $K^{\inter}$ and $K^{\res}$ need to be combined. The asymptotic analysis becomes easier, since $\varphi_k(\xi/\sqrt{n})$ has limiting formulas simpler than \eqref{eq:limit_Herm_bulk}, \eqref{eq:limit_Herm_bulk_x}, and \eqref{eq:limit_Herm_out}, see \cite[22.15.3--4]{Abramowitz-Stegun64}. We omit the detail, because a similar computation is done in \cite[Proof of Theorem 1.9]{Johansson07}.
\end{rmk}

\begin{proof}[Proof of Lemma \ref{lem:est_K_bulk}\ref{enu:lem:est_K_bulk:b}]
  The difficulty is that when $\arg z$ is close to $\pm \pi$, the denominator $1 + q^k z$ appearing in $K(x_1, x_2; z; q)$ can be close to zero. But since $\lvert z \rvert = q^{-n} - 1 + \delta_n/n = e^{c} - 1 + \delta_n/n$, and we assume that $\delta_n = 0$, for $k \geq n$
  \begin{equation}
    \lvert 1 + q^k z \rvert \geq 1 - q^n(q^{-n} - 1) \geq q^{n} = e^{-c} > 0,
  \end{equation}
  and the denominator is not close to zero. Then by the estimates that we use in the proof of part \ref{enu:lem:est_K_bulk:a}, we have for all $z \in \Gamma$,
  \begin{equation} \label{eq:est_K_bulk_leftmost_1}
    \sum^{\infty}_{k \geq n} \frac{q^k z}{1 + q^k z} \varphi_k(x_1) \varphi_k(x_2) = \bigO(n^{1/2}).
  \end{equation}
  On the other hand, for $k < n$, by assumption \eqref{eq:z_abs_value} we have $\lvert 1 + q^k z \rvert \geq \epsilon(q)/n$, and then by the uniform boundedness of the Hermite functions, 
  \begin{equation} \label{eq:est_K_bulk_leftmost_2}
    \left\lvert \sum^n_{k = 0} \frac{q^k z}{1 + q^k z} \varphi_k(x_1) \varphi_k(x_2) \right\rvert < \epsilon(q) n^2.
  \end{equation}
  The combination of \eqref{eq:est_K_bulk_leftmost_1} and \eqref{eq:est_K_bulk_leftmost_2} implies \eqref{eq:est_K_bulk_2}, and then finish the proof.
\end{proof}

\begin{proof}[Proof of Theorem \ref{thm:bulk}\ref{enu:thm:bulk_b} for $2$-correlation function]
  Using the estimates in Lemmas \ref{lem:est_F_bulk} and \ref{lem:est_K_bulk}, we have that the integral in \eqref{eq:formula_for_R^m} concentrates in the vicinity of the saddle point $z = e^{c} - 1$, and more precisely, in the region $\lvert z - (e^{c} - 1) \rvert = \bigO(n^{-1/2})$. A straightforward application of the Laplace method yields that if $x_1, x_2$ depend on $\xi, \eta$ as in \eqref{eq:specify_x_y_crit_bulk}, using \eqref{eq:saddle},
  \begin{equation}
    \begin{split}
    & \lim_{n \to \infty} \left(\frac{\pi}{\sqrt{n/c}}\right)^2 R^{(2)}_n(x_1, x_2) \\
    = {}& \lim_{n \to \infty} \frac{1}{2\pi i} \oint_0 F(z)
    \begin{vmatrix}
      \frac{\pi\sqrt{c}}{\sqrt{n}}K(x_1, x_2; z) & \frac{\pi\sqrt{c}}{\sqrt{n}}K(x_1, x_2; z) \\
     \frac{\pi\sqrt{c}}{\sqrt{n}} K(x_2, x_1; z) & \frac{\pi\sqrt{c}}{\sqrt{n}}K(x_2, x_2; z) 
    \end{vmatrix}
    \frac{dz}{z} \\
    = {}& \lim_{n \to \infty} \frac{1}{2\pi i} \int^{\infty \cdot i}_{-\infty \cdot i} \frac{F(e^{c} - 1)}{e^c-1} e^{\frac{w^2}{2c e^{c}(e^{c} - 1)}}
    \begin{vmatrix}
      K_{\interpolating}(\xi, \xi; c;x; e^{c} - 1) & K_{\interpolating}(\xi, \eta;x; c; e^{c} - 1) \\
      K_{\interpolating}(\eta, \xi; c;x; e^{c} - 1) & K_{\interpolating}(\eta, \eta;x; c; e^{c} - 1) 
    \end{vmatrix}
    \frac{dw}{\sqrt{n}} \\
    = {}&
    \begin{vmatrix}
      K_{\interpolating}(\xi, \xi; x; c; e^{c} - 1) & K_{\interpolating}(\xi, \eta; x; c; e^{c} - 1) \\
      K_{\interpolating}(\eta, \xi; x; c; e^{c} - 1) & K_{\interpolating}(\eta, \eta;x; c; e^{c} - 1) 
    \end{vmatrix}.
    \end{split}
  \end{equation}
  By \eqref{eq:relation_K_inters} we prove the $2$-correlation function formula in Theorem \ref{thm:bulk}\ref{enu:thm:bulk_b}.
\end{proof}

\subsection{Correlation functions for the bulk particles: fixed $q \in (0, 1)$}

We let $q$ be in a compact subset of $(0, 1)$. We assume that the contour in \eqref{eq:formula_for_R^m} is $\lvert z \rvert = q^{-n + 1/2}$, and take the change of variable like in \eqref{eq:w=q^nz}
\begin{equation} \label{eq:scaled_w}
  w = q^n z \quad \text{with} \quad \lvert w \rvert = \sqrt{q}.
\end{equation}
Then analogous to \eqref{eq:final_integral_rightmost_in_w}, we write the $m = 2$ case of \eqref{eq:formula_for_R^m} as 
\begin{equation}
  \begin{split}
    R^{(2)}_n(x_1, x_2) = {}& \frac{q^{n/2}}{Z_n(q)} q^{n^2} q^{-\frac{n(n + 1)}{2}} \frac{1}{2\pi i} \oint_0 \frac{dw}{w} \left( \prod^{\infty}_{k = 0} (1 + q^k w) \right) \left( \prod^n_{k = 1} (1 + q^k w^{-1}) \right) \\
    & \times
    \begin{vmatrix}
      K(x_1, x_1; q^{-n}w; q) & K(x_1, x_2; q^{-n}w; q) \\
      K(x_2, x_1; q^{-n}w; q) & K(x_2, x_2; q^{-n}w; q)
    \end{vmatrix} \\
    = {}& \frac{(q; q)_n}{(q; q)_{\infty}} \frac{1}{2\pi i} \oint_0 \frac{dw}{w} \left( \sum^{\infty}_{k = -\infty}  q^{\frac{k(k - 1)}{2}} w^k \right) \frac{(-q/w; q)_n}{(-q/w; q)_{\infty}} \\
    & \times
    \begin{vmatrix}
      K(x_1, x_1; q^{-n}w; q) & K(x_1, x_2; q^{-n}w; q) \\
      K(x_2, x_1; q^{-n}w; q) & K(x_2, x_2; q^{-n}w; q)
    \end{vmatrix},
  \end{split}
\end{equation}
where we make use of identity \eqref{eq:triple_product}.
% Now we write
% \begin{equation}
%   R^{(2)}(x_1, x_2) = \frac{(q; q)_n}{(q; q)_{\infty}} (\tilde{R}^{(2)}(x_1, x_2) - \hat{R}^{(2)}(x_1, x_2)),
% \end{equation}
% where
% \begin{align}
%   \tilde{R}^{(2)}(x_1, x_2) = {}& \frac{1}{2\pi i} \oint_0 \frac{dw}{w} (q; q)_{\infty} (-w; q)_{\infty} (-q/w; q)_{\infty} \det( K(x_i, x_j; q^{-n}w) )^2_{i, j = 1}, \\
%   \hat{R}^{(2)}(x_1, x_2) = {}& \frac{1}{2\pi i} \oint_0 \frac{dw}{w} (q; q)_{\infty} (-w; q)_{\infty} ((-q/w; q)_{\infty} - (-q/w; q)_{n - 1}) \det( K(x_i, x_j; q^{-n}w) )^2_{i, j = 1}.
% \end{align}
% Using identity \eqref{eq:triple_product}, we have that
% \begin{equation}
%   \tilde{R}^{(2)}(x_1, x_2) = \frac{1}{2\pi i} \oint_0 \frac{dw}{w}\left(\sum^{\infty}_{k = -\infty}  q^{\frac{k(k - 1)}{2}} w^k\right) \det( K(x_i, x_j; q^{-n}w) )^2_{i, j = 1}.
% \end{equation}
Next we find the asymptotics of $K(x_i, x_j; q^{-n} w; q)$. We write
\begin{equation}
  K(x_i, x_j; q^{-n}w; q) = K^{(0)}_n(x_i, x_j) + K^{(1)}(x_i, x_j; q^{-n}w; q) - K^{(2)}(x_i, x_j; q^{-n}w; q),
\end{equation}
where
\begin{align}
  K^{(0)}(x_i, x_j) = {}& \left( \sum^{n - 1}_{k = 0} \varphi_k(x_i) \varphi_k(x_j) \right), \\
  K^{(1)}(x_i, x_j; q^{-n}w; q) = {}& \left( \sum^{\infty}_{k = 0} \frac{q^k w}{1 + q^k w} \varphi_{n + k}(x_i) \varphi_{n + k}(x_j) \right), \\
  K^{(2)}(x_i, x_j; q^{-n}w; q) = {}& \left( \sum^n_{k = 1} \frac{q^k w^{-1}}{1 + q^k w^{-1}} \varphi_{n - k}(x_i) \varphi_{n - k}(x_j) \right).
\end{align}
It is well known that $K^{(0)}_n(x_i, x_j)$ is the correlation kernel of $n$-dimensional GUE random matrix, and for
\begin{equation}
  x_i = 2\sqrt{n} x + \frac{\pi \xi_i}{(1 - x^2)^{1/2} \sqrt{n}}, \quad x_j = 2\sqrt{n} x + \frac{\pi \xi_j}{(1 - x^2)^{1/2} \sqrt{n}}, \quad \text{where} \quad x \in (-1, 1),
\end{equation}
we have \cite[Chapter 3]{Anderson-Guionnet-Zeitouni10}
\begin{equation}
  \lim_{n \to \infty} \frac{\pi}{(1 - x^2)^{1/2} \sqrt{n}} K^{(0)}_n(x_i, x_j) = K_{\sin}(\xi_i, \xi_j) := \frac{\sin(\pi(\xi_i - \xi_j))}{\pi(\xi_i - \xi_j)}.
\end{equation}
To estimate $K^{(1)}$ and $K^{(2)}$, it suffices to use the rough estimate from \cite[22.14.17]{Abramowitz-Stegun64}, we have $\lvert \varphi_n(x) \rvert \leq \frac{\kappa}{2^{1/4} \pi^{1/4}}$, where $\kappa \approx 1.086435$. Then we have
\begin{equation}
  \left\lvert K^{(1)}(x_i, x_j; q^{-n}w; q) \right\rvert \leq \sum^{\infty}_{k = 0} \left\lvert \frac{q^k w}{1 + q^k w} \right\rvert \frac{\kappa^2}{\sqrt{2 \pi}} < \sum^{\infty}_{k = 0} \frac{q^k}{1 - \sqrt{q}} \frac{\kappa^2}{\sqrt{2 \pi}} < \frac{1}{(1 - q)(1 - \sqrt{q})}.
\end{equation}
Similarly, we also have 
\begin{equation}
  \left\lvert K^{(2)}(x_i, x_j; q^{-n}w; q) \right\rvert < \frac{1}{(1 - q)(1 - \sqrt{q})}.
\end{equation}
Hence we have that uniformly in $w$ on the circle $\lvert w \rvert = \sqrt{q}$
\begin{equation} \label{eq:asy_K_finite_q}
  \lim_{n \to \infty} \frac{\pi}{(1 - x^2)^{1/2} \sqrt{n}} K(x_i, x_j; q^{-n}w; q) = K_{\sin}(\xi_i, \xi_j).
\end{equation}
Using the very fast convergence \eqref{eq:convergence_inessential}, we have
\begin{equation}
  \begin{split}
    \lim_{n \to \infty} \left( \frac{\pi}{(1 - x^2)^{1/2} \sqrt{n}} \right)^2 R^{(2)}_n(x_1, x_2) = {}& \frac{1}{2\pi i} \oint_0 \frac{dw}{w} \left(\sum^{\infty}_{k = -\infty}  q^{\frac{k(k - 1)}{2}} w^k \right) (1 + o(1)) \det(K_{\sin}(\xi_i, \xi_j))^2_{i, j = 1} \\
    = {}& \det(K_{\sin}(\xi_i, \xi_j))^2_{i, j = 1}.
  \end{split}
\end{equation}
Hence Theorem \ref{thm:bulk}\ref{enu:thm:bulk_b} is proved for the $2$-correlation function case.
% To compute the limit of $\hat{R}^{(2)}(x_1, x_2)$, we note that
% \begin{equation}
%   \log \left\lvert \frac{(-q/w; q)_{\infty}}{(-q/w; q)_{n - 1}} \right\rvert \leq \sum^{\infty}_{k = n + 1} \log(1 + q^{k - 1/2}) < \sum^{\infty}_{k = n + 1} q^{k - 1/2} = \frac{q^{n + 1/2}}{1 - q},
% \end{equation}
% and analogously, under the assumption that $q^{n - 1/2} < 1/2$,
% \begin{equation}
%   \log \left\lvert \frac{(-q/w; q)_{\infty}}{(-q/w; q)_{n - 1}} \right\rvert \geq \sum^{\infty}_{k = n + 1} \log(1 - q^{k - 1/2}) > -2 \sum^{\infty}_{k = n + 1} q^{k - 1/2} = -2\frac{q^{n + 1/2}}{1 - q}.
% \end{equation}
% Hence by \eqref{eq:triple_product}, we have that
% \begin{equation}
%   \lvert (q; q)_{\infty} (-w; q)_{\infty} ((-q/w; q)_{\infty} - (-q/w; q)_{n - 1}) \rvert = \bigO(q^n),
% \end{equation}
% and then by the asymptotics \eqref{eq:asy_K_finite_q} we have that
% \begin{equation}
%   \lim_{n \to \infty} \left( \frac{\pi}{(1 - x^2)^{1/2} \sqrt{n}} \right)^2 \tilde{R}^{(2)}(x_1, x_2) = 0.
% \end{equation}

% The limit formulas above together with the simple limit 
% \begin{equation}
%   \lim_{n \to \infty} \frac{(q; q)_n}{(q; q)_{\infty}} = 1
% \end{equation}
% imply that
% \begin{equation}
%   \lim_{n \to \infty} \left( \frac{\pi}{(1 - x^2)^{1/2} \sqrt{n}} \right)^2 R^{(2)}(x_1, x_2) = \det(K_{\sin}(\xi_i, \xi_j))^2_{i, j = 1}.
% \end{equation}
\begin{rmk}
  The argument in this section also occurs in \cite[Proposition 3.7]{Johansson-Lambert15}.
\end{rmk}

\section{Relation to interacting particle systems} \label{sec:relations}

Theorem \ref{thm:algebraic}\ref{enu:thm:algebraic_1} for the gap probability in the MNS model has analogues in the study of several interacting particle systems that are related to the Kardar--Parisi--Zhang (KPZ) universality class. In this section we consider the $q$-Whittaker processes, which are obtained by a specialization of Macdonald processes \cite[Section 3]{Borodin-Corwin13}, and the Asymmetric Simple Exclusion Process (ASEP) as primary examples. We also consider the $q$-deformed Totally Asymmetric Simple Exclusion Process ($q$-TASEP), which is a continuous limit of the $q$-Whittaker process \cite[Section 3.3]{Borodin-Corwin13}, \cite{Borodin-Corwin-Sasamoto14} and the $q$-deformed Totally Asymmetric Zero Range Process ($q$-TAZRP), which is the dual process of $q$-TASEP \cite{Korhonen-Lee14}, \cite{Lee-Wang17}. 

\paragraph{Identity of Fredholm determinants}

Let $f(\xi)$ be a meromorphic function on $\compC$ with the finite set of poles $\A = \{ a_1, \dotsc, a_m \} \not\ni 0$, and suppose that $f(0) = 1$. Let $\Gamma_{0, \A}$ be a contour with positive orientation such that $0$ and $\A$ are enclosed in $\Gamma_{0, \A}$. On the other hand, let $\Gamma_{\A}$ be a contour with positive orientation such that $\A$ is enclosed in $\Gamma_{\A}$ but $0$ is outside of $\Gamma_{\A}$. We assume the condition
\begin{equation} \label{eq:essential_condition_of_Gamma}
  \Gamma_{0, \A} \cap q \cdot \Gamma_{0, \A} = \emptyset, \quad \text{and} \quad \Gamma_{\A} \cap q^k \cdot \Gamma_{\A} = \emptyset, \quad \text{for $k = 1, 2, \dotsc$}.
\end{equation}
where for a contour $C$, $q^k \cdot C = \{q^k z \mid z \in C \}$. Furthermore, we define $\Gamma = \Gamma_{0, \A} \cup (-\Gamma_{\A})$, where $-\Gamma_{\A}$ is the contour $\Gamma_{\A}$ with negative orientation. 

\begin{figure}[htb]
  \centering
  \includegraphics{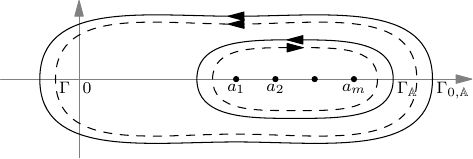}
  \caption{The shapes of $\Gamma_{\mathbb{A}}$ and $\Gamma_{0, \mathbb{A}}$ (solid) and the shape of $\Gamma$ (dashed).}
\end{figure}

We define kernel functions on $\Gamma$,
\begin{equation} \label{eq:kernel_formula_contours}
  M(\xi, \eta; q) = \frac{f(\eta)}{\xi - q\eta},
\end{equation}
and
\begin{equation} \label{eq:K_kernel_contour}
  \begin{split}
    K(\xi, \eta; z; q) = {}& \sum_{k=1}^\infty (-1)^{k+1} z^k  \frac{F(\eta; q; k) }{\xi - q^k \eta},
 \\
%    = {}& \frac{f(\xi)}{\xi} \frac{-1}{2\pi i} \int^{i \cdot \infty}_{-i \cdot \infty} \frac{z}{1 + zq^s} \frac{\pi}{\sin(\pi s)} \left( -\frac{q\eta}{\xi} \right)^s ds,
  \end{split}
\end{equation}
where
\begin{equation}\label{def:F}
F(\eta; q; k):=
\prod_{j=0}^{k-1} f(q^j \eta), \quad k =1,2, 3 \dots
\end{equation}
Since $f(0)=1$ and $0<q<1$, the function $F(\eta; q; k)$ is bounded in $k$, so the power series \eqref{eq:K_kernel_contour} converges for all $|z|<1$.

In many cases the infinite sum \eqref{eq:K_kernel_contour} can be written in a compact form as a contour integral, which often gives the continuation to $|z|\ge 1$. Suppose that the function $F(\eta; q; k)$ is such that the discrete variable $k$ may be extended to a complex variable $s$ in such a way that $F(\eta; q; s)$ is analytic in the right-half of the $s$-plane and decays fast enough as $s\to\infty$ in the right-half of the $s$-plane. Then by calculation of residues, we have the contour integral formula for all $z \in \compC$
\begin{equation}\label{K-integral}
 K(\xi, \eta; z; q) = \frac{1}{2\pi i} \int_{\delta-i\infty}^{\delta+i\infty} \frac{\pi z^s}{\sin(\pi s)}  \frac{ F(\eta; q; s)}{\xi -q^s\eta}ds,
\end{equation}
for some small positive number $0<\delta<1$.

These two kernels define integral operators on $L^2(\Gamma_{0, \A})$, $L^2(\Gamma_{\A})$ and $L^2(\Gamma)$, where the measure is $(2\pi i)^{-1} d\eta$ with the orientation positive on $\Gamma_{0, \A}$ and negative on $\Gamma_{\A}$. We denote these integral operators all by $\M(q)$ and $\K(z; q)$, and the domain is assumed to be  $L^2(\Gamma)$ unless otherwise specified. We have the following technical lemma.

\begin{lem} \label{lem:contour_det}
  Let contours $\Gamma_{0, \A}$, $\Gamma_{\A}$ and meromorphic function $f(\xi)$ be given above. Suppose the meromorphic functions $M(\xi, \eta; q)$ and $K(\xi, \eta; z; q)$ are defined by \eqref{eq:kernel_formula_contours} and \eqref{eq:K_kernel_contour} respectively, and $\M(q)$ and $\K(z; q)$ are integral operators with kernels $M(\xi, \eta; q)$ and $K(\xi, \eta; z; q)$. Then for all $|z|<1$,
  \begin{align}
    \det(I + z \M(q))_{L^2(\Gamma_{0, \A})} = {}& (-z; q)_{\infty} \det(I + \K(z; q))_{L^2(\Gamma_{\A})}, \label{eq:Fred_det_id_contour} \\
    \det(I - z \M(q))_{L^2(\Gamma_{\A})} = {}& (-z; q)_{\infty} \det(I - \K(z; q))_{L^2(\Gamma_{0, \A})}. \label{eq:Fred_det_id_contour_alt}
  \end{align}
  Moreover, if the kernel function $K(\xi, \eta; z; q)$ can be extended by \eqref{K-integral} to $\lvert z \rvert \geq 1$, identities \eqref{eq:Fred_det_id_contour} and \eqref{eq:Fred_det_id_contour_alt} holds for general $z$.
\end{lem}

\begin{proof}
  We only prove identity \eqref{eq:Fred_det_id_contour} for $\lvert z \rvert < 1$, since the result for $\lvert z \rvert \geq 1$ can be obtained by direct analytic continuation, and the proof of \eqref{eq:Fred_det_id_contour_alt} is analogous.
  
  The proof is similar to that of Lemma \ref{prop:MK_identity}. The main difference is that the operator $\M(q)$ on $L^2(\Gamma)$ no longer satisfies $\M(q)^k= \M(q^k)$. Instead we have the formula for the kernel of $\M(q)^k$,
  \begin{equation}\label{Mk_kernel}
  \M^k(\xi,\eta; q) =  \frac{F(\eta; q; k)}{\xi - q^k\eta},
  \end{equation}
where $F(\eta; q; k)$ is defined in \eqref{def:F}. Then analogous to \eqref{eq:defn_R_inverse}, we define the operator $\R(z; q)$ on ${L^2(\Gamma)}$ by
  \begin{equation}
    I - \R(z; q) = (I + z \M(q))^{-1},
  \end{equation}
Similar to \eqref{eq:inv_sub}, we have the identity of operators on $L^2(\Gamma)$,
  \begin{equation}
    (I + z\M(q))(I - \R(z; q) \chi_{(-\Gamma_{\A}})) = I + z \M(q) \chi_{\Gamma_{0, \A}},
  \end{equation}
  and we are left to find an expression for the kernel of $\R(z;q)$.
Assuming that $\lvert z \rvert < 1$, then we can write $\R(z;q)$ as a power series in $z$:
  \begin{equation}
  \R(z;q) = \sum_{k=1}^\infty (-1)^{k+1} z^k \M(q)^k.
  \end{equation}
  Using \eqref{Mk_kernel} we then find that the kernel of $  \R(z;q) $ can be written as
    \begin{equation}
  \R(\xi, \eta; z;q) = \sum_{k=1}^\infty (-1)^{k+1} z^k  \frac{F(\eta; q; k)}{\xi - q^k\eta},
  \end{equation}
  and we find that $\R(z;q) = \K(z; q)$ for $|z|<1$.

  Then similar to \eqref{eq:MK_identity_det}, we have that
  \begin{equation}
    \begin{split}
      \det(I + z \M(q))_{L^2(\Gamma_{0, \A})} = {}& \det(I + z \M(q) \chi_{\Gamma_{0, \A}})_{L^2(\Gamma)} \\
      = {}& \det(I + z\M(q))_{L^2(\Gamma)} \det(I - \K(z; q)\chi_{(-\Gamma_{\A})})_{L^2(\Gamma)}  \\
      = {}& \det(I + z\M(q))_{L^2(\Gamma)} \det(I - \K(z; q))_{L^2(-\Gamma_{\A})} \\
      = {}& \det(I + z\M(q))_{L^2(\Gamma)} \det(I + \K(z; q))_{L^2(\Gamma_{\A})},
    \end{split}
  \end{equation}
  where in the last line the orientation of the integral contour is changed, and so does the sign for the operator $\K(z; q)$. 
  
   In order to complete the proof of Lemma \ref{lem:contour_det}, we are left to prove that
     \begin{equation}\label{det_pochammer}
    \det(I + z\M(q))_{L^2(\Gamma)} = (-z; q)_{\infty}.
  \end{equation}
   We first prove this result under the restriction that 
  \begin{equation} \label{eq:convergence_condition_on_Gamma}
    q < \left\lvert \frac{\xi}{\eta} \right\rvert < q^{-1} \quad \text{for all $\xi, \eta \in \Gamma$}.
  \end{equation}
   Define two bases for $L^2(\Gamma)$: $\{ \varphi_k(\xi) =\xi^{-k-1}  \}^{\infty}_{k = -\infty}$ and $\{ \psi_k(\eta) =f(\eta) \eta^{k } \}^{\infty}_{k = -\infty}$. They satisfy the condition
  \begin{equation}\label{triangular}
    \frac{1}{2\pi i} \int_{\Gamma} \varphi_j(\xi) \psi_k(\xi) d\xi = \left\{
    \begin{aligned}
    & 0 \quad &j<k, \\
    & \frac{f^{(j-k)}(0)}{(j-k)!} \quad &j \ge k.
    \end{aligned}\right.
      \end{equation}
  Under the additional condition \eqref{eq:convergence_condition_on_Gamma}, we have that for $\xi, \eta \in \Gamma$,
  \begin{equation}\label{M-expansion}
    M(\xi, \eta; q) = \sum^{\infty}_{k = 0} q^k \varphi_k(\xi) \psi_k(\eta).
       \end{equation}
Combining this expansion in $\varphi_k$ and $\psi_k$ with the triangularity condition \eqref{triangular} and noting that $f(0)=1$, we easily obtain \eqref{det_pochammer} and prove \eqref{eq:Fred_det_id_contour} under technical restrictions $\lvert z \rvert < 1$ and \eqref{eq:convergence_condition_on_Gamma}.
  
  To remove the technical restriction \eqref{eq:convergence_condition_on_Gamma}, we note that for fixed $q$ and $a_1, \dotsc, a_m$, the contour $\Gamma = \Gamma_{0, \A} \cup (-\Gamma_{\A})$ that satisfies both \eqref{eq:essential_condition_of_Gamma} and \eqref{eq:convergence_condition_on_Gamma} may not exist. However, if $q$ is fixed and $a_1, \dotsc, a_m$ are regarded as movable parameters, then the contour $\Gamma$ exists given that $a_1, \dotsc, a_m$ cluster tightly enough. Although \eqref{eq:Fred_det_id_contour} is proved under the additional condition \eqref{eq:convergence_condition_on_Gamma}, since the kernels $M(\xi, \eta; q)$ and $K(\xi, \eta; z; q)$ are meromorphic functions, by deforming the contours $\Gamma_{0, \A}$ and $\Gamma_{\A}$ and moving the poles $a_1, \dotsc, a_m$ if necessary, we can remove the technical restriction \eqref{eq:convergence_condition_on_Gamma}. 
\end{proof}

  % To remove the technical restriction $\lvert z \rvert < 1$ relies on using analytic continuation in $z$, and it can be helpful to write the kernel $K(\xi, \eta; z; q)$ in a closed form as in \eqref{K-integral}. In general this depends on the existence and analytical properties of the function $F(\eta; q; s)$ of the complex variable $s$. Below we consider applications of Lemma \ref{lem:contour_det} on three interacting particle systems. In each case the formula \eqref{K-integral} holds and the integral converges for all $z\in \compC$, thus giving the analytic extension of $\K(z; q)$ for all $z$ and allowing us to remove the condition $|z|<1.$

\paragraph{$q$-Whittaker processes}

The $q$-Whittaker processes are interacting particle systems defined in \cite[Section 3]{Borodin-Corwin13}. Since the definition is relatively involved, we refer the reader to the original paper, and only remark that in the $N$-particle model,
\begin{enumerate*}[label=(\roman*)]
\item
  the speeds of particles depend on parameters $a_1, \dotsc, a_N \in (0, \infty)$, and
\item
  the transition probabilities depend on parameters $\alpha_1, \dotsc, \alpha_N; \beta_1, \dotsc, \beta_N; \gamma$.
\end{enumerate*}

A moment generating formula for the $q$-Whittaker processes is given in \cite[Theorem 3.23]{Borodin-Corwin13}: 
\begin{equation}
  \left\langle \frac{1}{(-z q^{\lambda_N}; q)_{\infty}} \right\rangle_{\MM_{t = 0}(a_1, \dotsc, a_N; \rho)} = \frac{1}{(-z; q)_{\infty}} \det(1 + z \M(q))_{L^2(\Gamma_{0, \A})},
\end{equation}
where the integral operator $\M(q)$ has the kernel $M(\xi, \eta; q)$ defined by \eqref{eq:kernel_formula_contours}, with the function $f$ given as
\begin{equation} \label{eq:f_for_Whittaker}
  f(\eta) = \left( \prod^N_{m = 1} \frac{a_m}{a_m - \eta} \right) \left( \prod^N_{i = 1} (1 - \alpha_i \eta) \frac{1 + q\beta_i \eta}{1 + \beta_i \eta} \right) \exp((q - 1) \gamma \eta),
\end{equation}
and $\Gamma_{0, \A}$ (denoted by $\tilde{C}_{a, \rho}$ in \cite{Borodin-Corwin13}) is a star-shaped contour centered at $0$ and containing $\A = \{ a_1, \dotsc, a_N \}$ but no other singularities of $f(\eta)$. Then \cite[Corollary 3.24]{Borodin-Corwin13} gives the probability distribution
\begin{equation}
  \Prob_{\MM_{t = 0}(a_1, \dotsc, a_N; \rho)}(\lambda_N = n) = \frac{q^n}{2\pi i} \oint_C \frac{\det(1 + z\M(q))_{L^2(\Gamma_{0, \A})}}{(-z; q)_{n + 1}} dz,
\end{equation}
where the contour $C$ encloses poles $-1, -q^{-1}, \dotsc, -q^n$. It is obvious that this $\Gamma_{0, \A}$ satisfies \eqref{eq:essential_condition_of_Gamma}. For the meaning of the notations and technical conditions, see the paper \cite{Borodin-Corwin13}.

\begin{rmk}
Actually the kernel in for the operator ${\bf M}(q)$ given in \cite[Theorem 3.23]{Borodin-Corwin13} is of the form $ M(\xi, \eta; q) = \frac{f(\xi)}{\xi - q\eta}$ rather than $ M(\xi, \eta; q) = \frac{f(\eta)}{\xi - q\eta}$ as in \eqref{eq:kernel_formula_contours}. It is clear that these two operators give the same Fredholm determinant, since they are each the composition of the integral operator with kernel $ \frac{1}{\xi - q\eta}$ with multiplication by the function $f$, only in different orders.
\end{rmk}

Suppose there also exists a contour $\Gamma_{\A}$ that encloses $\A$ but not $0$ and satisfies condition \eqref{eq:essential_condition_of_Gamma}. Lemma \ref{lem:contour_det} immediately implies that
\begin{align} \label{eq:moment_gene_formula_on_A}
  \left\langle \frac{1}{(-z q^{\lambda_N}; q)_{\infty}} \right\rangle_{\MM_{t = 0}(a_1, \dotsc, a_N; \rho)} = {}& \det(I + \K(z;q))_{L^2(\Gamma_{\A})}, \\
  \Prob_{\MM_{t = 0}(a_1, \dotsc, a_N; \rho)}(\lambda_N = n) = {}& \frac{q^n}{2\pi i} \oint_C (zq^{n + 1}; q)_{\infty} \det(I + \K(z;q))_{L^2(\Gamma_{\A})},
\end{align}
where and the integral operator $\K(z; q)$ has the kernel $K(\xi, \eta; z; q)$ defined in \eqref{eq:K_kernel_contour} with $f$ specified in \eqref{eq:f_for_Whittaker}.

This essentially reproduces the result \cite[Corollary 3.17]{Borodin-Corwin13}, which expresses the moment generating formula on the left-hand side of \eqref{eq:moment_gene_formula_on_A} by a Fredholm determinant formula where the domain consists of infinitely many copies of $\Gamma_{\A}$ (denoted by $C_{a, \rho}$ there). Note that the function $f(\eta)$ can be written as
\begin{equation}
  f(\eta) = \frac{g(\eta)}{g(q\eta)}, \quad g(\eta) = e^{-\gamma \eta} \left( \prod_{m=1}^N \frac{1}{(\eta/a_m; q)_\infty} \right) \left( \prod^N_{i = 1} \frac{(\alpha_i \eta; q)_{\infty}}{1 + \beta_i \eta} \right).
\end{equation}
Thus the function $F(\eta; q,k)$ can be  written in the closed form
\begin{equation}
F(\eta; q, k) = \frac{g(\eta)}{g(q^k \eta)}.
\end{equation}
Then as in \eqref{K-integral} we can write the kernel for the operator $\K(z; q)$ in the closed form
\begin{equation} \label{eq:alt_q-Whittaker}
 K(\xi, \eta; z; q) = \frac{1}{2\pi i} \int_{\delta-i\infty}^{\delta+i\infty} \frac{\pi z^s}{\sin(\pi s)} \frac{g(\eta)}{g(q^s\eta)} \frac{ds}{\xi -q^s\eta},
\end{equation}
which allows for analytic continuation to all $z\in \compC$. The special case of \eqref{eq:alt_q-Whittaker} with all $\alpha_i, \beta_i = 0$ is given in \cite[Theorem 3.18]{Borodin-Corwin13}, except that $\xi$ and $\eta$ are exchanged, which does not change the Fredholm determinant. 

\paragraph{$q$-TASEP}

The $q$-deformed Totally Asymmetric Simple Exclusion Process ($q$-TASEP) is a well-studied model in the KPZ universality class \cite{Borodin-Corwin-Sasamoto14}, \cite{Ferrari-Veto15}, \cite{Barraquand15}, \cite{Imamura-Sasamoto17}. It is also a continuous limit of the $q$-Whittaker processes \cite{Borodin-Corwin13}. We refer to \cite{Borodin-Corwin-Sasamoto14} for the definition of $q$-TASEP and for the meaning of the notations, and only remark that the speeds of the particles $x_1, \dotsc, x_n$ depend on parameters $a_1, \dotsc, a_n$.

In \cite[Theorem 3.13]{Borodin-Corwin-Sasamoto14}, with the so-called step initial condition, a moment generating function for the position of the $n$-th particle at time $t$ is provided as
\begin{equation}
  \E \left[ \frac{1}{(-z q^{x_n(t) + n}; q)_{\infty}} \right] = \frac{1}{(-z; q)_{\infty}} \det( I + z \M(q))_{L^2(\Gamma_{0, \A})},
\end{equation}
where the integral operator $\M(q)$ has the kernel $M(\xi, \eta; q)$ given by \eqref{eq:kernel_formula_contours} with
\begin{equation}
  f(\eta) = \left( \prod^n_{m = 1} \frac{a_m}{a_m - \eta} \right) \exp((q - 1) t \eta),
\end{equation}
and the contour $\Gamma_{0, \A}$ (denoted by $\tilde{C}_a$ in \cite{Borodin-Corwin-Sasamoto14}) is a star-shaped contour centered at $0$ and enclosing $\A = \{ a_1, \dotsc, a_n \}$. Hence suppose there exists a contour $\Gamma_{\A}$ that encloses $\A$ but not $0$, then by Lemma \ref{lem:contour_det}, we have the alternative moment generating function
\begin{equation}
  \E \left[ \frac{1}{(-z q^{x_n(t) + n}; q)_{\infty}} \right] = \det( I + \K(z; q))_{L^2(\Gamma_{\A})},
\end{equation}
where $\K(z; q)$ is the integral operator with kernel \eqref{eq:K_kernel_contour}. Analogous to \eqref{eq:alt_q-Whittaker}, we can write the kernel
\begin{equation}
  K(\xi, \eta; z; q) = \frac{1}{2\pi i} \int_{\delta-i\infty}^{\delta+i\infty} \frac{\pi z^s}{\sin(\pi s)} \frac{g(\eta)}{g(q^s\eta)} \frac{ds}{\xi -q^s\eta}, \quad g(\eta) = e^{-t \eta} \prod_{m=1}^n \frac{1}{(\eta/a_m; q)_\infty}.
\end{equation}
This reproduces the formula \cite[Theorem 3.12]{Borodin-Corwin-Sasamoto14}, again up to the exchange of the variables $\xi$ and $\eta$. 

\paragraph{$q$-TAZRP}

The $q$-deformed Totally Asymmetric Zero Range process ($q$-TAZRP) is a dual process to the $q$-TASEP, see \cite{Korhonen-Lee14}, \cite{Wang-Waugh16} and \cite{Lee-Wang17} for a detailed definition of the model and the duality. The $q$-TAZRP was originally defined in \cite{Sasamoto-Wadati98} with the name $q$-boson process.

Let the (inhomogeneous) $q$-TAZRP be defined as in \cite{Lee-Wang17}, with the conductance of the sites given by $a_k = (1 - q)^{-1}b_k$ ($k \in \intZ$), and assume that the particle number is $N$ and the initial condition is the so-called step initial condition that $x_1(0) = \dotsb = x_N(0) = 0$. Then the distribution of the leftmost particle $x_N$ at time $t > 0$ is by \cite[Proposition 2.1, Formulas (117) and (118)]{Lee-Wang17}
\begin{equation}
  \begin{split}
    & \Prob_{0^N}(x_N(t) > M) \\
    = {}& \frac{1}{(2\pi i)^N} \int_{\Gamma_{0, \A}} \frac{dw_1}{dw_1} \dotsi \int_{\Gamma_{0, \A}} \frac{dw_N}{dw_N} \prod_{1 \leq i < j \leq N} \frac{w_i - w_j}{qw_i - w_j} \prod^N_{j = 1} \left[ \prod^M_{k = 0} \left( \frac{b_k}{b_k - w_j} \right) e^{-w_j t} \right] \\
    = {}& \frac{[N]_q!}{N!}\frac{(q - 1)^N q^{-N(N - 1)/2}}{(2\pi i)^N} \int_{\Gamma_{0, \A}} \frac{dw_1}{dw_1} \dotsi \int_{\Gamma_{0, \A}} \frac{dw_N}{dw_N} \det \left( \frac{1}{qw_k - w_j} \right)^N_{j, k = 1} \\
    & \phantom{\frac{[N]_q!}{N!}\frac{(q - 1)^N q^{-N(N - 1)/2}}{(2\pi i)^N} \int_{\Gamma_{0, \A}} \frac{dw_1}{dw_1} \dotsi \int_{\Gamma_{0, \A}}} \times \prod^N_{j = 1} \left[ \prod^M_{k = 0} \left( \frac{b_k}{b_k - w_j} \right) e^{-w_j t} \right] \\
    = {}& [N]_q! \frac{(q - 1)^N}{q^{N(N - 1)/2}} \frac{1}{2\pi i} \oint_0 \det(1 + z\M(q)_{L^2(\Gamma_{0, \A})}) \frac{dz}{z^{N + 1}},
  \end{split}
\end{equation}
where $\M(q)$ is the integral operator on $L^2(\Gamma_{0, \A})$ with kernel $M(\xi, \eta; q)$ given in \eqref{eq:kernel_formula_contours} with the function $f$ specified as
\begin{equation} \label{eq:f_q_TAZRP}
  f(\xi) = \prod^M_{k = 0} \left( \frac{b_k}{b_k - \xi} \right) e^{-\xi t},
\end{equation}
and the contour $\Gamma_{0, \A}$ is the same as the contour $C$ in \cite[Proposition 2.1]{Lee-Wang17} that is a large enough circle. Then applying Lemma \ref{lem:contour_det}, we have that
\begin{equation}
  \Prob_{0^N}(x_N(t) > M) = [N]_q! \frac{(q - 1)^N}{q^{N(N - 1)/2}} \frac{1}{2\pi i} \oint_0 \frac{(-z; q)_{\infty}}{z^{N + 1}} \det(1 + \K(z;q))_{L^2(\Gamma_{\A})}\, dz,
\end{equation}
where $\K(z; q)$ is the integral operator on $L^2(\Gamma_{\A})$ with kernel
\begin{equation}
  K(\xi, \eta; z; q) = \frac{1}{2\pi i} \int_{\delta-i\infty}^{\delta+i\infty} \frac{\pi z^s}{\sin(\pi s)} \frac{g(\eta)}{g(q^s\eta)} \frac{ds}{\xi -q^s\eta}, \quad g(\eta) = e^{-\frac{t}{1 - q} \eta} \prod_{m=0}^M \frac{1}{(\eta/b_m; q)_\infty},
\end{equation}
and the contour $\Gamma_{\A}$ encloses $b_1, \dotsc, b_N$ counterclockwise and satisfies \eqref{eq:essential_condition_of_Gamma}, given that such $\Gamma_{\A}$ exists.

\paragraph{ASEP}

Asymmetric Simple Exclusion Process (ASEP) is another important interacting particle system in the KPZ universality class. In \cite{Tracy-Widom09a}, the distribution of the $m$-th rightmost particle $x_m$ is derived with the Bernoulli initial condition that all positive sites are initially occupied with probability $\rho$ and all non-positive sites are empty initially. As a special $\rho = 1$ case, the step initial condition $x_n(0) = n$ is considered in \cite{Tracy-Widom08a}. Let $\Gamma_{\A}$ be a small counterclockwise circle around $-1/\tau$, where $\tau = p/q$ such that $p$ and $q = 1 - p$ are the right jumping rate and left jumping rate respectively for a particle, and $\M(\tau)$ be an integral operator on $L^2(\Gamma_{\A})$ defined by the kernel $M(\xi, \eta; \tau)$ in the form of \eqref{eq:kernel_formula_contours}, with $q$ replaced by $\tau$ (since $q$ is reserved as the left hopping rate in ASEP), with
\begin{equation}
  f(\eta) = \left( \frac{1 + \eta}{1 + \eta/\tau} \right)^x e^{-\frac{1 - \tau}{1 + \tau} t \left( \frac{1}{1 + \eta/\tau} - \frac{1}{1 + \eta} \right)} \frac{1}{1 - \frac{\eta}{\theta \tau}}, \quad \text{where} \quad \theta = \frac{\rho}{1 - \rho}.
\end{equation}
Then we have
\begin{equation} \label{eq:ASEP_Gamma_A}
  \Prob(x_m(t) \leq x) = \frac{1}{2\pi i} \oint \frac{\det(I + z \M(\tau))_{L^2(\Gamma_{\A})}}{(-z; \tau)_m} \frac{dz}{z},
\end{equation}
where the contour is a large enough circle.
Formula \eqref{eq:ASEP_Gamma_A} is equivalent to \cite[Formula (1)]{Tracy-Widom08a} with the change of variables $\lambda \mapsto -z$. The equivalence of $\det(I + z\M(\tau))$ and $\det(I - \lambda q K)$ there can be seen by the change of variables $\xi \mapsto (1 + \eta)/(1 + \eta/\tau)$ and $\xi' \mapsto (1 + \xi)/(1 + \xi/\tau)$, where $\xi, \xi'$ are notations in the definition of $K$ in \cite[Section 2]{Tracy-Widom08a}. 

By Lemma \ref{lem:contour_det}, we have that if $\Gamma_{0, \A}$ is a contour enclosing $0$ and $-1/\tau$, then
\begin{equation}
  \Prob(x_m(t) \leq x) = \frac{1}{2\pi i} \oint (-z q^m; \tau)_{\infty} \det(I + \K(z; \tau))_{L^2(\Gamma_{0, \A})} \frac{dz}{z},
\end{equation}
where $\K(z; \tau)$ is the integral operator on $L^2(\Gamma_{0, \A})$ with kernel
\begin{equation}
  K(\xi, \eta; z; \tau) = \frac{1}{2\pi i} \int_{\delta-i\infty}^{\delta+i\infty} \frac{\pi z^s}{\sin(\pi s)} \frac{g(\eta)}{g(\tau^s\eta)} \frac{ds}{\xi -\tau^s\eta}, \quad g(\eta) = (1 + \eta/\tau)^{-x} e^{\frac{1 - \tau}{1 + \tau} \frac{t}{1 + \eta/\tau}} \frac{1}{(\eta/(\theta \tau), \tau)_{\infty}}.
\end{equation}
This formula is equivalent to \cite[Corollary 5.4]{Borodin-Corwin-Sasamoto14}.

\section{Multi-time correlation functions and gap probabilities} \label{sec:multi_time}

In this section, we prove Theorem \ref{thm:multi_formulas}. Our derivation of the multi-time correlation functions is based on \cite[Formulas (50), (51) and (52)]{Le_Doussal-Majumdar-Schehr17}, while our derivation of the multi-time gap probabilities is based on \cite[Formulas (60) and (61)]{Le_Doussal-Majumdar-Schehr17}.

Since the model of free fermions at finite temperature and that of time-periodic nonintersecting OU processes are equivalent, we prove Theorem \ref{thm:multi_formulas} for free fermions at finite temperature.
% Our goal is modest: By collecting formulas in \cite{Le_Doussal-Majumdar-Schehr17}, one can already write down explicit formulas for the multi-time correlation functions and gap probability, and we just express these formulas in compact forms analogous to \eqref{eq:formula_for_R^m} and \eqref{eq:general_gap_prob}. Our formulas for the multi-time correlation functions and multi-time gap probabilities are amenable to asymptotic analysis in the way that \eqref{eq:formula_for_R^m} and \eqref{eq:general_gap_prob} yiled Theorems \ref{thm:edge} and \ref{thm:bulk}. However, we leave the computation of the limiting multi-time correlations and limiting gap probabilities in a furture publication.

\subsection{Multi-time correlation functions}

Let $\tau_1, \tau_2, \dotsc, \tau_m \in [0, \beta)$ be imaginary times in Proposition \ref{prop:multi}. Since the free fermions at finite temperature are distributed over eigenstates with respect to the Boltzmann distribution as in \cite[Formula (78)]{Le_Doussal-Majumdar-Schehr17}, The multi-time $m$-correlation function at $x_i$ and time $\tau_i$ with $i = 1, \dotsc, m$ is expressed as
\begin{equation} \label{eq:m-corr_decomp}
  R^{(m)}_n(x_1, \dotsc, x_m; \tau_1, \dotsc, \tau_m) = \frac{q^{n/2}}{Z_n(q)} \sum_{0 \leq k_1 < k_2 < \dotsb < k_n} q^{k_1 + \dotsb + k_n} R^{(m)}_{k_1, \dotsc, k_n}(x_1, \dotsc, x_m; \tau_1, \dotsc, \tau_m),
\end{equation}
where $R^{(m)}_{k_1, \dotsc, k_n}(x_1, \dotsc, x_m; \tau_1, \dotsc, \tau_m)$ is the multi-time $m$-correlation function for the $n$-particle model with eigenstate $(k_1, \dotsc, k_n)$. By \cite[Formulas (50)--(52)]{Le_Doussal-Majumdar-Schehr17}, We have
\begin{equation}
  R^{(m)}_{k_1, \dotsc, k_n}(x_1, \dotsc, x_m; \tau_1, \dotsc, \tau_m) = \det \left( K_{k_1, \dotsc, k_n}(x_i, x_j; \tau_i, \tau_j) \right)^m_{i, j = 1},
\end{equation}
where
\begin{equation} \label{eq:multi_time_corr_kernel}
  K_{k_1, \dotsc, k_n}(x, y; \tau, \sigma) = \tilde{K}_{k_1, \dotsc, k_n}(x, y; \tau, \sigma) - E(x, y; \tau, \sigma),
\end{equation}
such that $E(x, y; \tau, \sigma)$ is defined in \eqref{eq:defn_E}, and
\begin{equation}
  \tilde{K}_{k_1, \dotsc, k_n}(x, y; \tau, \sigma) = \sum^n_{i = 1} \varphi_{k_i}(x) \varphi_{k_i}(y) e^{k_i(\tau - \sigma)}.
\end{equation}
% Note that if $\tau = \sigma$, $K_{k_1, \dotsc, k_n}(x, y; \tau, \sigma)$ degenerates into $K_{k_1, \dotsc, k_n}(x, y)$ in \eqref{eq:defn_R^(m)_and_K_k_1--k_n}, and if $\tau_1, \dotsc, \tau_m$ are identical, then $R^{(m)}_{k_1, \dotsc, k_n}(x_1, \dotsc, x_m; \tau_1, \dotsc, \tau_m)$ degenerates into $R^{(m)}_{k_1, \dotsc, k_n}(x_1, \dotsc, x_m)$ in \eqref{eq:defn_R^(m)_and_K_k_1--k_n}.

Then % analogous to \eqref{eq:expansion_R^(m)},
it is straightforward to write (with $k_0 = -1$)
\begin{equation} \label{eq:expansion_R^(m)_x_i_tau_i}
  \begin{split}
     R^{(m)}_{k_1, \dotsc, k_n}(x_1, \dotsc, x_m; \tau_1, \dotsc, \tau_m) 
    = {}& \sum^n_{i_1 = 0} \dotsm \sum^n_{i_m = 0} \det \left( \hat{K}_{k_{i_l}}(x_j, x_l; \tau_j, \tau_l) \right)^m_{j, l = 1} \\
    = {}& \sum_{\substack{j_1 < j_2 < \dotsb < j_m \\ \{ j_1, \dotsc, j_m \} \subseteq \{ k_0 = -1, k_1, \dotsc, k_n \}}} \hat{R}^{(m)}_{j_1, \dotsc, j_m}(x_1, \dotsc, x_m; \tau_1, \dotsc, \tau_m),
  \end{split}
\end{equation}
where
\begin{equation}
  \hat{K}_k(x_j, x_l; \tau_j, \tau_l) =
  \begin{cases}
    \varphi_k(x_j) \varphi_k(x_l) e^{k(\tau_j - \tau_l)} & \text{if $k \geq 0$}, \\
    - E(x_j, x_l; \tau_j, \tau_l) & \text{if $k = -1$},
  \end{cases}
\end{equation}
% is a generalization of $\varphi_k(x_j) \varphi_k(x_l)$ occurring in the entries of the determinant in \eqref{eq:expansion_R^(m)}, 
and
\begin{equation}
  \begin{split}
    \hat{R}^{(m)}_{j_1, \dotsc, j_m}(x_1, \dotsc, x_m; \tau_1, \dotsc, \tau_m; n) = {}& \det \left( \sum^m_{i = 1} \hat{K}_{j_i}(x_k, x_l; \tau_k, \tau_l; n) \right)^m_{k, l = 1} \\
    = {}& \sum_{\kappa, \lambda \in S_n}  \sgn(\lambda) \prod^m_{i = 1} \hat{K}_{j_{\kappa(i)}}(x_i, x_{\lambda(i)}; \tau_i, \tau_{\lambda(i)}; n).
  \end{split}
\end{equation}
% is a generalization of $\hat{R}^{(m)}_{k_1, \dotsc, k_n}(x_1, \dotsc, x_m)$ in \eqref{eq:expansion_R^(m)_in_hat_R^(m)}. 
Hence % analogous to \eqref{eq:expansion_R_1},
we have
\begin{multline}
  R^{(m)}_n(x_1, \dotsc, x_m; \tau_1, \dotsc, \tau_m) = \frac{q^{n/2}}{Z_n(q)} \left( \sum_{0 \leq j_1 < \dotsb < j_m} C_{j_1, \dotsc, j_m} \hat{R}^{(m)}_{j_1, \dotsc, j_m}(x_1, \dotsc, x_m; \tau_1, \dotsc, \tau_m) \right. \\
  + \left. \sum_{0 \leq i_1 < \dotsb < i_{m - 1}} C_{i_1, \dotsc, i_{m - 1}} \hat{R}^{(m)}_{-1, i_1, \dotsc, i_{m - 1}}(x_1, \dotsc, x_m; \tau_1, \dotsc, \tau_m) \right),
\end{multline}
where $C_{j_1, \dotsc, j_m}$ and $C_{i_1, \dotsc, i_{m - 1}}$ are defined as
\begin{equation} \label{eq:defn_C_j_1--j_m}
 C_{j_1, \dotsc, j_m} = \sum_{\substack{0 \leq k_1 < \dotsb < k_n \\ \{ k_1, \dotsc, k_n \} \supseteq \{ j_1, \dotsc, j_m \}}} q^{k_1 + \dotsb + k_n}.
\end{equation}
Now we state the explicit formula of $C_{j_1, \dotsc, j_m}$ and postpone its proof to the end of this section:
\begin{equation} \label{eq:explicit_C_j_1_j_m}
 C_{j_1, \dotsc, j_m} = q^{j_1 + \dotsb + j_m} \frac{1}{2\pi i} \oint_0 \frac{dz}{z^{n - m + 1}} \left[ \prod^{\infty}_{k = 0} (1 + q^k z) \right] \left[ \prod^m_{i = 1} (1 + q^{j_i} z)^{-1} \right].
\end{equation}
We note that the contour in \eqref{eq:explicit_C_j_1_j_m} can be any one that encloses $0$ in positive orientation, for the integrand has only one pole at $0$.

% On the other hand, we consider the kernel $K(x, y; z; q)$ defined in \eqref{eq:kernel_for_each_z}. It is straightforward to get by the Cauchy--Binet identity that for all $z \neq -q^{-k}$, $k = 0, 1, 2, \dotsc$,
% \begin{equation} \label{eq:expansion_det_K(x,y;z)}
%  \det(K(x_i, x_j; z; q))^m_{i, j = 1} = \sum_{0 \leq j_1 < \dotsb < j_m} q^{j_1 + \dotsb + j_m} z^m \left[ \prod^m_{i = 1} (1 + q^{j_i} z)^{-1} \right] \hat{R}^{(m)}_{j_1, \dotsc, j_m}(x_1, \dotsc, x_m).
% \end{equation}
% Hence by comparing \eqref{eq:expansion_R_1}, \eqref{eq:explicit_C_j_1_j_m} and \eqref{eq:expansion_det_K(x,y;z)}, we prove the desired identity \eqref{eq:formula_for_R^m}, where the contour can be any one that encloses $0$ in positive orientation. By \eqref{eq:expansion_det_K(x,y;z)} we have that the points $z = -q^{-k}$ are first order poles of $\det(K(x_i, x_j; z; q))^m_{i, j = 1}$. Hence they are not the poles of the integrand.

% \begin{rmk} \label{rmk:vanishing}
%  If $m > n$ in \eqref{eq:formula_for_R^m}, we can see by \eqref{eq:expansion_det_K(x,y;z)} that $R^{(m)}_n(x_1, \dotsc, x_m) = 0$ for all $x_1, \dotsc, x_m$. This confirms that there are no more than $n$ particles in the model.
% \end{rmk}

in \eqref{eq:defn_C_j_1--j_m}. On the other hand, we have
\begin{multline} \label{eq:expansion_det_K(x,y;tau,sigma;z)}
  \det \left( K(x_i, x_j; \tau_i, \tau_j; z; q) \right)^m_{i, j = 1} \\
  = \sum_{0 \leq j_1 < \dotsb < j_m} q^{j_1 + \dotsb + j_m} z^m \left[ \prod^m_{l = 1} (1 + q^{j_l} z)^{-1} \right] \hat{R}^{(m)}_{j_1, \dotsc, j_m}(x_1, \dotsc, x_m; \tau_1, \dotsc, \tau_m) \\
  + \sum_{0 \leq i_1 < \dotsb < i_{m - 1}} q^{i_1 + \dotsb + i_{m - 1}} z^{m - 1} \left[ \prod^{m - 1}_{l = 1} (1 + q^{i_l} z)^{-1} \right] \hat{R}^{(m)}_{-1, i_1, \dotsc, i_{m - 1}}(x_1, \dotsc, x_m; \tau_1, \dotsc, \tau_m),
\end{multline}
where $K(x, y; \tau, \sigma; z; q)$ is defined in \eqref{eq:kernel_K(xytausigmazq)}. Hence analogous to \eqref{eq:formula_for_R^m}, we prove \eqref{eq:contour_integral_repr_m-corr} by comparing \eqref{eq:expansion_R^(m)_x_i_tau_i}, \eqref{eq:explicit_C_j_1_j_m} and \eqref{eq:expansion_det_K(x,y;tau,sigma;z)}. 

\begin{proof}[Proof of equation \eqref{eq:explicit_C_j_1_j_m}]
  In this proof, we use the notational convention that $j_0 = -1$ and $j_{m + 1} = \infty$.
  
  By definition,
  \begin{equation}
    C_{j_1, \dotsc, j_m} = q^{j_1 + \dotsb + j_m} \sum_{\substack{0 \leq l_0 \leq j_1,\ 0 \leq l_1 \leq j_2 - j_1 - 1, \dotsc, 0 \leq l_{m - 1} \leq j_m - j_{m - 1} - 1,\\ 0 \leq l_m = n - m - (l_0 + l_1 + \dotsb + l_{m - 1})}} \prod^m_{i = 0} g_i(l_i),
  \end{equation}
  where for $k = 0, 1, \dotsc, m$
  \begin{equation}
    g_k(l) = \sum_{j_k < i_1 < \dotsb < i_l < j_{k + 1}} q^{i_1 + \dotsb + i_l}.
  \end{equation}
  For $k = 0, 1, \dotsc, m$, letting
  \begin{equation}
    G_k(z) = \sum^{j_{k + 1} - j_k - 1}_{l = 0} g_k(l) z^l,
  \end{equation}
  we find
  \begin{equation} \label{eq:C_j_1_j_m_integral}
    C_{j_1, \dotsc, j_m} = q^{j_1 + \dotsb + j_m} \frac{1}{2\pi i} \oint_0 \frac{dz}{z^{n - m + 1}} \prod^m_{k = 0} G_k(z).
  \end{equation}
  Now we compute $G_k(z)$. Inductively, we compute that
  \begin{equation}
    g_m(l) = \frac{q^{\binom{l}{2}}}{(q; q)_l} q^{l(j_m + 1)}.
  \end{equation}
  and similarly we can obtain that for $k = 0, 1, \dotsc, m - 1$
  \begin{equation}
    g_k(l) = \qbinom{j_{k + 1} - j_k - 1}{l} q^{\binom{l}{2}} q^{l(j_k + 1)},
  \end{equation}
  with the understanding that $j_0 = -1$.
  Hence by \cite[Corollary 10.2.2(b)]{Andrews-Askey-Roy99} for the $k = m$ case, and \cite[Corollary 10.2.2(c)]{Andrews-Askey-Roy99} for the $k = 0, 1, \dotsc, m - 1$ cases, we have
  \begin{align} \label{eq:expr_G_k}
    G_k(z) = (-q^{j_k + 1}z; q)_{j_{k + 1} - j_k - 1} = \prod^{j_{k + 1} - 1}_{l = j_k + 1} (1 + q^l z).
  \end{align}
  Hence
  \begin{equation} \label{eq:prodct_G_k}
    \prod^m_{k = 0} G_k(z) = \left[ \prod^{\infty}_{l = 0} (1 + q^l z) \right] \left[ \prod^m_{i = 1} (1 + q^{j_i} z)^{-1} \right],
  \end{equation}
  and we prove \eqref{eq:explicit_C_j_1_j_m} by plugging \eqref{eq:prodct_G_k} into \eqref{eq:C_j_1_j_m_integral}.
\end{proof}

\subsection{Multi-time gap probability}

Next we assume the times $\tau_1, \dotsc, \tau_m$ are distinct, and compute the gap probability for free fermions at finite temperature such that at times $\tau_1, \dotsc, \tau_m$, all particles are in the measurable sets $A_1, \dotsc, A_m$ respectively. We denote this probability by $\Prob_n(A_1, \dotsc, A_m; \tau_1, \dotsc, \tau_m)$. According to the Boltzmann distribution of eigenstates, we have that \cite[Formula (78)]{Le_Doussal-Majumdar-Schehr17}
\begin{equation} \label{eq:multi-time_gap_prob_decomp}
  \Prob_n(A_1, \dotsc, A_m; \tau_1, \dotsc, \tau_m) = \frac{q^{n/2}}{Z_n(q)} \sum_{0 \leq k_1 < k_2 < \dotsb k_n} q^{k_1 + \dotsb + k_n} \Prob_{k_1, \dotsc, k_n}(A_1, \dotsc, A_m; \tau_1, \dotsc, \tau_m),
\end{equation}
where $\Prob_{k_1, \dotsc, k_n}(A_1, \dotsc, A_m; \tau_1, \dotsc, \tau_m)$ is the gap probability that all particles ate in $A_1, \dotsc, A_m$ at times $\tau_1, \dotsc, \tau_m$ respectively for the determinantal point process characterized by the multi-time correlation kernel $K_{k_1, \dotsc, k_n}(x, y; \tau, \sigma)$ in \eqref{eq:multi_time_corr_kernel}. By \cite[Formulas (60) and (61)]{Le_Doussal-Majumdar-Schehr17}, we have
\begin{equation} \label{eq:multi-time_comp_Fred_det}
  \Prob_{k_1, \dotsc, k_n}(A_1, \dotsc, A_m; \tau_1, \dotsc, \tau_m) = \det(I - \K_{k_1, \dotsc, k_n}(\tau_1, \dotsc, \tau_m) \chi_{A^c_1, \dotsc, A^c_m}),
\end{equation}
where $\K_{k_1, \dotsc, k_n}(\tau_1, \dotsc, \tau_m)$ is, analogous to $\K(\tau_1, \dotsc, \tau_m; z; q)$ in \eqref{eq:defn_final_K_op}, an integral operator on $L^2(\realR \times \{ 1, 2, \dotsc, m \})$ whose kernel is represented by an $m \times m$ matrix $(K_{k_1, \dotsc, k_n}(x_i, x_j; \tau_i, \tau_j))^m_{i, j = 1}$, and
\begin{equation}
  (\K_{k_1, \dotsc, k_n}(\tau_1, \dotsc, \tau_m)f)(x; k) = \sum^m_{j = 1} \int_{\realR} K_{k_1, \dotsc, k_n}(x, y; \tau_k, \tau_j) f(y; j) dy.
\end{equation}

By \cite[Formula (3.5)]{Simon05}, we have the expansion 
\begin{equation} \label{eq:formal_expan_Fred_det}
  \det(I - \K_{k_1, \dotsc, k_n}(\tau_1, \dotsc, \tau_m) \chi_{A^c_1, \dotsc, A^c_m}) = 1 + \sum^{\infty}_{l = 1} \frac{(-1)^l}{l!} \Tr \left[ \Lambda^l \left( \K_{k_1, \dotsc, k_n}(\tau_1, \dotsc, \tau_m) \chi_{A^c_1, \dotsc, A^c_m} \right) \right],
\end{equation}
where the trace of the $l$-th exterior power of $\K_{k_1, \dotsc, k_n}(\tau_1, \dotsc, \tau_m) \chi_{A^c_1, \dotsc, A^c_m}$ is computed as
\begin{equation} \label{eq:computation_trace_Lambda_multi-time}
  \begin{split}
    & \Tr \left[ \Lambda^l \left( \K_{k_1, \dotsc, k_n}(\tau_1, \dotsc, \tau_m) \chi_{A^c_1, \dotsc, A^c_m} \right) \right] \\
    = {}& \sum^m_{s_1 = 1} \sum^m_{s_2 = 1} \dotsi \sum^m_{s_l = 1} \int_{A^c_{s_1}} dx_1 \int_{A^c_{s_2}} dx_2 \dotsi \int_{A^c_{s_l}} dx_l \det(K_{k_1, \dotsc, k_n}(x_i, x_j; \tau_{s_i}, \tau_{s_j}))^l_{i, j = 1} \\
    = {}& \sum^m_{s_1 = 1} \sum^m_{s_2 = 1} \dotsi \sum^m_{s_l = 1} \int_{A^c_{s_1}} dx_1 \int_{A^c_{s_2}} dx_2 \dotsi \int_{A^c_{s_l}} dx_l R^{(m)}_{k_1, \dotsc, k_n}(x_1, \dotsc, x_l; \tau_{s_1}, \dotsc, \tau_{s_l}).
  \end{split}
\end{equation}
The proof of \eqref{eq:computation_trace_Lambda_multi-time} is analogous to that of \cite[Theorem 3.10]{Simon05}.

By \eqref{eq:multi-time_gap_prob_decomp}, \eqref{eq:multi-time_comp_Fred_det}, \eqref{eq:formal_expan_Fred_det} and \eqref{eq:computation_trace_Lambda_multi-time}, we have
\begin{multline}
  \Prob_n(A_1, \dotsc, A_m; \tau_1, \dotsc, \tau_m) = \frac{q^{n/2}}{Z_n(q)} \sum_{0 \leq k_1 < k_2 < \dotsb k_n} q^{k_1 + \dotsb + k_n} \cdot 1 \\
  + \sum^{\infty}_{l = 1} \frac{(-1)^l}{l!} \sum^m_{s_1 = 1} \dotsi \sum^m_{s_l = 1} \int_{A^c_{s_1}} dx_1 \dotsi \int_{A^c_{s_l}} dx_l \\
  \frac{q^{n/2}}{Z_n(q)} \sum_{0 \leq k_1 < k_2 < \dotsb k_n} q^{k_1 + \dotsb + k_n} R^{(m)}_{k_1, \dotsc, k_n}(x_1, \dotsc, x_l; \tau_{s_1}, \dotsc, \tau_{s_l}).
\end{multline}
By \eqref{eq:explicit_Z_n(q)}, \eqref{eq:m-corr_decomp} and \eqref{eq:contour_integral_repr_m-corr}, we have
\begin{equation}
  \begin{split}
    & \Prob_n(A_1, \dotsc, A_m; \tau_1, \dotsc, \tau_m) \\
    = {}& 1 + \sum^{\infty}_{l = 1} \frac{(-1)^l}{l!} \sum^m_{s_1 = 1} \dotsi \sum^m_{s_l = 1} \int_{A^c_{s_1}} dx_1 \dotsi \int_{A^c_{s_l}} dx_l R^{(m)}(x_1, \dotsc, x_m; \tau_1, \dotsc, \tau_m) \\
    = {}& 1 + \sum^{\infty}_{l = 1} \frac{(-1)^l}{l!} \sum^m_{s_1 = 1} \dotsi \sum^m_{s_l = 1} \int_{A^c_{s_1}} dx_1 \dotsi \int_{A^c_{s_l}} dx_l \frac{1}{2\pi i} \oint_0 F(z) \det(K(x_i, x_j; \tau_{s_i}, \tau_{s_j}; z; q))^l_{i, j = 1} \frac{dz}{z} \\
    = {}& \frac{1}{2\pi i} \oint_0 F(z) \left[ 1 + \sum^{\infty}_{l = 1} \frac{(-1)^l}{l!} \sum^m_{s_1 = 1} \dotsi \sum^m_{s_l = 1} \int_{A^c_{s_1}} dx_1 \dotsi \int_{A^c_{s_l}} dx_l \det(K(x_i, x_j; \tau_{s_i}, \tau_{s_j}; z; q))^l_{i, j = 1} \right] \frac{dz}{z} \\
    = {}& \frac{1}{2\pi i} \oint_0 F(z) \det(I - \K(\tau_1, \dotsc, \tau_m; z; q) \chi_{A^c_1, \dotsc, A^c_m}) \frac{dz}{z},
  \end{split}
\end{equation}
and prove Theorem \ref{thm:multi_formulas}\ref{enu:thm:multi_formulas_b}.

\appendix
  
\section{Equivalence to MNS random matrix model} \label{sec:relation_to_MNS_RM}

With the help of the two integral representations of Hermite polynomials (\cite[22.10.9 and 22.10.15]{Abramowitz-Stegun64}):
\begin{align}
 H_k(x) = {}& \frac{1}{\sqrt{2\pi} i} \int^{i\infty}_{-i\infty} e^{\frac{1}{2} (s - x)^2} s^k ds, \label{eq:contour_integral_H_k_1st} \\
 \frac{1}{k!} H_k(x) e^{-x^2/2} = {}& \frac{1}{2\pi i} \oint_{\Gamma} e^{-\frac{1}{2} (t - x)^2} t^{-k} \frac{dt}{t}, \label{eq:contour_integral_H_k_2nd}
\end{align}
where $\Gamma$ is a contour around $0$ with positive orientation, the density function $P_n(x_1, \dotsc, x_n)$ in \eqref{eq:density_n_particle} is expressed as
\begin{equation} \label{eq:alt_density_n_particle}
  \begin{split}
    P_n(x_1, \dotsc, x_n) = {}& \frac{q^{n/2}}{n! Z_n(q)} \sum^{\infty}_{k_1, \dotsc, k_n = 0}
    \begin{vmatrix}
      \varphi_{k_1}(x_1) & \dots & \varphi_{k_1}(x_n) \\
      \vdots & & \vdots \\
      \varphi_{k_n}(x_1) & \dots & \varphi_{k_n}(x_n) \\
    \end{vmatrix}^2
    q^{k_1 + \dotsb + k_n} \\
    = {}& \frac{q^{n/2}}{n! Z_n(q)} \sum^{\infty}_{k_1, \dotsc, k_n = 0} q^{k_1 + \dotsb + k_n}
    \begin{vmatrix}
      H_{k_1}(x_1) & \dots & H_{k_1}(x_n) \\
      \vdots & & \vdots \\
      H_{k_n}(x_1) & \dots & H_{k_n}(x_n)
    \end{vmatrix} \\
    & \times
    \begin{vmatrix}
      \frac{1}{\sqrt{2\pi} k_1 !} H_{k_1}(x_1) e^{-x^2_1/2} & \dots & \frac{1}{\sqrt{2\pi} k_1 !} H_{k_1}(x_n) e^{-x^2_n/2} \\
      \vdots & & \vdots \\
      \frac{1}{\sqrt{2\pi} k_n !} H_{k_n}(x_1) e^{-x^2_1/2} & \dots & \frac{1}{\sqrt{2\pi} k_n !} H_{k_n}(x_n) e^{-x^2_n/2}
    \end{vmatrix} \\
    = {}& \frac{q^{n/2}}{n! Z_n(q)} \sum^{\infty}_{k_1, \dotsc, k_n = 0} \frac{1}{(\sqrt{2\pi} i)^n} \int^{i\infty}_{-i\infty} ds_1 \dotsi \int^{i\infty}_{-i\infty} ds_n \prod^n_{j = 1} e^{\frac{1}{2} (s_j - x_j)^2} \det \left( (qs)^{k_l}_j \right)^n_{j, l = 1} \\
    & \times \frac{1}{((2\pi)^{3/2} i)^n} \oint_{\Gamma} \frac{dt_1}{t_1} \dotsi \oint_{\Gamma} \frac{dt_n}{t_n} \prod^n_{j = 1} e^{-\frac{1}{2} (t_j - x_j)^2} \det \left( t^{-k_l}_j \right)^n_{j, l = 1} \\
    = {}& \frac{q^{n/2}}{Z_n(q)} \frac{1}{(2\pi i)^{2n}} \int^{i\infty}_{-i\infty} ds_1 \dotsi \int^{i\infty}_{-i\infty} ds_n \oint_{\Gamma} \frac{dt_1}{t_1} \dotsi \oint_{\Gamma} \frac{dt_n}{t_n} \prod^n_{j = 1} \frac{e^{\frac{1}{2} (s_j - x_j)^2}}{e^{\frac{1}{2} (t_j - x_j)^2}} \\
    & \times \sum^{\infty}_{k_1, \dotsc, k_n = 0} \sum_{\sigma \in S_n} \sgn(\sigma) \left( \frac{qs_{\sigma(j)}}{t_j} \right)^{k_j},
  \end{split}
\end{equation}
where in the first step we symmetrize the indices $k_1, \dotsc, k_n$, and in the last step we use the symmetry among $k_1, \dotsc, k_n$. Under the assumption that $\lvert t_j \rvert > q \lvert s_k \rvert$ for all $j, k$, We have
\begin{equation} \label{eq:summation_power_func}
  \begin{split}
    \sum^{\infty}_{k_1, \dotsc, k_n = 0} \sum_{\sigma \in S_n} \sgn(\sigma) \left( \frac{qs_{\sigma(j)}}{t_j} \right)^{k_j} = {}& \sum_{\sigma \in S_n} \sgn(\sigma) \prod^n_{j = 1} \frac{1}{1 - qs_{\sigma(j)}/t_j} \\
    = {}& \prod^n_{j = 1} t_j \det \left( \frac{1}{t_j - q s_l} \right)^n_{j, l = 1}.
  \end{split}
\end{equation}
Hence we deform the contour $\Gamma$ in \eqref{eq:alt_density_n_particle} into $\Gamma_s = \{ z \in \compC \mid \lvert z \rvert = \max(\lvert s_k \rvert)$ that depends on $s_1, \dotsc, s_n$, and plug \eqref{eq:summation_power_func} into \eqref{eq:alt_density_n_particle}. Using the residue theorem, we have
\begin{equation} \label{eq:compact_P_n}
  \begin{split}
    & P_n(x_1, \dotsc, x_n) \\
    = {}& \frac{q^{n/2}}{Z_n(q)} \frac{1}{(2\pi i)^{2n}} \int^{i\infty}_{-i\infty} ds_1 \dotsi \int^{i\infty}_{-i\infty} ds_n \oint_{\Gamma_s} dt_1 \dotsi \oint_{\Gamma_s} dt_n \prod^n_{j = 1} \frac{e^{\frac{1}{2} (s_j - x_j)^2}}{e^{\frac{1}{2} (t_j - x_j)^2}} \det \left( \frac{1}{t_j - q s_l} \right)^n_{j, l = 1} \\
    = {}& \frac{q^{n/2}}{Z_n(q)} \frac{1}{(2\pi i)^n} \int^{i\infty}_{-i\infty} ds_1 \dotsi \int^{i\infty}_{-i\infty} ds_n \prod^n_{j = 1} e^{\frac{1}{2} (s_j - x_j)^2} \sum_{\sigma \in S_n} \sgn(\sigma) \prod^n_{j = 1} e^{-\frac{1}{2} (qs_j - x_{\sigma(j)})^2} \\
    = {}& \frac{q^{n/2}}{Z_n(q)} \sum_{\sigma \in S_n} \sgn(\sigma) \frac{1}{(2\pi i)^n} \int^{i\infty}_{-i\infty} ds_1 \dotsi \int^{i\infty}_{-i\infty} ds_n \prod^n_{j = 1} \exp \left( \frac{1}{2}[(s_j - x_j)^2 - (qs_j - x_{\sigma(j)})^2] \right) \\
    = {}& \frac{q^{n/2}}{Z_n(q)} (2\pi(1 - q^2))^{-n/2} \sum_{\sigma \in S_n} \sgn(\sigma) e^{-\frac{1}{2(1 - q^2)} (x_{\sigma(j)} - qx_j)^2} \\
    = {}& \frac{q^{n/2}}{Z_n(q)} (2\pi(1 - q^2))^{-n/2} \prod^n_{j = 1} e^{-\frac{1}{2} \frac{1 + q^2}{1 - q^2} x^2_j} \sum_{\sigma \in S_n} \sgn(\sigma) e^{\frac{q}{1 - q^2} x_{\sigma(j)} x_j} \\
    = {}& \frac{q^{n/2}}{Z_n(q)} (2\pi(1 - q^2))^{-n/2} \prod^n_{j = 1} e^{-\frac{1}{2} \frac{1 + q^2}{1 - q^2} x^2_j} \det \left( e^{\frac{q}{1 - q^2} x_j x_k} \right)^n_{j, k = 1}.
  \end{split}
\end{equation}
Comparing the right-hand side of \eqref{eq:compact_P_n} with \cite[Formula (3)]{Moshe-Neuberger-Shapiro94}, we prove Proposition \ref{prop:Moshe-Neuberger-Shapiro94}.

% \bibliographystyle{plain}
% \bibliography{../../bibliography/bibliography}

\begin{thebibliography}{10}

\bibitem{Abramowitz-Stegun64}
Milton Abramowitz and Irene~A. Stegun.
\newblock {\em Handbook of mathematical functions with formulas, graphs, and
  mathematical tables}, volume~55 of {\em National Bureau of Standards Applied
  Mathematics Series}.
\newblock For sale by the Superintendent of Documents, U.S. Government Printing
  Office, Washington, D.C., 1964.

\bibitem{Amir-Corwin-Quastel11}
Gideon Amir, Ivan Corwin, and Jeremy Quastel.
\newblock Probability distribution of the free energy of the continuum directed
  random polymer in {$1+1$} dimensions.
\newblock {\em Comm. Pure Appl. Math.}, 64(4):466--537, 2011.

\bibitem{Anderson-Guionnet-Zeitouni10}
Greg~W. Anderson, Alice Guionnet, and Ofer Zeitouni.
\newblock {\em An introduction to random matrices}, volume 118 of {\em
  Cambridge Studies in Advanced Mathematics}.
\newblock Cambridge University Press, Cambridge, 2010.

\bibitem{Andrews-Askey-Roy99}
George~E. Andrews, Richard Askey, and Ranjan Roy.
\newblock {\em Special functions}, volume~71 of {\em Encyclopedia of
  Mathematics and its Applications}.
\newblock Cambridge University Press, Cambridge, 1999.

\bibitem{Barraquand15}
Guillaume Barraquand.
\newblock A phase transition for {$q$}-{TASEP} with a few slower particles.
\newblock {\em Stochastic Process. Appl.}, 125(7):2674--2699, 2015.

\bibitem{Borodin-Corwin13}
Alexei Borodin and Ivan Corwin.
\newblock Macdonald processes.
\newblock {\em Probab. Theory Related Fields}, 158(1-2):225--400, 2014.

\bibitem{Borodin-Corwin-Sasamoto14}
Alexei Borodin, Ivan Corwin, and Tomohiro Sasamoto.
\newblock From duality to determinants for {$q$}-{TASEP} and {ASEP}.
\newblock {\em Ann. Probab.}, 42(6):2314--2382, 2014.

\bibitem{Boulatov-Kazakov92}
Dmitri Boulatov and Vladimir Kazakov.
\newblock Vortex-antivortex sector of one-dimensional string theory via the
  upside-down matrix oscillator.
\newblock {\em Nuclear Phys. B Proc. Suppl.}, 25A:38--53, 1992.
\newblock Random surfaces and $2$D quantum gravity (Barcelona, 1991).

\bibitem{Corwin11}
Ivan Corwin.
\newblock The {K}ardar-{P}arisi-{Z}hang equation and universality class.
\newblock {\em Random Matrices Theory Appl.}, 1(1):1130001, 76, 2012.

\bibitem{Cunden-Mezzadri-OConnell17}
Fabio~Deelan Cunden, Francesco Mezzadri, and Neil O'Connell.
\newblock Free {F}ermions and the {C}lassical {C}ompact {G}roups.
\newblock {\em J. Stat. Phys.}, 171(5):768--801, 2018.

\bibitem{Dean-Le_Doussal-Majumdar-Schehr15}
David~S. Dean, Pierre Le~Doussal, Satya~N. Majumdar, and Gr\'egory Schehr.
\newblock Finite-temperature free fermions and the {K}ardar-{P}arisi-{Z}hang
  equation at finite time.
\newblock {\em Phys. Rev. Lett.}, 114:110402, Mar 2015.

\bibitem{Dean-Le_Doussal-Majumdar-Schehr16}
David~S. Dean, Pierre Le~Doussal, Satya~N. Majumdar, and Gr\'egory Schehr.
\newblock Noninteracting fermions at finite temperature in a $d$-dimensional
  trap: {U}niversal correlations.
\newblock {\em Phys. Rev. A}, 94:063622, Dec 2016.

\bibitem{Le_Doussal-Majumdar-Schehr17}
Pierre~Le Doussal, Satya~N. Majumdar, and Gr\'egory Schehr.
\newblock Periodic {A}iry process and equilibrium dynamics of edge fermions in
  a trap.
\newblock {\em Ann. Physics}, 383:312--345, 2017.

\bibitem{Ferrari-Veto15}
Patrik~L. Ferrari and B\'alint Vet\H{o}.
\newblock Tracy-{W}idom asymptotics for {$q$}-{TASEP}.
\newblock {\em Ann. Inst. Henri Poincar\'e Probab. Stat.}, 51(4):1465--1485,
  2015.

\bibitem{Fetter-Walecka12}
Alexander~L. Fetter and John~Dirk Walecka.
\newblock {\em Quantum theory of many-particle systems}.
\newblock Courier Corporation, 2012.

\bibitem{Imamura-Sasamoto15}
Takashi Imamura and Tomohiro Sasamoto.
\newblock Determinantal structures in the {O}'{C}onnell-{Y}or directed random
  polymer model.
\newblock {\em J. Stat. Phys.}, 163(4):675--713, 2016.

\bibitem{Imamura-Sasamoto17}
Takashi Imamura and Tomohiro Sasamoto.
\newblock Fluctuations for stationary $q$-{TASEP}, 2017.
\newblock arXiv:1701.05991.

\bibitem{Johansson07}
K.~Johansson.
\newblock From {G}umbel to {T}racy-{W}idom.
\newblock {\em Probab. Theory Related Fields}, 138(1-2):75--112, 2007.

\bibitem{Johansson-Lambert15}
Kurt Johansson and Gaultier Lambert.
\newblock Gaussian and non-{G}aussian fluctuations for mesoscopic linear
  statistics in determinantal processes.
\newblock {\em Ann. Probab.}, 46(3):1201--1278, 2018.

\bibitem{Korhonen-Lee14}
Marko Korhonen and Eunghyun Lee.
\newblock The transition probability and the probability for the left-most
  particle's position of the {$q$}-totally asymmetric zero range process.
\newblock {\em J. Math. Phys.}, 55(1):013301, 15, 2014.

\bibitem{Le_Doussal-Majumdar-Rosso-Schehr16}
Pierre Le~Doussal, Satya~N. Majumdar, Alberto Rosso, and Gr\'egory Schehr.
\newblock Exact short-time height distribution in the one-dimensional
  {K}ardar-{P}arisi-{Z}hang equation and edge fermions at high temperature.
\newblock {\em Phys. Rev. Lett.}, 117:070403, Aug 2016.

\bibitem{Lee-Wang17}
Eunghyun Lee and Dong Wang.
\newblock Distributions of a particle's position and their asymptotics in the $
  q $-deformed totally asymmetric zero range process with site dependent
  jumping rates, 2017.
\newblock arXiv:1703.08839.

\bibitem{Liechty-Wang14-2}
Karl Liechty and Dong Wang.
\newblock Nonintersecting {B}rownian motions on the unit circle.
\newblock {\em Ann. Probab.}, 44(2):1134--1211, 2016.

\bibitem{Moshe-Neuberger-Shapiro94}
Moshe Moshe, Herbert Neuberger, and Boris Shapiro.
\newblock Generalized ensemble of random matrices.
\newblock {\em Phys. Rev. Lett.}, 73(11):1497--1500, 1994.

\bibitem{Olver97}
Frank W.~J. Olver.
\newblock {\em Asymptotics and special functions}.
\newblock AKP Classics. A K Peters, Ltd., Wellesley, MA, 1997.
\newblock Reprint of the 1974 original [Academic Press, New York; MR0435697 (55
  \#8655)].

\bibitem{Boisvert-Clark-Lozier-Olver10}
Frank W.~J. Olver, Daniel~W. Lozier, Ronald~F. Boisvert, and Charles~W. Clark,
  editors.
\newblock {\em N{IST} handbook of mathematical functions}.
\newblock U.S. Department of Commerce, National Institute of Standards and
  Technology, Washington, DC; Cambridge University Press, Cambridge, 2010.
\newblock With 1 CD-ROM (Windows, Macintosh and UNIX).

\bibitem{Beale-Pathria11}
R.~K. Pathria and Paul~D. Beale.
\newblock {\em Statistical mechanics}.
\newblock Elsevier/Academic Press, Amsterdam, third edition, 2011.

\bibitem{Quastel12}
Jeremy Quastel.
\newblock Introduction to {KPZ}.
\newblock In {\em Current developments in mathematics, 2011}, pages 125--194.
  Int. Press, Somerville, MA, 2012.

\bibitem{Sasamoto-Wadati98}
Tomohiro Sasamoto and Miki Wadati.
\newblock Exact results for one-dimensional totally asymmetric diffusion
  models.
\newblock {\em J. Phys. A}, 31(28):6057--6071, 1998.

\bibitem{Simon05}
Barry Simon.
\newblock {\em Trace ideals and their applications}, volume 120 of {\em
  Mathematical Surveys and Monographs}.
\newblock American Mathematical Society, Providence, RI, second edition, 2005.

\bibitem{Szego75}
G{\'a}bor Szeg{\H{o}}.
\newblock {\em Orthogonal polynomials}.
\newblock American Mathematical Society, Providence, R.I., fourth edition,
  1975.
\newblock American Mathematical Society, Colloquium Publications, Vol. XXIII.

\bibitem{Tracy-Widom08a}
Craig~A. Tracy and Harold Widom.
\newblock A {F}redholm determinant representation in {ASEP}.
\newblock {\em J. Stat. Phys.}, 132(2):291--300, 2008.

\bibitem{Tracy-Widom09a}
Craig~A. Tracy and Harold Widom.
\newblock On {ASEP} with step {B}ernoulli initial condition.
\newblock {\em J. Stat. Phys.}, 137(5-6):825--838, 2009.

\bibitem{Wang-Waugh16}
Dong Wang and David Waugh.
\newblock The transition probability of the {$q$}-{TAZRP} ({$q$}-bosons) with
  inhomogeneous jump rates.
\newblock {\em SIGMA Symmetry Integrability Geom. Methods Appl.}, 12:Paper No.
  037, 16, 2016.

\bibitem{Wood92}
David Wood.
\newblock The computation of polylogarithms.
\newblock Technical Report 15-92*, University of Kent, Computing Laboratory,
  University of Kent, Canterbury, UK, June 1992.

\end{thebibliography}

\end{document}